\documentclass[aps,twocolumn,tightenlines,superscriptaddress]{revtex4-1}

\let\coloneqq\relax

\usepackage[latin1]{inputenc}
\usepackage{amsthm}
\usepackage{amssymb}
\usepackage{amsmath}
\usepackage{bbold}
\usepackage{bbm}
\usepackage{nonfloat}
\usepackage[pdftex]{hyperref}
\usepackage{braket}
\usepackage{dsfont}
\usepackage{mathdots}
\usepackage{mathtools}
\usepackage{enumerate}
\usepackage[shortlabels]{enumitem}
\usepackage{csquotes}
\usepackage{stmaryrd}
\usepackage[cal=boondox]{mathalfa}
\usepackage{graphicx}
\usepackage{stackengine}
\usepackage{scalerel}
\usepackage{xr}
\usepackage{array}
\usepackage{makecell}
\newcolumntype{x}[1]{>{\centering\arraybackslash}p{#1}}
\usepackage{tikz}
\usepackage{pgfplots}
\usetikzlibrary{shapes.geometric, shapes.misc, positioning, arrows, arrows.meta, decorations.pathreplacing, decorations.pathmorphing, patterns, angles, quotes, calc}
\usepackage{booktabs}
\usepackage{xfrac}
\usepackage{siunitx}
\usepackage{centernot}
\usepackage{comment}
\usepackage{chngcntr}

\newtheorem{thm}{Theorem}
\newtheorem*{thm*}{Theorem}

\newtheorem*{prop*}{Proposition}
\newtheorem{lemma}[thm]{Lemma}
\newtheorem*{lemma*}{Lemma}

\newtheorem*{cor*}{Corollary}

\newtheorem*{cj*}{Conjecture}
\newtheorem{Def}[thm]{Definition}
\newtheorem*{Def*}{Definition}

\newtheorem{remark}{Remark}

\makeatletter
\def\thmhead@plain#1#2#3{%
  \thmname{#1}\thmnumber{\@ifnotempty{#1}{ }\@upn{#2}}%
  \thmnote{ {\the\thm@notefont#3}}}
\let\thmhead\thmhead@plain
\makeatother

\theoremstyle{definition}

\newcommand{\bb}{\begin{equation}\begin{aligned}\hspace{0pt}}
\newcommand{\bbb}{\begin{equation*}\begin{aligned}}
\newcommand{\ee}{\end{aligned}\end{equation}}
\newcommand{\eee}{\end{aligned}\end{equation*}}
\newcommand*{\coloneqq}{\mathrel{\vcenter{\baselineskip0.5ex \lineskiplimit0pt \hbox{\scriptsize.}\hbox{\scriptsize.}}} =}

\newcommand{\ketbra}[1]{\ket{#1}\!\!\bra{#1}}
\newcommand{\ketbraa}[2]{\ket{#1}\!\!\bra{#2}}

\newcommand{\N}{\mathds{N}}

\DeclareMathOperator{\Tr}{Tr}

\DeclareMathAlphabet{\pazocal}{OMS}{zplm}{m}{n}

\DeclareMathOperator{\Id}{Id}

\newcommand{\HH}{\pazocal{H}}

\newcommand{\lsmatrix}{\left(\begin{smallmatrix}}
\newcommand{\rsmatrix}{\end{smallmatrix}\right)}

\stackMath

\stackMath

\makeatletter
\newcommand*\rel@kern[1]{\kern#1\dimexpr\macc@kerna}
\newcommand*\widebar[1]{%
  \begingroup
  \def\mathaccent##1##2{%
    \rel@kern{0.8}%
    \overline{\rel@kern{-0.8}\macc@nucleus\rel@kern{0.2}}%
    \rel@kern{-0.2}%
  }%
  \macc@depth\@ne
  \let\math@bgroup\@empty \let\math@egroup\macc@set@skewchar
  \mathsurround\z@ \frozen@everymath{\mathgroup\macc@group\relax}%
  \macc@set@skewchar\relax
  \let\mathaccentV\macc@nested@a
  \macc@nested@a\relax111{#1}%
  \endgroup
}

\counterwithin*{equation}{part}
\counterwithin*{thm}{part}
\counterwithin*{figure}{part}

\tikzset{meter/.append style={draw, inner sep=10, rectangle, font=\vphantom{A}, minimum width=30, line width=.8, path picture={\draw[black] ([shift={(.1,.3)}]path picture bounding box.south west) to[bend left=50] ([shift={(-.1,.3)}]path picture bounding box.south east);\draw[black,-latex] ([shift={(0,.1)}]path picture bounding box.south) -- ([shift={(.3,-.1)}]path picture bounding box.north);}}}
\tikzset{roundnode/.append style={circle, draw=black, fill=gray!20, thick, minimum size=10mm}}
\tikzset{squarenode/.style={rectangle, draw=black, fill=none, thick, minimum size=10mm}}

\definecolor{Blues5seq1}{RGB}{239,243,255}
\definecolor{Blues5seq2}{RGB}{189,215,231}
\definecolor{Blues5seq3}{RGB}{107,174,214}
\definecolor{Blues5seq4}{RGB}{49,130,189}
\definecolor{Blues5seq5}{RGB}{8,81,156}

\definecolor{Greens5seq1}{RGB}{237,248,233}
\definecolor{Greens5seq2}{RGB}{186,228,179}
\definecolor{Greens5seq3}{RGB}{116,196,118}
\definecolor{Greens5seq4}{RGB}{49,163,84}
\definecolor{Greens5seq5}{RGB}{0,109,44}

\definecolor{Reds5seq1}{RGB}{254,229,217}
\definecolor{Reds5seq2}{RGB}{252,174,145}
\definecolor{Reds5seq3}{RGB}{251,106,74}
\definecolor{Reds5seq4}{RGB}{222,45,38}
\definecolor{Reds5seq5}{RGB}{165,15,21}

\usepackage{pgfplots}

\newtheorem{definition}{Definition}

\allowdisplaybreaks

\DeclareMathOperator{\arccosh}{arccosh}

\pgfplotsset{width=10cm,compat=1.9}
\usetikzlibrary{decorations.markings}

\begin{document}

\title{Maximum tolerable excess noise in CV-QKD and improved lower bound on two-way capacities}

\author{Francesco Anna Mele}
\email{francesco.mele@sns.it}
\affiliation{NEST, Scuola Normale Superiore and Istituto Nanoscienze, Consiglio Nazionale delle Ricerche, Piazza dei Cavalieri 7, IT-56126 Pisa, Italy}

\author{Ludovico Lami}
\email{ludovico.lami@gmail.com}
\affiliation{QuSoft, Science Park 123, 1098 XG Amsterdam, the Netherlands}
\affiliation{Korteweg-de Vries Institute for Mathematics, University of Amsterdam, Science Park 105-107, 1098 XG Amsterdam, the Netherlands}
\affiliation{Institute for Theoretical Physics, University of Amsterdam, Science Park 904, 1098 XH Amsterdam, the Netherlands}
\affiliation{Institut f\"{u}r Theoretische Physik und IQST, Universit\"{a}t Ulm, Albert-Einstein-Allee 11, D-89069 Ulm, Germany}

\author{Vittorio Giovannetti}
\email{vittorio.giovannetti@sns.it}
\affiliation{NEST, Scuola Normale Superiore and Istituto Nanoscienze, Consiglio Nazionale delle Ricerche, Piazza dei Cavalieri 7, IT-56126 Pisa, Italy}

\maketitle
\textbf{The two-way capacities of quantum channels determine the ultimate entanglement and secret-key distribution rates achievable by two distant parties that are connected by a noisy transmission line, in absence of quantum repeaters. Since repeaters will likely be expensive to build and maintain, a central open problem of quantum communication is to understand what performances are achievable without them. In this paper, we find a new lower bound on the energy-constrained and unconstrained two-way quantum and secret-key capacities of all phase-insensitive bosonic Gaussian channels, namely thermal attenuator, thermal amplifier, and additive Gaussian noise, which are realistic models for the noise affecting optical fibres or free-space links. Ours is the first nonzero lower bound on the two-way quantum capacity in the parameter range where the (reverse) coherent information becomes negative, and it shows explicitly that entanglement distribution is always possible when the channel is not entanglement breaking. This completely solves a crucial open problem of the field, namely, establishing the maximum excess noise which is tolerable in continuous-variable quantum key distribution. In addition, our construction is fully explicit, i.e.~we devise and optimise a concrete entanglement distribution and distillation protocol that works by combining recurrence and hashing protocols}.

\bigskip

Quantum key distribution (QKD) stands as the gold standard for unconditionally secure communication. Since its inception in 1984~\cite{bennett1984quantum}, QKD has transitioned from a theoretical concept to a commercially viable technology. Continuous-variable systems~\cite{BUCCO}, such as electromagnetic modes, offer 
a powerful approach to QKD, known as CV-QKD~\cite{CV_qkd,Pirandola20,Laudenbach_2018}. Unlike its discrete-variable counterpart~\cite{Pirandola20}, CV-QKD is expected to be seamlessly integrated into existing optical fibre-based Internet infrastructure in the near future~\cite{Pirandola20,Record1,Record2,Record3,Record4,Record5}, making it highly practical for real-world applications. 
A significant challenge in this field, essential for the development of a large-scale quantum internet~\cite{quantum_internet_Wehner,Pirandola20}, is extending CV-QKD over long distances without intermediate nodes~\cite{Pirandola20}. In recent years, numerous experiments have been setting new distance records, achieving CV-QKD across optical fibres exceeding one hundred kilometers in length~\cite{Record1,Record2,Record3,Record4,Record5}. This progress prompts a critical question: ``What is the maximum achievable distance of CV-QKD according to the laws of quantum physics?'' Equivalently, ``What is the maximum tolerable excess noise in CV-QKD?'' This has been identified as a crucial open problem in the field~\cite[Section~7]{Pirandola18}.

In quantum Shannon theory~\cite{MARK,Sumeet_book}, the capacities of quantum channels determine the ultimate limits of quantum communication which are achievable across the channel. Specifically, the \emph{secret-key capacity} $K(\Phi)$ and the \emph{two-way quantum capacity} $Q_2(\Phi)$ of a quantum channel $\Phi$, collectively called the \emph{two-way capacities}, are the maximum rates of secret-key bits and entanglement bits (or \emph{ebits}), respectively, achievable across $\Phi$ when the sender (Alice) and the receiver (Bob) are assisted by two-way classical communication~\cite{MARK,Sumeet_book}. These capacities quantify the optimal performance for QKD and entanglement distribution. Importantly, QKD across $\Phi$ is achievable if and only if $K(\Phi)>0$, while entanglement distribution across $\Phi$ is achievable if and only if $Q_2(\Phi)>0$. Since in practice only a finite amount of energy can be utilised in a communication protocol, it is common to consider the \emph{energy-constrained} two-way capacities~\cite{Davis2018}. These are denoted as $K(\Phi, N_s)$ and $Q_2(\Phi, N_s)$, where $N_s$ represents the maximum average photon number per input signal to the channel $\Phi$. Since an ebit can be converted into a secret-key bit~\cite{Ekert91}, these capacities satisfy $ K(\Phi,N_s)\ge Q_2(\Phi,N_s)$. 

Within the framework of quantum optical communication, a realistic model to describe optical fibres is the \emph{thermal attenuator} $\mathcal{E}_{\lambda,\nu}$~\cite{BUCCO,PLOB,Pirandola18}. It is a single-mode quantum channel characterised by two parameters: the transmissivity $\lambda\in [0,1]$ of the fibre and the added thermal noise $\nu\in [0,\infty)$~\cite{BUCCO}. 
We provide a detailed definition of the thermal attenuator in the Methods. Computing the two-way capacities of the thermal attenuator is essential to determine the ultimate performances of QKD or entanglement distribution protocols that can be achieved on optical networks without the use of intermediate nodes or quantum repeaters~\cite{repeaters, Munro2015,Pirandola20}. The two-way capacities of the thermal attenuator have been determined for all $\lambda$ only when $\nu=0$ and there is no energy constraint~\cite{PLOB}, assumptions that may not be entirely physically realistic, depending on the setting. Except for this very special case, the two-way capacities of the thermal attenuator are still unknown, despite the many upper~\cite{PLOB, Davis2018, Goodenough16, TGW, holwer, MMMM, squashed_channel} and lower bounds~\cite{holwer, Pirandola2009, Noh2020, Ottaviani_new_lower, Pirandola18} that have been established. In particular, in a large parameter region all lower bounds (prior to our work) vanish, while the upper bounds do not. This entails that the precise noise threshold above which entanglement distribution or key distribution become impossible had not been determined in the prior literature. This leads to the mathematical formulation of the aforementioned open problem~\cite[Section~7]{Pirandola18}: \emph{for a given $\nu$, determine the minimum transmissivity $\lambda$ for which the secret-key capacity of the thermal attenuator is strictly positive, i.e.~for which CV-QKD is achievable.}

\section*{Results}
\subsection{Maximum tolerable excess noise in CV-QKD}\label{subsec_b} 

In this section, we solve the above problem~\cite[Section~7]{Pirandola18}, establishing the following simple formula for the minimum admissible transmissivity $\lambda_{\mathrm{min}}$ in CV-QKD:
\bb\label{lamba_min}
    \lambda_{\mathrm{min}}=\frac{\nu}{\nu+1}\,. 
\ee
Note that the transmissivity $\lambda$ of an optical fibre decreases exponentially with its length $L$ as~\cite{Tamura2018, Li2020}
\begin{equation}
    \lambda=10^{-\gamma \frac{L}{\SI{10}{\km}}}\,,
    \label{lambda_decay}
\end{equation}
where the attenuation coefficient $\gamma$ typically satisfies $\gamma\simeq 0.2$, with the lowest recorded value being $\gamma\simeq 0.14$~\cite{Tamura2018, Li2020}. We can thus determine the maximum achievable distance $L_{\mathrm{max}}$ in CV-QKD:
\bb
    L_{\mathrm{max}}=\frac{10\,\mathrm{km}}{\gamma}\log_{10}\!\left(1+\nu^{-1}\right)\,.
\ee
The parameter $\nu$ is related to the wavelength $\lambda_{\mathrm{wave}}$ employed for communication through the Bose--Einstein distribution:
\bb
    \nu=\frac{1}{\exp\!\left({\frac{hc}{\lambda_{\mathrm{wave}} k_B T}}\right)-1}\,,
\ee
where $T$ is the room temperature, $h$ is the Planck constant, $c$ is the speed of light, and $k_B$ is the Boltzmann constant. Hence, the maximum achievable distance in CV-QKD can be expressed as a function of $\gamma$, $T$, and $\lambda_{\mathrm{wave}}$ as follows:
\bb\label{L_max}
    L_{\mathrm{max}}&= \frac{10\log_{10}e}{\gamma}\left(\frac{hc}{\lambda_{\mathrm{wave}}k_BT}\right)\, \\
    & \simeq  \left(\frac{298\, \mathrm{K}}{T}\right) \left(\frac{1.5\,\mathrm{\mu m}}{\lambda_{\mathrm{wave}}}\right)\frac{\SI{139}{km}}{\gamma}\,.
\ee
This serves as the ultimate benchmark for CV-QKD: if the fibre length exceeds $L_{\mathrm{max}}$, quantum mechanics rules out the possibility of achieving CV-QKD without intermediate nodes; conversely, if the length is smaller than $L_{\mathrm{max}}$, CV-QKD without intermediate nodes is achievable.

To illustrate the significance of \eqref{L_max}, 
consider its application to the current Internet infrastructure. Set the attenuation coefficient $\gamma$ to the lowest recorded value of $\gamma \simeq 0.14$, the temperature $T$ to the standard room temperature of $T = 298\,\mathrm{K}$, and the wavelength $\lambda_{\mathrm{wave}}$ to the telecom wavelength of $\lambda_{\mathrm{wave}} \simeq 1.5\,\mathrm{\mu m}$, which is used in the current Internet infrastructure as well as in state-of-the-art CV-QKD experiments~\cite{Record1,Record2,Record3,Record4,Record5}. By using \eqref{L_max}, the maximum achievable distance is estimated as
\bb
    L_{\mathrm{max}}\simeq \SI{990}{km}\,.
\ee
This result is informative: it indicates that achieving CV-QKD over optical fibres longer than approximately $\SI{990}{km}$ necessitates the use of intermediate (trusted, and thus more costly) nodes. Conversely, it also establishes that there exists a protocol enabling point-to-point CV-QKD over distances up to approximately $\SI{990}{km}$. Strikingly, this estimated maximum distance of approximately $\SI{990}{km}$ is not far from the distances achieved by state-of-the-art CV-QKD experiments, which are around $200\,\mathrm{km}$~\cite{Record1,Record2,Record3,Record4,Record5}. This ultimately establishes that to go much beyond the current distance records, the use of intermediate nodes is essential.

The above result can be generalised to encompass all Gaussian channels~\cite{BUCCO}, including, besides the thermal attenuator $\mathcal{E}_{\lambda,\nu}$, also the \emph{thermal amplifier} $\Phi_{g,\nu}$, the \emph{additive Gaussian noise} $\Lambda_\xi$, and many others. The channels $\mathcal{E}_{\lambda,\nu}$, $\Phi_{g,\nu}$, and $\Lambda_\xi$ are collectively called ``phase-insensitive bosonic Gaussian channels'' (piBGCs). The thermal amplifier $\Phi_{g,\nu}$ is characterised by two parameters, the gain $g\in [0,\infty)$ and the added thermal noise $\nu$. The additive Gaussian noise $\Lambda_\xi$ is characterised by a single parameter $\xi\in[0,\infty)$, representing the added \emph{classical} noise~\cite{BUCCO}. Detailed definitions of these channels are provided in the Methods. In the forthcoming Theorem~\ref{th1_main} we characterise the parameter region where the two-way capacities of the piBGCs are strictly positive, i.e.~where they can be used for entanglement distribution and QKD. Specifically, we show that this is possible if and only if the channel is not \emph{entanglement breaking}~\cite{Sumeet_book}. This provides a simple and concise solution to the problem posed in~\cite[Section~7]{Pirandola18}; remarkably, our proof is based solely on Gaussian quantum information techniques~\cite{BUCCO}.
\begin{thm}\label{th1_main}
    Let $\Phi$ be a single-mode Gaussian channel and let $N_s>0$. The energy-constrained two-way quantum capacity $Q_2(\Phi,N_s)$ and secret-key capacity $K(\Phi,N_s)$ are strictly positive if and only if $\Phi$ is not entanglement breaking.

    In particular, for all $\nu\ge0$, $\lambda\in[0,1]$, $g\ge 1$, and $\xi\ge0$, the energy-constrained two-way quantum capacity of the thermal attenuator $\mathcal{E}_{\lambda,\nu}$, thermal amplifier $\Phi_{g,\nu}$, and additive Gaussian noise $\Lambda_{\xi}$ satisfy:
    \bb\label{param_reg}
    Q_2(\mathcal{E}_{\lambda,\nu},N_s)>0 &\text{ $\Leftrightarrow$ } \lambda\in\left(\frac{\nu}{\nu+1},1\right]\,,\\
    Q_2(\Phi_{g,\nu},N_s)>0 &\text{ $\Leftrightarrow$ } g\in\left[1,\frac{\nu+1}{\nu}\right)\,,\\
    Q_2(\Lambda_{\xi},N_s)>0 &\text{ $\Leftrightarrow$ } \xi\in[0,1)\,.
    \ee
    The same holds for the secret-key capacity $K(\cdot, N_s)$ as well as for the unconstrained capacities. 
\end{thm}
\begin{proof}[Proof] 
    Since any entanglement-breaking channel has vanishing two-way capacities~\cite{Sumeet_book}, it suffices to consider the case where $\Phi$ is not entanglement breaking. Assume that Alice prepares many copies of the two-mode squeezed vacuum state $\ket{\Psi_{N_s}}_{AA'}$ with mean local photon number $N_s$, i.e.
    \bb\label{TMSV_def_main}
     \ket{\Psi_{N_s}}\coloneqq \frac{1}{\sqrt{N_s+1}}\sum_{n=0}^\infty \left(\frac{N_s}{N_s+1}\right)^{n/2}\ket{n}_{A}\otimes \ket{n}_{A'}\,,
    \ee 
    and sends the systems $A'$ through the Gaussian channel $\Phi_{A'\to B}$. Now Alice and Bob share many copies of the two-mode Gaussian state
    \bb\label{def_C_N_s}
        C_{N_s}\coloneqq \big(\Id_{A}\otimes\,\Phi_{A'\to B}\big)(\ketbra{\Psi_{N_s}}_{AA'})\,,
    \ee
    which is a \emph{generalised Choi state} of $\Phi$. By definition, a generalised Choi state of a quantum channel $\mathcal{N}_{A'\to B}$ is a bipartite state of the form $\big(\Id_{A}\otimes\,\mathcal{N}_{A'\to B}\big)(\ketbra{\Psi}_{AA'})$, where the input state $\ket{\Psi}_{AA'}$ is such that its reduced state $\Tr_{A}[\ketbra{\Psi}_{AA'}]$ is invertible. In Lemma~S9 of the Supplementary Information we prove that all the generalised Choi states of a non-entanglement-breaking channel are entangled. In particular, we deduce that $C_{N_s}$ is entangled. Additionally, by exploiting the fact that a two-mode Gaussian state is entangled if and only if it is not \emph{PPT}~\cite{PeresPPT,Simon00, BUCCO}, it thus follows that $C_{N_s}$ is not PPT. Finally, since any two-mode Gaussian state that it is not PPT is also \emph{distillable}~\cite{Giedke01} --- i.e.~it can be converted into ebits with a strictly positive rate --- we conclude that $K(\Phi,N_s)\ge Q_2(\Phi,N_s)>0$.

If $\Phi$ is a piBGC, an alternative, more explicit proof of the fact that $C_{N_s}$ is entangled uses the following entanglement criterion~\cite{Simon00, BUCCO}: a two-mode Gaussian state is entangled if and only if its covariance matrix $V\coloneqq \left(\begin{matrix} V_{A} & V_{AB} \\ V_{AB}^{\intercal} & V_{B}\end{matrix}\right)$ satisfies the condition
\bb\label{condition_ENT}
    1+\det V+2\det V_{AB}< \det V_{A}+\det V_{B}\,.
\ee
In Theorem~S12 of the Supplementary Information we show that the covariance matrix of $C_{N_s}$ satisfies such a condition.

The parameter regions in \eqref{param_reg} are precisely those where $\mathcal{E}_{\lambda,\nu}$, $\Phi_{g,\nu}$, and $\Lambda_{\xi}$ are not entanglement breaking~\cite{PLOB,Ent_breaking_Gaussian, Holevo-EB}.

\end{proof}

We remark that Theorem~\ref{th1_main} establishes the maximum tolerable noise not only in key distribution but also in entanglement distribution. Thus, the condition $\lambda>\frac{\nu}{\nu+1}$ serves as a necessary and sufficient condition not only for achieving CV-QKD without relying on intermediate nodes but also for achieving entanglement distribution without the use of (possibly expensive) quantum repeaters~\cite{repeaters, Munro2015}.

As another example of application of Eq.~\eqref{lamba_min}, using~\cite[Eq.~(1)--(2)]{Pirandola2021} and the values in~\cite[Table~I]{Pirandola2021}, one can see that entanglement distribution between Earth and the Moon, with lenses of aperture $a_R = w_0 = \SI{5}{cm}$, are only possible at wavelengths below $\SI{0.17}{mm}$, due to the cosmic microwave background at $\SI{2.725}{K}$.

\subsection{Improved lower bound on two-way capacities}~\label{subsec_a}
While~\eqref{param_reg} establishes the parameter ranges for which the piBGCs have positive capacities, it does not give us any explicit estimate on those capacities. Here we do precisely that, finding a new lower bound on the two-way capacities of the piBGCs that constitutes a significant improvement upon the state-of-the-art lower bounds~\cite{Ottaviani_new_lower, Pirandola2009, Pirandola18, Noh2020, Wang_Q2_amplifier, holwer} in a large parameter region. Our result is fully constructive, as it is proved by designing and analysing a concrete entanglement distribution protocol.  

The best known lower bounds on the two-way capacities of the piBGCs, prior to our work, are the (reverse) \emph{coherent information} lower bounds~\cite{Pirandola2009,holwer}. These bounds are derived by evaluating the ebit rate of a two-step entanglement distribution protocol: first, Alice sends halves of the two-mode squeezed vacuum state through the channel; second, Alice and Bob apply the hashing protocol~\cite{devetak2005,reviewEDP_dur} to distil entanglement. Our protocol improves on both steps. It draws inspiration from techniques used for distilling entanglement from two-qubit Werner states~\cite{Bennett-error-correction,Bennett-distillation-mixed}, where the ebit rate can be increased by introducing a \emph{recurrence stage}~\cite{Bennett-error-correction,Bennett-distillation-mixed,reviewEDP_dur} before the hashing protocol. In essence, our protocol involves sending halves of a suitably encoded finite-dimensional entangled state into the channel; projecting the channel's output into an appropriate finite-dimensional subspace; applying recurrence protocols~\cite{Bennett-error-correction,Bennett-distillation-mixed,reviewEDP_dur,p1orp2}; and finally executing an improved version of the hashing protocol~\cite{Improvement-Hashing}. While some steps of our protocol are obtained by suitably combining existing finite-dimensional protocols, the conceptual and technical novelty of our approach is to find a way to effectively apply them to the continuous-variable setting at hand. 

Let us present our entanglement distribution protocol across a piBGC $\Phi$. It is composed of six steps named S1--S6, and it depends on three parameters over which we will optimise numerically: $M\in\N^+$, $c\in(0,1)$, $k\in\N$.  These parameters have the following intuitive interpretations: $M$ represents the maximum number of photons in the input state of the protocol, $c$ indicates the level of coherence in that state, and $k$ corresponds to the number of iterations of a specific subroutine within the protocol.

\bigskip 
\noindent \textbf{\emph{Entanglement distribution protocol}}: 
\vspace{0ex}
\begin{enumerate}[leftmargin=3.76ex]
\item[\textbf{S1}:] Alice prepares many copies of the state 
\bb\label{state_our_protocol}
    \hspace{3.5ex} \ket{\Psi_{M,c}}_{\!AA'}\coloneqq c\ket{e_0}_{\!A}\!\otimes\!\ket{0}_{\!A'}+\sqrt{1\!-\!c^2}\ket{e_1}_{\!A}\!\otimes\!\ket{M}_{\!A'}\,,
\ee
and sends the subsystem $A'$ to Bob through the channel $\Phi$. Here, $\ket{0}_{A'}$ and $\ket{M}_{A'}$ denote the vacuum and the $M$th Fock state, while $\ket{e_0}_{\!A}$ and $\ket{e_1}_{\!A}$ represent orthogonal states of Alice's register $A$, which may be any quantum system (optical or not). As detailed in Section III.C of the Supplementary Information, an experimental realisation of the state in \eqref{state_our_protocol} may involve \emph{NOON states}~\cite{Sanders1989}. 

Now Alice and Bob share many pairs of the state $\big(\Id_{A}\otimes\,\Phi\big)(\ketbra{\Psi_{M,c}})$. If there is an energy constraint $N_s$, the parameters $c$ and $M$ have to satisfy $(1-c^2)M\le N_s$.

\item[\textbf{S2}:] Bob performs the local POVM $\{\Pi_M, \mathds{1} -\Pi_M\}$ on each pair, where $\Pi_M\coloneqq \ketbra{0}+\ketbra{M}$. If Bob finds the outcome associated with $\Pi_M$, then the pair is kept, otherwise it is discarded. Hence, by re-mapping $\ket{0}\to \ket{e_0}$ and $\ket{M}\to \ket{e_1}$, each pair is now in an effective two-qubit state.

\item[\textbf{S3}:] Alice and Bob apply the Pauli-based twirling, reported in~\eqref{def_twirling_map} in the Methods, in order to transform each of the remaining pairs in a Bell-diagonal state.

\item[\textbf{S4}:] Alice and Bob run $k$ times the following sub-routine, 
dubbed \emph{P1-or-P2}~\cite{p1orp2}.  
 
\medskip
\noindent \textbf{\emph{P1-or-P2 sub-routine}}:  
     \begin{itemize}
         \item \textbf{Step 4.1}: At this point of the protocol, the two-qubit state $\rho_{AB}$ of each pair is of the form
         \bb\label{bell_diagonal_state_step4}
            \qquad\quad  \rho_{AB}=\sum_{i,j=0}^{1}p_{ij}\,\mathds{1}_A\otimes  X^j Z^i \ketbra{\psi_{00}}(\mathds{1}_A\otimes   X^j Z^i)^\dagger\,,
         \ee
        where $\ket{\psi_{00}}$ denotes the ebit state (see~\eqref{Bell_states_main} of the Methods), $X,Z$ denote the well-known Pauli operators, and $\{p_{ij}\}_{i,j\in\{0,1\}}$ is a probability distribution.
        
        Alice and Bob collect all pairs in groups of two.  
        For a given group, call $A_1B_1$ and $A_2B_2$ the four qubits involved.
        If $p_{10}<p_{01}$, they apply the CNOT gate on $A_1A_2$ and $B_1B_2$, respectively, where $A_1,B_1$ are the control qubits and $A_2,B_2$ are the target qubits.
        Otherwise, they apply first the Hadamard gate on each qubit, then the CNOT gate as in the above case, and finally the Hadamard gate on $A_1$ and $B_1$.
        
        \item \textbf{Step 4.2}: For each 
        group of two pairs, Alice and Bob perform a projective measurement in the computational basis of $A_2$ and $B_2$, thereby discarding these systems.
        If the outcomes are different, they discard also the pair $A_1B_1$.
     \end{itemize}

The sub-routine tends to increase the value of $p_{00}$, bringing the state $\rho_{AB}$ closer to the ebit state $\ket{\psi_{00}}$. The condition $p_{10}<p_{01}$ in Step~4.1 means that the $X$ error in~\eqref{bell_diagonal_state_step4} is more prominent than the $Z$ error. If that is the case, the sub-routine reduces the $X$ error at the expense of the $Z$ error. However, by selectively applying the two procedures both errors end up being corrected~\cite{p1orp2}.

\item[\textbf{S5}:] Depending on the state $\rho$ of the remaining pairs (see Section III.A of the Supplementary Information for details), both Alice and Bob apply on each qubit one of the following unitaries: the qubit rotation around the $y$-axis of an angle $\pi/2$, the Hadamard gate, or the identity.
 
\item[\textbf{S6}:] In the end, Alice and Bob run the improved version of the hashing protocol introduced in~\cite{Improvement-Hashing} in order to generate ebits.
\end{enumerate}

Our lower bound on the two-way quantum capacity $Q_2(\Phi)$, calculated in Theorem S14 of the Supplementary Information, is the supremum over $M\in\N^+$, $c\in(0,1)$, and $k\in\N$ of the ebit rate of the above protocol. Since the secret-key capacity $K$ is always larger than $Q_2$, the resulting expression is also a lower bound on $K(\Phi)$. Our lower bound on the energy-constrained two-way capacities with energy constraint $N_s$ is obtained by optimising with the additional condition $(1-c^2)M\le N_s$.

We plot our bounds on the two-way capacities of the thermal attenuator $\mathcal{E}_{\lambda,\nu}$ in Fig.~1(a), of the thermal amplifier $\Phi_{g,\nu}$ in Fig.~1(b), and of the additive Gaussian noise $\Lambda_\xi$ in Fig.~1(c). These plots demonstrate that our bound is strictly tighter than all known lower bounds on both the two-way quantum and the secret-key capacity~\cite{Ottaviani_new_lower,Pirandola2009,Pirandola18,Wang_Q2_amplifier,holwer} in a large parameter region. Additionally, as shown in Section V of the Supplementary Information, the energy-constrained version of our bound outperforms the tightest known lower bound on the energy-constrained two-way capacities~\cite{Noh2020} in a substantial parameter region. Indeed, our numerical analysis reveals that the optimal value of $M$ in~\eqref{state_our_protocol} is never greater than three. This implies that our protocol is  highly energy efficient, using only states with up to three photons to distribute entanglement.

Notably, unlike all known lower bounds~\cite{Ottaviani_new_lower,Pirandola2009,Pirandola18,Wang_Q2_amplifier,holwer,Noh2020}, which vanish in a large parameter region, we observe numerically that our lower bound is \emph{faithful}: it remains strictly positive if and only if the tightest known upper bound~\cite{PLOB} is also strictly positive. Note that the proof of our Theorem~\ref{th1_main} provides an alternative entanglement distribution protocol which is mathematically guaranteed to be faithful. However, its ebit rate is much lower than that of the protocol presented in this section.

In Section IV of the Supplementary Information, we introduce an additional entanglement distribution protocol that combines and optimises the multi-rail protocol from~\cite{Winnel} and the qudit P1-or-P2 protocol from~\cite{p1orp2}. The ebit rate of this protocol is calculated in Theorem~S14 of the Supplementary Information and constitutes an additional lower bound on the two-way capacities of the piBGCs. Importantly, this additional lower bound turns out to be tighter than our previously discussed lower bound on $Q_2(\mathcal{E}_{\lambda,\nu})$ in the low excess noise regime $\nu\lesssim1$, as shown in Fig.~6 of the Supplementary Information. Note that in Fig.~1 we considered a higher excess noise regime where $\nu =10$.


Finally, let us briefly discuss the experimental feasibility of the entanglement distribution protocol introduced in this section. This protocol encounters the same experimental challenges as all known entanglement distillation protocols, as it primarily consists of the fundamental building blocks of any entanglement distillation protocol, namely recurrence and hashing protocols~\cite{Bennett-error-correction,Bennett-distillation-mixed,reviewEDP_dur,p1orp2,Improvement-Hashing}. Significant experimental progress has been made in recent years regarding entanglement distillation~\cite{Kalb_2017,Hu_2021,Ecker_2021}, leaving hope that these experimental challenges will be mitigated in the near future. See Section III.C of the Supplementary Information for a detailed discussion of the experimental feasibility of our protocol.

    \begin{figure}
\begin{tabular}{c}
\includegraphics[width=1.0\linewidth]{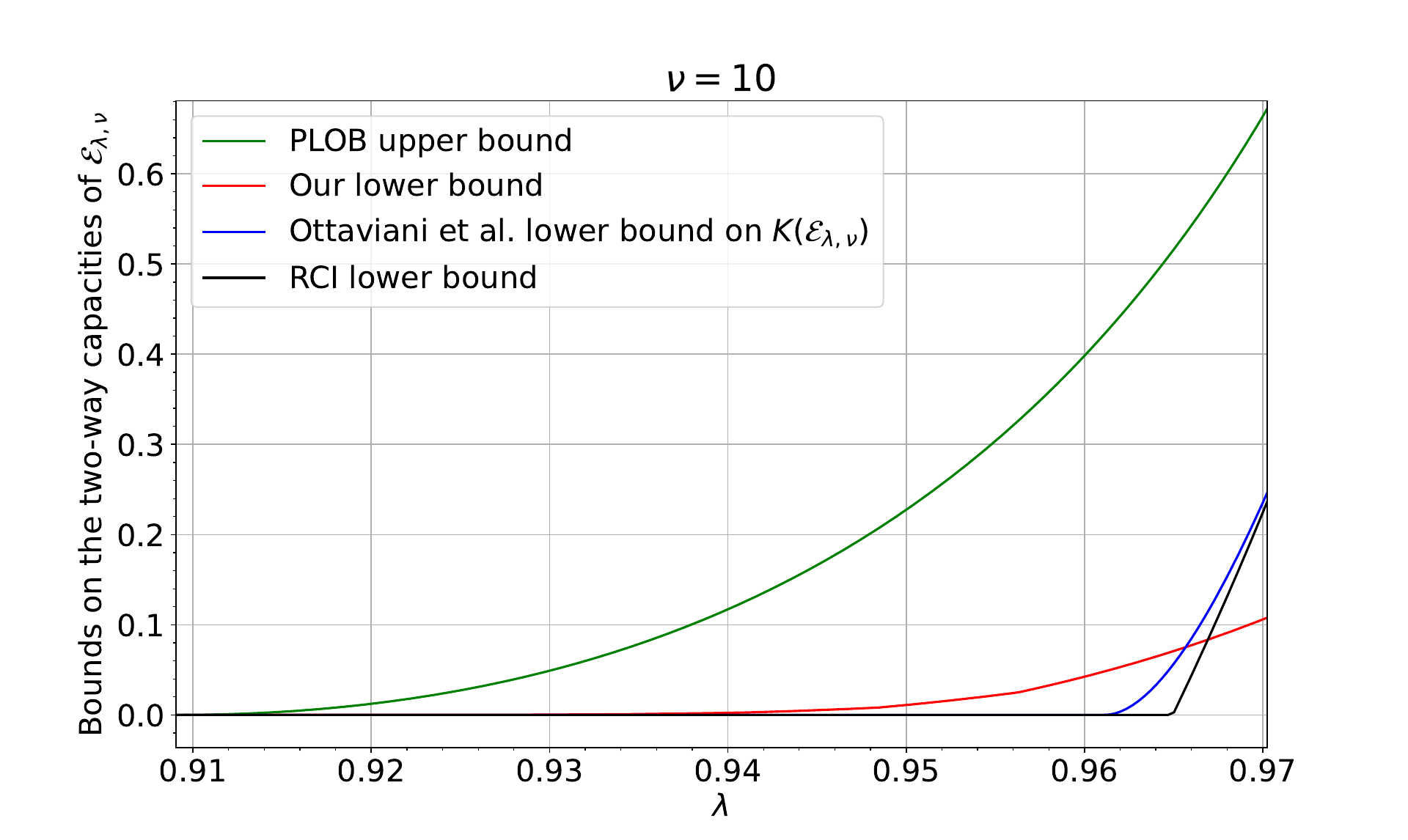} \\  
(a)\\
 \includegraphics[width=1.0\linewidth]{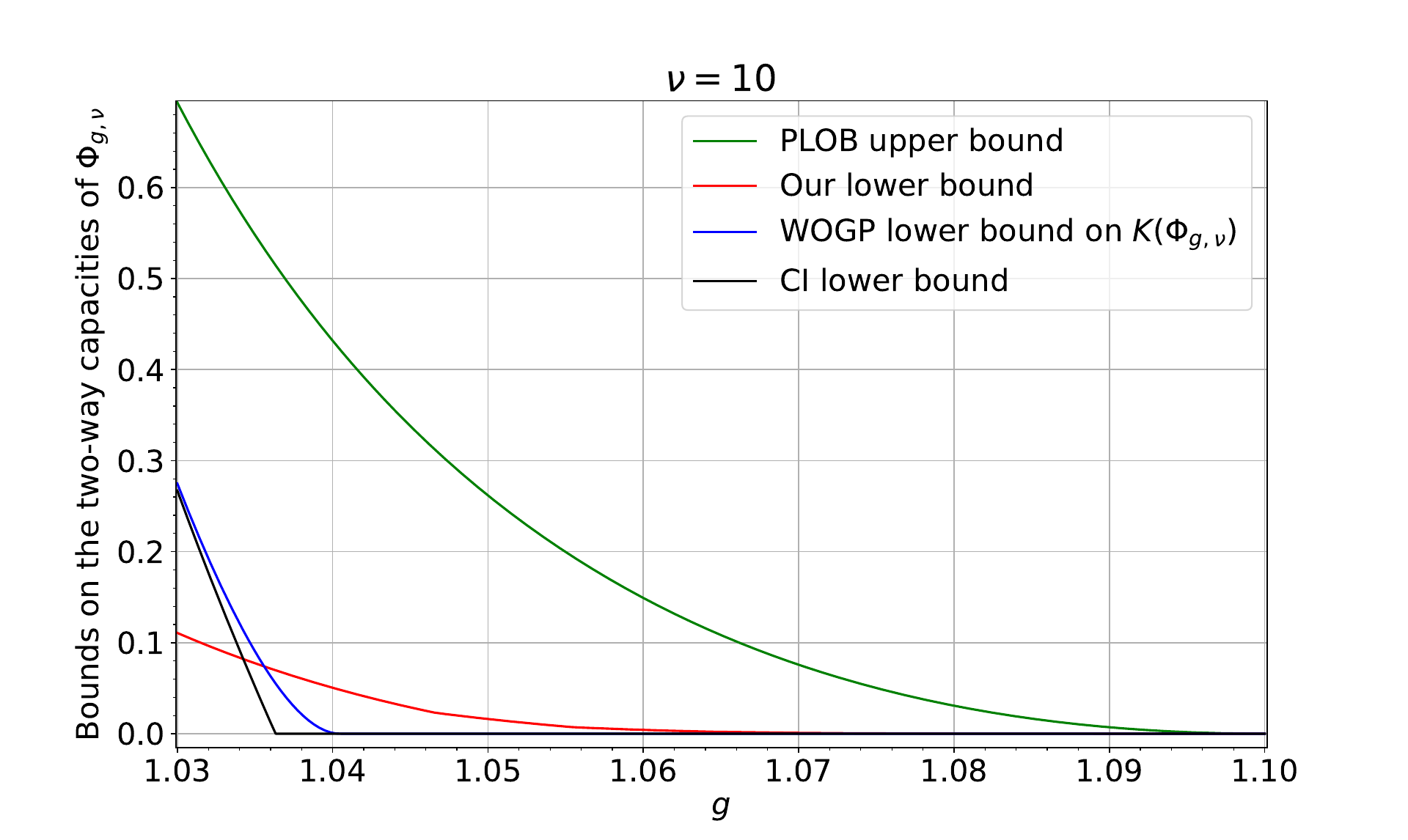} \\ 
(b) \\
 \includegraphics[width=1.0\linewidth]{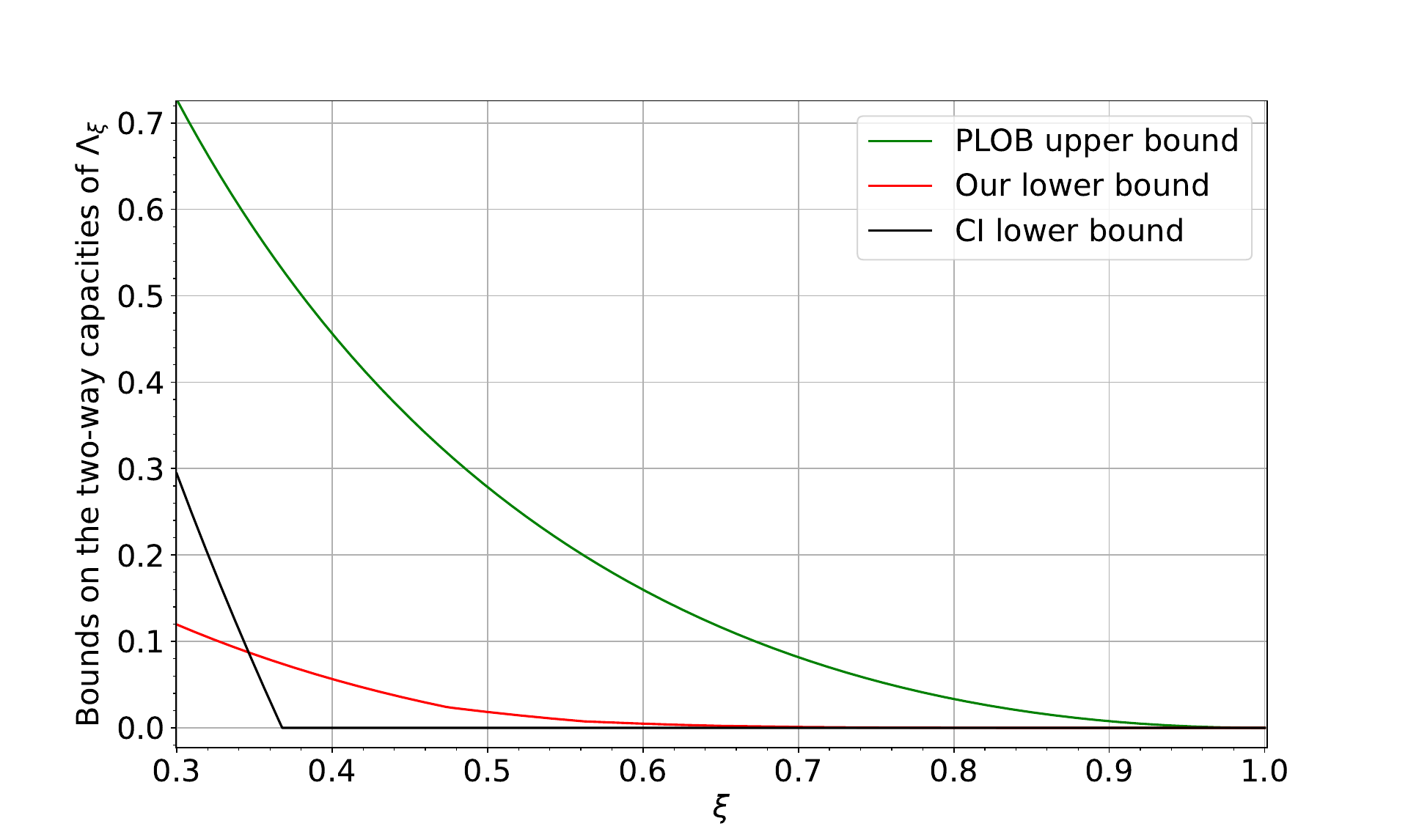} \\ 
(c)
\end{tabular}
\caption{Bounds on the two-way capacities of the thermal attenuator, thermal amplifier, and additive Gaussian noise.  \textbf{(a).}~Bounds on the two-way quantum capacity $Q_2$ and secret-key capacity $K$ of the thermal attenuator $\mathcal{E}_{\lambda,\nu}$ plotted with respect to $\lambda$ for $\nu=10$. The red line is our lower bound; the black line is the reverse coherent information lower bound reported in~\eqref{lowQ2_main_att}; the blue line is the lower bound on $K(\mathcal{E}_{\lambda,\nu})$ discovered by Ottaviani et al.~\cite{Ottaviani_new_lower} (this is not a lower bound on $Q_2$, since in general it only holds that $Q_2\le K$); and the green line is the upper bound discovered by Pirandola et al.~\cite{PLOB} reported in~\eqref{PLOB_Q2_main}. \textbf{ (b).}~Bounds on the two-way quantum capacity $Q_2$ and secret-key capacity $K$ of the thermal amplifier $\Phi_{g,\nu}$ plotted with respect to $g$ for $\nu=10$. The red line is our lower bound; the black line is the coherent information lower bound reported in~\eqref{lowQ2_main_amp}; the blue line is the lower bound on $K(\Phi_{g,\nu})$ discovered by Wang et al.~\cite{Wang_Q2_amplifier}; and the green line is the upper bound discovered by Pirandola et al.~\cite{PLOB} reported in~\eqref{PLOB_Q2_main}.  \textbf{ (c).}~Bounds on the two-way quantum capacity $Q_2$ and secret-key capacity $K$ of the additive Gaussian noise $\Lambda_\xi$ plotted with respect to $\xi$. The red line is our lower bound; the black line is the coherent information lower bound reported in~\eqref{lowQ2_main_noise}; and the green line is the upper bound discovered by Pirandola et al.~\cite{PLOB} reported in~\eqref{PLOB_Q2_main}. }
\label{fig_main}
\end{figure}

\section*{Discussion}
In this work, we have determined a simple necessary and sufficient condition for achieving continuous-variable quantum-key distribution without intermediate nodes~\cite{Pirandola20} and entanglement distribution without quantum repeaters~\cite{repeaters, Munro2015}: point-to-point key distribution and entanglement distribution across an optical link of transmissivity $\lambda$ and added thermal noise $\nu$ are achievable if and only if $\lambda>\frac{\nu}{\nu+1}$. This solves completely the problem of determining the ultimate limitations to CV-QKD imposed by the laws of quantum physics~\cite{Pirandola18}. By leveraging this result, we have established the maximum achievable distance of CV-QKD, demonstrating that the current Internet infrastructure can in principle support CV-QKD only up to distances of approximately $\SI{990}{km}$, which is not so far from the distance records of approximately $200\,\mathrm{km}$ achieved by state-of-the-art 
experiments~\cite{Record1, Record2, Record3, Record4, Record5}. Consequently, surpassing these distance records by an order of magnitude at standard telecom wavelengths will necessarily require the use of intermediate nodes.

Moreover, we have derived the best lower bound to date on the two-way quantum and secret-key capacities of all phase-insensitive bosonic Gaussian channels in the regime of large noise --- in a large parameter region, ours is in fact the \emph{only} non-zero lower bound. Our new bound constitutes a significant improvement upon state-of-the-art lower bounds~\cite{Ottaviani_new_lower,Pirandola2009,Pirandola18,Wang_Q2_amplifier,Noh2020}. We remark that the last improvement on the (unconstrained) two-way quantum capacity prior to our work dates back to 2009~\cite{Pirandola2009}. It was an open question whether the latter could equal the true two-way quantum capacities: our work provides a negative answer to this question, showing that entanglement distribution is possible in a much broader parameter region than previously known. Our results are fully explicit: to prove our lower bound, we constructed a general yet relatively simple entanglement distribution (and hence also key distribution) protocol that works whenever the underlying channel is not entanglement breaking.

In conclusion, we have discovered new protocols and fundamental limitations to quantum communication across optical links, establishing the ultimate noise threshold at which entanglement and secret keys can be distributed. This is likely to bear a significant impact on the design of practical QKD and entanglement distribution protocols on optical networks.

\medskip
\noindent \textbf{Acknowledgements} --- 
FAM and VG acknowledge financial support by MUR (Ministero dell'Istruzione, dell'Universit\`a e della Ricerca) through the following projects: PNRR MUR project PE0000023-NQSTI, PRIN 2017 Taming complexity via Quantum Strategies: a Hybrid Integrated Photonic approach (QUSHIP) Id. 2017SRN-BRK, and project PRO3 Quantum Pathfinder. LL was partially supported by the Alexander von Humboldt Foundation. FAM and LL thank the Freie Universit\"{a}t Berlin for hospitality. FAM, LL, and VG thank  Daniel Miller, Matteo Rosati, Michele Notarnicola, Marco Avesani, Mateusz Mazelanik, and Marco Barbieri for useful discussions.

\medskip
\noindent\textbf{Author contributions} --- The entanglement distribution protocol was designed and optimised by FAM. The proof of Theorem~\ref{th1_main} was found in a blackboard discussion between the three authors. FAM wrote a first complete draft of the paper, which was subsequently improved by LL and VG. 

\medskip
\noindent\textbf{Supplementary Information} is available for this paper.

\medskip
\noindent\textbf{Competing interest} --- The authors declare no competing interests.

\medskip
\noindent\textbf{Data availability} --- No data sets were generated during this study.

\section*{Methods}\label{sec_methods}
Let $\mathfrak{S}(\HH)$ denote the space of density operators on a Hilbert space $\HH$. Let $\HH_2$ be a single-qubit Hilbert space with orthonormal basis $\{\ket{0}, \ket{1}\}$. For all $i,j\in\{0,1\}$, the state $\ket{\psi_{ij}}_{AB}\in\HH_2^{(A)}\otimes\HH_2^{(B)}$ defined as 
\bb\label{Bell_states_main}
\ket{\psi_{ij}}_{AB}\coloneqq \frac{1}{\sqrt{2}}\sum_{m=0}^1 (-1)^{im}\ket{m}_A\otimes\ket{m\oplus j}_B\, ,
\ee
where $\oplus$ denotes the modulo $2$ addition, is called a Bell state (or maximally entangled state). We will also refer to $\ket{\psi_{00}}$ as an entanglement bit, or \emph{ebit}. Any Bell state can be written in terms of the ebit as
\bb
    \ket{\psi_{ij}}_{AB}=\mathds{1}_A\otimes X^j Z^i \ket{\psi_{00}}_{AB}\,,
\ee
where $\mathds{1}_A$ denotes the identity operator on $\HH_2^{(A)}$ and $X,Z$ denote the well-known Pauli operators on $\HH_2^{(B)}$.  The Pauli-based twirling 
\bb
\mathcal{T}:\mathfrak{S}( \HH_2^{(A)}\otimes\HH_2^{(B)} )\to \mathfrak{S}( \HH_2^{(A)}\otimes\HH_2^{(B)} )
\ee
is defined as 
\bb\label{def_twirling_map}
    \mathcal{T}(\rho_{AB})=\frac{1}{4}\sum_{i,j=0}^{1}X^j_A Z^i_A\otimes  X^j_B Z^i_B\,\rho_{AB}\,(X^j_A Z^i_A\otimes  X^j_B Z^i_B)^\dagger
\ee
for all $\rho_{AB}\in\mathfrak{S}( \HH_2^{(A)}\otimes\HH_2^{(B)} )$ and it maps any input state in a Bell-diagonal state: 
\bb
    \mathcal{T}(\rho_{AB})=\sum_{i,j=0}^{1} \bra{\psi_{ij}}\rho_{AB}\ket{\psi_{ij}}\, \ketbra{\psi_{ij}}_{AB}.
\ee
Physically realisable transformations between two quantum systems with Hilbert spaces $\HH$ (input) and $\HH'$ (output) are modelled by quantum channels, i.e.\ completely positive and trace preserving maps $\Phi:\mathfrak{S}(\HH)\to \mathfrak{S}(\HH')$. The two-way quantum capacity $Q_2(\Phi)$ and secret-key capacity $K(\Phi)$ of a quantum channel $\Phi$ is the maximum achievable rate of qubits and secret-key bits, respectively, that can be reliably transmitted through $\Phi$ by assuming that the sender Alice and the receiver Bob have free access to a public, noiseless, two-way classical communication line. The rate of qubits (resp.\ secret-key bits) is defined as the ratio between the number of reliably transmitted qubits (resp.\ secret-key bits) and the number of uses of $\Phi$. A rigorous definition of the two-way capacities can be found in~\cite[Chapters 14 and 15]{Sumeet_book}. For any $\Phi$, the two-way capacities satisfy
\bb\label{relation_2waycap_main}
Q_2(\Phi)\le K(\Phi)\,,
\ee
since an ebit can generate a secret-key bit, thanks to the `E91' protocol~\cite{Ekert91}.

In practice, Alice has access to a limited budget $N_s$ of energy to produce each input signal, as measured by a Hamiltonian $H$ on $\HH$. Fixed $N_s>0$, the energy-constrained  two-way capacities $Q_2(\Phi, N_s)$ and $K(\Phi,N_s)$ are defined in the same way as the two-way capacities defined above, apart from the fact that the maximisation of the rate is restricted to the strategies such that the average expected value of $H$ on all input signals is required to be at most $N_s$.  
In addition note that the generalisation of~\eqref{relation_2waycap_main} to the energy-constrained case holds, i.e.
\bb\label{relation_2waycapEC_main}
Q_2(\Phi,N_s)\le K(\Phi,N_s)\,,
\ee
that any energy-constrained capacity is upper bounded by the corresponding unconstrained capacity, and {that it} tends to it in the limit $N_s\rightarrow\infty$.  

The goal of an entanglement distillation protocol is to turn a large number $n$ of copies of a bipartite entangled state $\rho_{AB} $ shared between Alice and Bob into a number $m$ of ebits by local operations and classical communication. The yield 
is defined by the ratio $m/n$. The distillable entanglement $E_d(\rho_{AB})$ of $\rho_{AB}$ is defined as the maximum yields over all the possible entanglement distillation protocols~\cite{reviewEDP_dur}~\cite[Chapter 8]{Sumeet_book}. The state $\rho_{AB}$ is said to be \emph{distillable} if $E_d(\rho_{AB})>0$. The coherent information (resp.~reverse coherent information) of $\rho_{AB}$ is defined by \bb
    I_{\text{c}}(\rho_{AB})\coloneqq S(\rho_B)-S(\rho_{AB})
\ee
(resp.\ $I_{\text{rc}}(\rho_{AB})\coloneqq S(\rho_A)-S(\rho_{AB})$), where $\rho_B\coloneqq \Tr_A\rho_{AB}$ and analogously for $\rho_A$, and moreover 
\bb
S(\sigma)\coloneqq -\Tr[ \sigma \log_2 \sigma]
\ee
is the von Neumann entropy. The yield $I_{\text{c}}(\rho_{AB})$ (resp.\ $I_{\text{rc}}(\rho_{AB})$) is achievable by an entanglement distillation protocol~\cite{devetak2005} that only exploits one-way forward (resp.~backward) classical communication. In particular, the following inequality, known as \emph{hashing inequality}, holds:
\begin{equation}\label{hashing_ineq_main}
    E_d(\rho_{AB})\ge \max\{I_{\text{c}}(\rho_{AB})\,, I_{\text{rc}}(\rho_{AB})\}\,.
\end{equation}
Let us briefly link the notions of distillable entanglement $E_d$ and two-way quantum capacity $Q_2(\Phi ,N_s)$. Suppose that Alice produces $n$ copies of a state $\rho_{AA'}$ such that $\Tr \rho_{A'} H_{A'}\leq N_s$. Then, she can use the channel $n$ times to send all subsystems $A'$ to Bob. Then, Alice and Bob share $n$ copies of $\Id_{A}\otimes\Phi(\rho_{AA'})$, which can now be used to generate $\simeq n\, E_d\left(\Id_{A}\otimes\Phi(\rho_{AA'}) \right) $ ebits by means of a suitable entanglement distillation protocol. Consequently, it holds that
\bb\label{link_D2_D_main}
    Q_2(\Phi,N_s)\ge E_d\left(\Id_{A}\otimes\Phi(\rho_{AA'}) \right)\,.
\ee

In the context of entanglement distillation, the goal of a \emph{recurrence protocol} is to transform a certain number of copies of the state $\rho_{AB}$ into fewer copies of another state $\rho'_{AB}$ such that $\bra{\psi_{00}}\rho'_{AB}\ket{\psi_{00}}>\bra{\psi_{00}}\rho_{AB}\ket{\psi_{00}}$~\cite{Bennett-error-correction,Bennett-distillation-mixed,reviewEDP_dur}. Examples of recurrence protocols for qubits can be found in~\cite{Bennett-distillation-mixed,DEJMPS,DNMV}, and their generalisations to the case of qudits in~\cite{Horodecki1999,Alber_2001,Dist-Number-Theory}. In the present paper we will exploit the recently introduced P1-or-P2 recurrence protocol~\cite{p1orp2}.   Since an infinite number of iterations of a recurrence protocol is generally needed to generate a Bell state $\ket{\psi_{00}}$, the yield of a recurrence protocol is zero. To achieve a nonzero yield, one may adopt a suitable number of iterations of a recurrence protocol and then apply the hashing or breeding protocol~\cite{Bennett-error-correction,Bennett-distillation-mixed}.
The latter protocols, which exploit only one-way classical communication, achieves the yield of the hashing inequality in~\eqref{hashing_ineq_main}.
Improvements of the hashing and breeding protocols, which exploit two-way classical communication and work on bipartite-qubit systems that are diagonal in the Bell basis, have been provided in~\cite{Improvement-Hashing}.

Let $\HH_S,\HH_E\coloneqq L^2(\mathbb{R})$ be single modes of electromagnetic radiation with definite frequency and polarisation, and let $a$ and $b$ be the corresponding annihilation operators. 
Now, let us define the piBGCs. For all $\lambda\in[0,1]$, $g\ge1$, $\nu\ge0$, $\xi\ge0$, the thermal attenuator $\mathcal{E}_{\lambda,\nu}$, the thermal amplifier $\Phi_{g,\nu}$, and the additive Gaussian noise $\Lambda_\xi$ are quantum channels on $\HH_S$ defined by 
\bb
    \mathcal{E}_{\lambda,\nu}(\rho)&\coloneqq\Tr_E\left[U_\lambda^{SE} \big(\rho^S \otimes\tau_\nu^E \big) {U_\lambda^{SE}}^\dagger\right]\,,\\
    \Phi_{g,\nu}(\rho)&\coloneqq\Tr_E\left[U_g^{SE} \big(\rho^S\otimes\tau_\nu^E\big) {U_g^{SE}}^\dagger\right]\,,\\
    \Lambda_\xi(\rho)&\coloneqq\frac{1}{\pi\xi}\int_{\mathbb C} \mathrm{d}^2 {z}\, e^{-\frac{|z|^2}{\xi}}  D(z)\,\rho\,  D(z)^\dagger\,,
\ee
where $\tau^E_\nu\coloneqq\frac{1}{\nu+1}\sum_{n=0}^\infty \left(\frac{\nu}{\nu+1}\right)^{n}\ketbra{n}_E$ is the thermal state with $\{\ket{n}_E\}_{n\in\N}$ being the Fock states on $\HH_E$, $U_\lambda^{SE}$ is the beam splitter unitary of transmissivity $\lambda$, $U_g^{SE}$ is the two-mode squeezing unitary of gain $g$, and $D(z)$ is the displacement operator:
\bb
    U_{\lambda}^{S E}&\coloneqq\exp\left[\arccos\sqrt{\lambda}\left(a^\dagger b-a\, b^\dagger\right)\right]\,,\\
    U_{g}^{S E}&\coloneqq\exp\left[\arccosh\sqrt{g}\left(a^\dagger b^\dagger-a\, b\right)\right]\,,\\
    D(z)&\coloneqq\exp{\left[z a^\dagger-z^\ast a\right] }\,.
\ee
In a communication scenario, the piBGCs are understood to map Alice's single-mode systems $A'$ to Bob's single-mode systems $B$.
The Hamiltonian on $\HH_S$ is the photon number operator $a^\dagger a$ and, by definition, the energy of an input signal initialised in a state $\rho$ is equal to its mean photon number $\Tr[\rho\, a^\dagger a]$.
The tightest known upper bounds on the two-way capacities of these channels, shown by Pirandola et al.~\cite{PLOB}, are
\bb\label{PLOB_Q2_main}
    K(\mathcal{E}_{\lambda,\nu})&\le \begin{cases}
-h(\nu)-\log_2[(1-\lambda)\lambda^\nu], & \text{if $\lambda\in(\frac{\nu}{\nu+1}, 1]$,} \\
0, & \text{otherwise}
\end{cases}\\
    K(\Phi_{g,\nu})&\le \begin{cases}
-h(\nu)+\log_2\left(\frac{g^{\nu+1}}{g-1}\right), & \text{if $g\in[1, 1+\frac{1}{\nu})$,} \\
0, & \text{otherwise}
\end{cases}\\
K(\Lambda_\xi)&\le  \begin{cases}
\frac{\xi-1}{\ln2}-\log_2(\xi), & \text{if $\xi\in[0,1)$,} \\
0, & \text{otherwise}
\end{cases}
\ee
where $h(\nu)\coloneqq(\nu+1)\log_2(\nu+1)-\nu\log_2\nu$ (see~\cite{MMMM} for a strong-converse extension of the formulas above). These upper bounds vanish if and only if the piBGCs are entanglement breaking~\cite{PLOB,Ent_breaking_Gaussian, Holevo-EB}. 
The tightest known lower bounds (before our work) on $Q_2$ are~\cite{Pirandola2009,holwer}
\begin{align}
    Q_2(\mathcal{E}_{\lambda,\nu}) &\ge\max\{0,-h(\nu)-\log_2(1-\lambda)\}\,, \label{lowQ2_main_att} \\
    Q_2(\Phi_{g,\nu})&\ge\max\left\{0,-h(\nu)+\log_2\left(\frac{g}{g-1}\right)\right\}\,, \label{lowQ2_main_amp} \\
    Q_2(\Lambda_\xi)&\ge \max\{0,-\log_2(e\,\xi)\}\,. \label{lowQ2_main_noise}
\end{align}
These lower bounds can be proved first by applying~\eqref{link_D2_D_main} with the choice $\rho_{AA'}=\ketbra{\Psi_{N_s}}_{AA'}$, where
\bb\label{TMSV_def}
\ket{\Psi_{N_s}}\coloneqq \frac{1}{\sqrt{N_s+1}}\sum_{n=0}^\infty \left(\frac{N_s}{N_s+1}\right)^{n/2}\ket{n}\otimes \ket{n}
\ee
is the two-mode squeezed vacuum state with local mean photon number equal to $N_s$, second by applying the hashing inequality in~\eqref{hashing_ineq_main}, and finally by taking the limit $N_s\rightarrow\infty$. Specifically, the lower bound in~\eqref{lowQ2_main_att} is achieved by the reverse coherent information, while that in~\eqref{lowQ2_main_amp}--\eqref{lowQ2_main_noise} is achieved by the coherent information.
Although the right-hand sides of~\eqref{lowQ2_main_att}--\eqref{lowQ2_main_noise} also lower bound the secret-key capacity $K$, improved estimates of $K(\mathcal{E}_{\lambda,\nu})$ and $K(\Phi_{g,\nu})$ have been put forth by Ottaviani et al.~\cite{Ottaviani_new_lower} (see also~\cite[Sec.~VII]{Pirandola18}) and by Wong et al.~\cite{Wang_Q2_amplifier}, respectively.

Lower bounds on the energy-constrained two-way capacities with energy constraint $N_s$ of piBGCs are the coherent information and the reverse coherent information evaluated on the state obtained by sending the subsystem $A'$ of $\ket{\Psi_{N_s}}_{AA'}$ through the channel. For sufficiently small values of $N_s$, improved lower bounds have been found by Noh et al.~\cite{Noh2020}. The best known upper bound on the energy-constrained two-way capacity of the thermal attenuator is --- depending on the parameters $\lambda$, $\nu$, and $N_s$ --- the unconstrained upper bound discovered by Pirandola et al.~\cite{PLOB} reported in~\eqref{PLOB_Q2_main} or the bound found by Davis et al.~\cite{Davis2018} (which is equal to the bound found in~\cite{TGW} for $\nu=0$). Upper bounds on the energy-constrained two-way capacities of the thermal amplifier and additive Gaussian noise are the unconstrained upper bound discovered by Pirandola et al.~\cite{PLOB} reported in~\eqref{PLOB_Q2_main} and the bounds which can be obtained by exploiting the results of~\cite{Goodenough16, Davis2018}.  

\bibliographystyle{unsrt}
\bibliography{biblio}

\begin{thebibliography}{10}

\bibitem{bennett1984quantum}
C.~H. Bennett.
\newblock Quantum cryptography: public key distribution and coin tossing.
\newblock In {\em Proc. IEEE International Conference on Computers, Systems and
  Signal Processing, Bangalore, India}, pages 175--179, 1984.

\bibitem{BUCCO}
A.~Serafini.
\newblock {\em Quantum Continuous Variables: A Primer of Theoretical Methods}.
\newblock CRC Press, Taylor \& Francis Group, Boca Raton, USA, 2017.

\bibitem{CV_qkd}
F.~Grosshans and P.~Grangier.
\newblock Continuous variable quantum cryptography using coherent states.
\newblock {\em Phys. Rev. Lett.}, 88(5), 2002.

\bibitem{Pirandola20}
S.~Pirandola et~al.
\newblock Advances in quantum cryptography.
\newblock {\em Advances in Optics and Photonics}, 12(4):1012--1236, 2020.

\bibitem{Laudenbach_2018}
F.~Laudenbach, C.~Pacher, C.-H.~F. Fung, A.~Poppe, M.~Peev, B.~Schrenk,
  M.~Hentschel, P.~Walther, and H.~H\"{u}bel.
\newblock Continuous-variable quantum key distribution with gaussian
  modulation-the theory of practical implementations.
\newblock {\em Advanced Quantum Technologies}, 1(1), 2018.

\bibitem{Record1}
Y.~Zhang, Z.~Chen, S.~Pirandola, X.~Wang, C.~Zhou, B.~Chu, Y.~Zhao, B.~Xu,
  S.~Yu, and H.~Guo.
\newblock Long-distance continuous-variable quantum key distribution over
  202.81 km of fiber.
\newblock {\em Phys. Rev. Lett.}, 125:010502, 2020.

\bibitem{Record2}
A.~A.~E. Hajomer, I.~Derkach, N.~Jain, H.-M. Chin, U.~L. Andersen, and
  T.~Gehring.
\newblock Long-distance continuous-variable quantum key distribution over
  100-km fiber with local local oscillator.
\newblock {\em Science Advances}, 10(1):eadi9474, 2024.

\bibitem{Record3}
Y.~Zhang, Z.~Li, Z.~Chen, C.~Weedbrook, Y.~Zhao, X.~Wang, Y.~Huang, C.~Xu,
  X.~Zhang, Z.~Wang, M.~Li, X.~Zhang, Z.~Zheng, B.~Chu, X.~Gao, N.~Meng,
  W.~Cai, Z.~Wang, G.~Wang, S.~Yu, and H.~Guo.
\newblock Continuous-variable {QKD} over 50 km commercial fiber.
\newblock {\em Quantum Science and Technology}, 4(3):035006, 2019.

\bibitem{Record4}
Y.~Pi, H.~Wang, Y.~Pan, Y.~Shao, Y.~Li, J.~Yang, Y.~Zhang, W.~Huang, and B.~Xu.
\newblock Sub-mbps key-rate continuous-variable quantum key distribution with
  local local oscillator over 100-km fiber.
\newblock {\em Optics Letters}, 48(7):1766, 2023.

\bibitem{Record5}
D.~Huang, P.~Huang, D.~Lin, and G.~Zeng.
\newblock Long-distance continuous-variable quantum key distribution by
  controlling excess noise.
\newblock {\em Scientific Reports}, 6:19201, 01 2016.

\bibitem{quantum_internet_Wehner}
S.~Wehner, D.~Elkouss, and R.~Hanson.
\newblock Quantum internet: A vision for the road ahead.
\newblock {\em Science}, 362:eaam9288, 10 2018.

\bibitem{Pirandola18}
S.~Pirandola, S.~L. Braunstein, R.~Laurenza, C.~Ottaviani, T.~P.~W. Cope,
  G.~Spedalieri, and L.~Banchi.
\newblock Theory of channel simulation and bounds for private communication.
\newblock {\em Quantum Science and Technology}, 3(3):035009, 2018.

\bibitem{MARK}
M.~M. Wilde.
\newblock {\em Quantum Information Theory}.
\newblock Cambridge University Press, 2nd edition, 2017.

\bibitem{Sumeet_book}
S.~Khatri and M.~M. Wilde.
\newblock Principles of quantum communication theory: A modern approach, 2020.

\bibitem{Davis2018}
N.~Davis, M.~E. Shirokov, and M.~M. Wilde.
\newblock Energy-constrained two-way assisted private and quantum capacities of
  quantum channels.
\newblock {\em Phys. Rev. A}, 97:062310, 2018.

\bibitem{Ekert91}
A.~K. Ekert.
\newblock Quantum cryptography based on {B}ell's theorem.
\newblock {\em Phys. Rev. Lett.}, 67:661--663, 1991.

\bibitem{PLOB}
S.~Pirandola, R.~Laurenza, C.~Ottaviani, and L.~Banchi.
\newblock Fundamental limits of repeaterless quantum communications.
\newblock {\em Nat. Commun.}, 8(1):15043, 2017.

\bibitem{repeaters}
H.-J. Briegel, W.~D\"ur, J.~I. Cirac, and P.~Zoller.
\newblock Quantum repeaters: The role of imperfect local operations in quantum
  communication.
\newblock {\em Phys. Rev. Lett.}, 81:5932--5935, 1998.

\bibitem{Munro2015}
W.~J. Munro, K.~Azuma, K.~Tamaki, and K.~Nemoto.
\newblock Inside quantum repeaters.
\newblock {\em IEEE Journal of Selected Topics in Quantum Electronics},
  21(3):78--90, 2015.

\bibitem{Goodenough16}
K.~Goodenough, D.~Elkouss, and S.~Wehner.
\newblock Assessing the performance of quantum repeaters for all
  phase-insensitive {G}aussian bosonic channels.
\newblock {\em New Journal of Physics}, 18(6):063005, 2016.

\bibitem{TGW}
M.~Takeoka, S.~Guha, and M.~M. Wilde.
\newblock Fundamental rate-loss tradeoff for optical quantum key distribution.
\newblock {\em Nat. Commun.}, 5(1):5235, 2014.

\bibitem{holwer}
A.~S. Holevo and R.~F. Werner.
\newblock Evaluating capacities of bosonic {G}aussian channels.
\newblock {\em Phys. Rev. A}, 63:032312, 2001.

\bibitem{MMMM}
M.~M. Wilde, M.~Tomamichel, and M.~Berta.
\newblock Converse bounds for private communication over quantum channels.
\newblock {\em IEEE Transactions on Information Theory}, 63(3):1792--1817,
  2017.

\bibitem{squashed_channel}
M.~Takeoka, S.~Guha, and M.~M. Wilde.
\newblock The squashed entanglement of a quantum channel.
\newblock {\em IEEE Transactions on Information Theory}, 60(8):4987--4998,
  2014.

\bibitem{Pirandola2009}
S.~Pirandola, R.~Garc\'{\i}a-Patr\'on, S.~L. Braunstein, and S.~Lloyd.
\newblock Direct and reverse secret-key capacities of a quantum channel.
\newblock {\em Phys. Rev. Lett.}, 102:050503, 2009.

\bibitem{Noh2020}
K.~Noh, S.~Pirandola, and L.~Jiang.
\newblock Enhanced energy-constrained quantum communication over bosonic
  {G}aussian channels.
\newblock {\em Nat. Commun.}, 11(1):457, 2020.

\bibitem{Ottaviani_new_lower}
C.~Ottaviani, R.~Laurenza, T.~P.~W. Cope, G.~Spedalieri, S.~L. Braunstein, and
  S.~Pirandola.
\newblock {Secret key capacity of the thermal-loss channel: improving the lower
  bound}.
\newblock In M.~T. Gruneisen, M.~Dusek, and J.~G. Rarity, editors, {\em Quantum
  Information Science and Technology II}, volume 9996, page 999609.
  International Society for Optics and Photonics, SPIE, 2016.

\bibitem{Tamura2018}
Y.~Tamura, H.~Sakuma, K.~Morita, M.~Suzuki, Y.~Yamamoto, K.~Shimada, Y.~Honma,
  K.~Sohma, T.~Fujii, and T.~Hasegawa.
\newblock The first 0.14-db/km loss optical fiber and its impact on submarine
  transmission.
\newblock {\em J. Lightwave Technol.}, 36(1):44--49, 2018.

\bibitem{Li2020}
M.-J. Li and T.~Hayashi.
\newblock Chapter 1 -- {A}dvances in low-loss, large-area, and multicore
  fibers.
\newblock In A.~E. Willner, editor, {\em Optical Fiber Telecommunications VII},
  pages 3--50. Academic Press, 2020.

\bibitem{PeresPPT}
A.~Peres.
\newblock {Separability criterion for density matrices}.
\newblock {\em Phys. Rev. Lett.}, 77:1413--1415, 1996.

\bibitem{Simon00}
R.~Simon.
\newblock {Peres--Horodecki} separability criterion for continuous variable
  systems.
\newblock {\em Phys. Rev. Lett.}, 84:2726--2729, 2000.

\bibitem{Giedke01}
G.~Giedke, B.~Kraus, M.~Lewenstein, and J.~I. Cirac.
\newblock Entanglement criteria for all bipartite {G}aussian states.
\newblock {\em Phys. Rev. Lett.}, 87:167904, 2001.

\bibitem{Ent_breaking_Gaussian}
A.~S. Holevo and V.~Giovannetti.
\newblock Quantum channels and their entropic characteristics.
\newblock {\em Reports on Progress in Physics}, 75(4):046001, 2012.

\bibitem{Holevo-EB}
A.~S. Holevo.
\newblock Entanglement-breaking channels in infinite dimensions.
\newblock {\em Probl. Pered. Inform.}, 44(3):3--18, 2008.
\newblock (English translation: Probl. Inf. Transm. 44(3):171--184, 2008).

\bibitem{Pirandola2021}
S.~Pirandola.
\newblock Limits and security of free-space quantum communications.
\newblock {\em Phys. Rev. Research}, 3:013279, 2021.

\bibitem{Wang_Q2_amplifier}
G.~Wang, C.~Ottaviani, H.~Guo, and S.~Pirandola.
\newblock Improving the lower bound to the secret-key capacity of the thermal
  amplifier channel.
\newblock {\em The European Physical Journal D}, 73(1):17, 2019.

\bibitem{devetak2005}
I.~Devetak and A.~Winter.
\newblock Distillation of secret key and entanglement from quantum states.
\newblock {\em Proc. Royal Soc. A}, 461(2053):207--235, 2005.

\bibitem{reviewEDP_dur}
H.~D\"{u}r and Briegel~H. J.
\newblock Entanglement purification and quantum error correction.
\newblock {\em Reports on Progress in Physics}, 70(8):1381--1424, 2007.

\bibitem{Bennett-error-correction}
C.~H. Bennett, D.~P. DiVincenzo, J.~A. Smolin, and W.~K. Wootters.
\newblock Mixed-state entanglement and quantum error correction.
\newblock {\em Phys. Rev. A}, 54:3824--3851, 1996.

\bibitem{Bennett-distillation-mixed}
C.~H. Bennett, G.~Brassard, S.~Popescu, B.~Schumacher, J.~A. Smolin, and W.~K.
  Wootters.
\newblock Purification of noisy entanglement and faithful teleportation via
  noisy channels.
\newblock {\em Phys. Rev. Lett.}, 76:722--725, 1996.

\bibitem{p1orp2}
J.~Miguel-Ramiro and W.~D\"ur.
\newblock Efficient entanglement purification protocols for $d$-level systems.
\newblock {\em Phys. Rev. A}, 98:042309, 2018.

\bibitem{Improvement-Hashing}
K.~G.~H. Vollbrecht and F.~Verstraete.
\newblock Interpolation of recurrence and hashing entanglement distillation
  protocols.
\newblock {\em Phys. Rev. A}, 71:062325, 2005.

\bibitem{Sanders1989}
B.~C. Sanders.
\newblock Quantum dynamics of the nonlinear rotator and the effects of
  continual spin measurement.
\newblock {\em Phys. Rev. A}, 40:2417--2427, 1989.

\bibitem{Winnel}
M.~S. Winnel, J.~J. Guanzon, N.~Hosseinidehaj, and T.~C. Ralph.
\newblock Achieving the ultimate end-to-end rates of lossy quantum
  communication networks.
\newblock {\em arXiv:2203.13924}, 2022.

\bibitem{Kalb_2017}
N.~N.~Kalb, A.~A. A.~A.~Reiserer, P.~C. Humphreys, J.~J.~W. Bakermans, S.~J.
  Kamerling, N.~H. Nickerson, S.~C. Benjamin, D.~J. Twitchen, M.~Markham, and
  R.~Hanson.
\newblock Entanglement distillation between solid-state quantum network nodes.
\newblock {\em Science}, 356(6341):928--932, 2017.

\bibitem{Hu_2021}
X.-M. Hu, C.-X. Huang, Y.-B. Sheng, L.~Zhou, B.-H. Liu, Y.~Guo, C.~Zhang, W.-B.
  Xing, Y.-F. Huang, C.-F. Li, and G.-C. Guo.
\newblock Long-distance entanglement purification for quantum communication.
\newblock {\em Phys. Rev. Lett.}, 126(1), 2021.

\bibitem{Ecker_2021}
S.~Ecker, P.~Sohr, L.~Bulla, M.~Huber, M.~Bohmann, and R.~Ursin.
\newblock Experimental single-copy entanglement distillation.
\newblock {\em Phys. Rev. Lett.}, 127(4), 2021.

\bibitem{DEJMPS}
D.~Deutsch, A.~Ekert, R.~Jozsa, C.~Macchiavello, S.~Popescu, and A.~Sanpera.
\newblock Quantum privacy amplification and the security of quantum
  cryptography over noisy channels.
\newblock {\em Phys. Rev. Lett.}, 77:2818--2821, 1996.

\bibitem{DNMV}
J.~Dehaene, M.~Van~den Nest, B.~De~Moor, and F.~Verstraete.
\newblock Local permutations of products of bell states and entanglement
  distillation.
\newblock {\em Phys. Rev. A}, 67:022310, 2003.

\bibitem{Horodecki1999}
M.~Horodecki and P.~Horodecki.
\newblock Reduction criterion of separability and limits for a class of
  distillation protocols.
\newblock {\em Phys. Rev. A}, 59:4206--4216, 1999.

\bibitem{Alber_2001}
G.~Alber, A.~Delgado, N.~Gisin, and I.~Jex.
\newblock Efficient bipartite quantum state purification in arbitrary
  dimensional hilbert spaces.
\newblock {\em Journal of Physics A: Mathematical and General},
  34(42):8821--8833, 2001.

\bibitem{Dist-Number-Theory}
H.~Bombin and M.~A. Martin-Delgado.
\newblock Entanglement distillation protocols and number theory.
\newblock {\em Phys. Rev. A}, 72:032313, 2005.

\bibitem{Numerical-Improvement-Hashing}
E.~Hostens, J.~Dehaene, and B.~De~Moor.
\newblock Asymptotic adaptive bipartite entanglement-distillation protocol.
\newblock {\em Phys. Rev. A}, 73:062337, 2006.

\bibitem{BARNETT-RADMORE}
S.~Barnett and P.~M. Radmore.
\newblock {\em Methods in Theoretical Quantum Optics}.
\newblock Oxford Series in Optical and Imaging Sciences. Clarendon Press, 2002.

\bibitem{Cushen1971}
C.~D. Cushen and R.~L. Hudson.
\newblock A quantum-mechanical central limit theorem.
\newblock {\em Journal of Applied Probability}, 8(3):454--469, 1971.

\bibitem{Die-Hard-2-PRA}
F.~A. Mele, L.~Lami, and V.~Giovannetti.
\newblock Quantum optical communication in the presence of strong attenuation
  noise.
\newblock {\em Phys. Rev. A}, 106:042437, 2022.

\bibitem{Holevo-CJ}
A.~S. Holevo.
\newblock The {Choi--Jamiolkowski} forms of quantum {G}aussian channels.
\newblock {\em J. Math. Phys.}, 52(4):042202, 2011.

\bibitem{Holevo-CJ-arXiv}
A.~S. Holevo.
\newblock On the {Choi--Jamiolkowski} correspondence in infinite dimensions.
\newblock {\em Preprint arXiv:1004.0196}, 2010.

\bibitem{2-qubit-distillation}
M.~Horodecki, P.~Horodecki, and R.~Horodecki.
\newblock Inseparable two spin-$\frac{1}{2}$ density matrices can be distilled
  to a singlet form.
\newblock {\em Phys. Rev. Lett.}, 78:574--577, 1997.

\bibitem{q_memory1}
E.~Bersin, M.~Sutula, Y.~Q. Huan, A.~Suleymanzade, D.~R. Assumpcao, Y.-C. Wei,
  P.-J. Stas, C.~M. Knaut, E.~N. Knall, C.~Langrock, N.~Sinclair, R.~Murphy,
  R.~Riedinger, M.~Yeh, C.~J. Xin, S.~Bandyopadhyay, D.~D. Sukachev,
  B.~Machielse, D.~S. Levonian, M.~K. Bhaskar, S.~Hamilton, H.~Park,
  M.~Lon\ifmmode~\check{c}\else \v{c}\fi{}ar, M.~M. Fejer, P.~B. Dixon, D.~R.
  Englund, and M.~D. Lukin.
\newblock Telecom networking with a diamond quantum memory.
\newblock {\em PRX Quantum}, 5:010303, 2024.

\bibitem{q_memory2}
A.~Wallucks, I.~Marinkovic, B.~Hensen, R.~Stockill, and S.~Groblacher.
\newblock A quantum memory at telecom wavelengths.
\newblock {\em Nat. Phys.}, 16(7):772--777, 2020.

\bibitem{exp_noon1}
G.~J. Pryde and A.~G. White.
\newblock Creation of maximally entangled photon-number states using optical
  fiber multiports.
\newblock {\em Phys. Rev. A}, 68(5), 2003.

\bibitem{exp_noon2}
I.~Afek, O.~Ambar, and Y.~Silberberg.
\newblock High-{NOON} states by mixing quantum and classical light.
\newblock {\em Science}, 328(5980):879--881, 2010.

\bibitem{measurement1}
L.~Liang, G.~W. Lin, Y.~M. Hao, Y.~P. Niu, and S.~Q. Gong.
\newblock Quantum nondemolition measurement of small photon numbers using
  stored light.
\newblock {\em Phys. Rev. A}, 90:055801, Nov 2014.

\bibitem{measurement2}
A.~Cabello and F.~Sciarrino.
\newblock Loophole-free {B}ell test based on local precertification of photon's
  presence.
\newblock {\em Phys. Rev. X}, 2:021010, 2012.

\bibitem{rydberg1}
A.~V. Gorshkov, R.~Nath, and T.~Pohl.
\newblock Dissipative many-body quantum optics in {R}ydberg media.
\newblock {\em Phys. Rev. Lett.}, 110:153601, 2013.

\bibitem{rydberg2}
G.-S. Ye, B.~Xu, Y.~Chang, S.~Shi, T.~Shi, and L.~Li.
\newblock A photonic entanglement filter with {R}ydberg atoms.
\newblock {\em Nat. Photonics}, 17(6):538--543, 2023.

\bibitem{rydberg3}
J.~Honer, R.~L\"ow, H.~Weimer, T.~Pfau, and H.~P. B\"uchler.
\newblock Artificial atoms can do more than atoms: Deterministic single photon
  subtraction from arbitrary light fields.
\newblock {\em Phys. Rev. Lett.}, 107:093601, 2011.

\bibitem{rydberg4}
N.~Stiesdal, H.~Busche, K.~Kleinbeck, J.~Kumlin, M.~G. Hansen, H.~P.
  B\"{u}chler, and S.~Hofferberth.
\newblock Controlled multi-photon subtraction with cascaded {R}ydberg
  superatoms as single-photon absorbers.
\newblock {\em Nat. Commun.}, 12(1), 2021.

\bibitem{substraction1}
M.~M\"{u}cke, E.~Figueroa, J.~Bochmann, C.~Hahn, K.~Murr, S.~Ritter,
  C.~Villas-Boas, and G.~Rempe.
\newblock Electromagnetically induced transparency with single atoms in a
  cavity.
\newblock {\em Nature}, 465:755--8, 06 2010.

\bibitem{substraction2}
S.~Rosenblum, O.~Bechler, I.~Shomroni, Y.~Lovsky, G.~Guendelman, and B.~Dayan.
\newblock Extraction of a single photon from an optical pulse.
\newblock {\em Nat. Photonics}, 10(1):19--22, 2015.

\bibitem{transf1}
A.~Kumar, A.~Suleymanzade, M.~Stone, L.~Taneja, A.~Anferov, D.~I. Schuster, and
  J.~Simon.
\newblock {Quantum-enabled millimetre wave to optical transduction using
  neutral atoms}.
\newblock {\em Nature}, 615(7953):614--619, 2023.

\bibitem{transf2}
J.~Rochman, T.~Xie, J.~G. Bartholomew, K.~C. Schwab, and A.~Faraon.
\newblock {Microwave-to-optical transduction with erbium ions coupled to planar
  photonic and superconducting resonators}.
\newblock {\em Nature Commun.}, 14(1):1153, 2023.

\bibitem{transf3}
Y.~Xu, A.~Al~Sayem, L.~Fan, C.-L. Zou, S.~Wang, R.~Cheng, W.~Fu, L.~Yang,
  M.~Xu, and H.~X. Tang.
\newblock Bidirectional interconversion of microwave and light with thin-film
  lithium niobate.
\newblock {\em Nat. Commun.}, 12, 2021.

\end{thebibliography}

\newpage
\clearpage

\onecolumngrid
\begin{center}
\vspace*{\baselineskip}
{\textbf{\large Supplemental material:\\ Maximum tolerable excess noise in CV-QKD and improved lower bound on two-way capacities}}\\
\end{center}

\renewcommand{\theequation}{S\arabic{equation}}
\renewcommand{\thethm}{S\arabic{thm}}
\setcounter{equation}{0}
\setcounter{thm}{0}
\setcounter{figure}{1}
\setcounter{table}{0}
\setcounter{section}{0}
\setcounter{page}{1}
\makeatletter

\setcounter{secnumdepth}{2}

\section{Notation and preliminaries}
Let $\mathfrak{S}(\HH)$ be the set of quantum states on a Hilbert space $\HH$. The trace norm of a bounded linear operator $\Theta$ is defined by $\|\Theta\|_1\coloneqq \Tr\sqrt{\Theta^\dagger\Theta}\,.$ The von Neumann entropy of a quantum state $\rho$ is denoted by $S(\rho)\coloneqq -\Tr\left[\rho\log_2\rho\right]\,$. Let $\HH_2$ be a bi-dimensional Hilbert space and let $\{\ket{0}, \ket{1}\}$ be an orthonormal basis. For all $i,j\in\{0,1\}$, the state $\ket{\psi_{ij}}_{AB}\in\HH_2^{(A)}\otimes\HH_2^{(B)}$ is defined as 
\bb \label{Bell_states}
\ket{\psi_{ij}}_{AB}\coloneqq \frac{1}{\sqrt{2}}\sum_{m=0}^1 (-1)^{im}\ket{m}_A\otimes\ket{m\oplus j}_B\,,
\ee
and is called a Bell state (or maximally entangled state), where $\oplus$ denotes the modulo $2$ addition.

\subsection{Gaussian quantum information}
Let us briefly review the formalism of Gaussian quantum information~\cite{BUCCO}. We consider $m$-modes of harmonic oscillators $S_1$, $S_2$, $\ldots$, $S_m$, which are associated with the Hilbert space $L^2(\mathbb R^m)$ of square integrable functions. Each of these modes represents a single-mode of electromagnetic radiation with definite frequency and polarisation.
For all $j=1,2,\ldots,m$ the annihilation operator $a_j$ of the mode $S_i$ is defined as $a_j\coloneqq \frac{\hat{x}_j+i\hat{p}_j}{\sqrt{2}}$, where $\hat{x}_j$ and $\hat{p}_j$ are the well-known position and momentum operators of $S_j$. The operator $a_j^\dagger a_j$ is called the photon number of the mode $S_j$. The $n$th Fock state of the mode $S_j$ is denoted by $\ket{n}_{S_j}$. By defining the so-called quadrature vector $\mathbf{\hat{R}}\coloneqq (\hat{x}_1,\hat{p}_1,...,\hat{x}_m,\hat{p}_m)^{\intercal}$, one can write the canonical commutation relations as $[\mathbf{\hat{R}},\mathbf{\hat{R}}^{\intercal}]=i\,\Omega_m$, where $\Omega_m\coloneqq\mathbb{1}_{m}\otimes\left(\begin{matrix}0&1\\-1&0\end{matrix}\right)$ and $\mathbb{1}_{m}$ is the $m\times m$ identity matrix. The characteristic function $\chi_\rho: \mathbb{R}^{2m}\to \mathbb{C}$ of a state $\rho\in\mathfrak{S}(L^2(\mathbb R^m))$ is defined as $\chi_\rho(\mathbf r)=\Tr[ \rho  D_{-\mathbf{r}} ]$, where for all $\mathbf{r}\in \mathbb{R}^{2m}$ the displacement operator $ D_{\mathbf{r}}$ is defined as 
\bb\label{def_charact_func}
D_{\mathbf{r}}\coloneqq e^{i {\mathbf{r}}^{\intercal}\Omega_m \mathbf{\hat{R}}}\,.
\ee
Any state $\rho$ can be written in terms of its characteristic function as
\bb\label{inverse_fourier_displacement}
\rho=\int_{\mathbb{R}^{2m}}\frac{\mathrm{d}^{2m}\mathbf{r}}{(2\pi)^m}\chi_\rho(\mathbf r) D_{\mathbf{r}}\,
\ee
and hence quantum states and characteristic functions are in one-to-one correspondence. The first moment and the covariance matrix of a quantum state $\rho$ are defined as
\begin{align}
	&\mathbf{m}(\rho)=\Tr\left[\mathbf{\hat{R}}\,\rho\right]\,,\\
	&V(\rho)=\Tr\left[\left\{\mathbf{(\hat{R}-m(\rho)),(\hat{R}-m(\rho))}^{\intercal}\right\}\rho\right]\, ,
\end{align}
respectively, where $\{A,B\}\coloneqq AB+BA$ is the anti-commutator. Note that the covariance matrix is defined with respect an ordering of the modes in the definition of the quadrature vector: here such an ordering is $(S_1, S_2, \ldots,S_m)$. A state $\rho$ is said to be Gaussian if there exists a $2m\times 2m$ real positive definite matrix $H_\rho$ and a vector $\mathbf{m}_\rho\in\mathbb{R}^{2m}$ such that $\rho$ can be written as a ground or a thermal state of the Hamiltonian $\frac{1}{2}(\mathbf{\hat{R}}-\mathbf{m}_\rho)^{\intercal}H(\mathbf{\hat{R}}-\mathbf{m}_\rho)$, i.e.
\bb
    \rho=\frac{e^{ -\frac{1}{2}(\mathbf{\hat{R}}-\mathbf{m}_\rho)^{\intercal}H_\rho(\mathbf{\hat{R}}-\mathbf{m}_\rho)} }{\Tr\left[ e^{ -\frac{1}{2}(\mathbf{\hat{R}}-\mathbf{m}_\rho)^{\intercal}H_\rho(\mathbf{\hat{R}}-\mathbf{m}_\rho)} \right]}\,.
\ee
It can be shown that $\mathbf{m}(\rho)=\mathbf{m}_\rho$ and $V(\rho)=V_\rho$, where $V_\rho\coloneqq \coth{\left(\frac{i\,\Omega_m H_\rho}{2}\right)}i\,\Omega_m$. The characteristic function of a Gaussian state $\rho$ is a Gaussian function in $\mathbf{r}$ which can be written in terms of $\mathbf{m}(\rho) $ and $V(\rho)$ as
\bb
\chi_{\rho}(\mathbf{r})=\exp\left( -\frac{1}{4}(\Omega_m \mathbf{r})^{\intercal}V(\rho)\Omega_m \mathbf{r}+i(\Omega_m \mathbf{r})^{\intercal}\mathbf{m}(\rho) \right)\,.
\ee
An example of Gaussian state is the thermal state $\tau_{N_s}\coloneqq \frac{1}{N_s+1}\sum_{n=0}^\infty \left(\frac{N_s}{N_s+1}\right)^{n}\ketbra{n}$, where the parameter $N_s\ge0$ is its mean photon number ($N_s=\Tr[a^\dagger a\,\tau_{N_s} ]$), which satisfies
\bb\label{moments_thermal}
\mathbf{m}(\tau_{N_s})&=(0,0)^{\intercal}\,,\\
V(\tau_{N_s})&=(2N_S+1)\mathbb{1}_2 \,.
\ee
Another example of Gaussian state is the two-mode squeezed vacuum state $\ket{\Psi_{N_s}}_{S_1 S_2}$, which for all $N_s\ge0$ it is defined as
\bb\label{two_mode_sq}
\ket{\Psi_{N_s}}_{S_1S_2}\coloneqq \frac{1}{\sqrt{N_s+1}}\sum_{n=0}^\infty \left(\frac{N_s}{N_s+1}\right)^{n/2}\ket{n}_{S_1}\ket{n}_{S_2}\,,
\ee
where $N_s$ denotes the mean photon number of the mode $S_1$ (or, equivalently, of the mode $S_2$), i.e.~
\bb
    N_s=\Tr_{S_2}[a_1^\dagger a_1\,\ketbra{\Psi_{N_s}}_{S_1S_2} ]\,.
\ee
The first moment and covariance matrix of $\ket{\Psi_{N_s}}_{S_1S_2}$ are
\bb\label{moments_squeezed}
\mathbf{m}(\ketbra{\Psi_{N_s}})&=(0,0,0,0)^{\text{T}}\,,\\
V(\ketbra{\Psi_{N_s}})&=\left(\begin{matrix} (2N_s+1)\mathbb{1}_2 & 2\sqrt{N_s(N_s+1)}\sigma_z \\ 2\sqrt{N_s(N_s+1)}\sigma_z  &(2N_s+1)\mathbb{1}_2\end{matrix}\right)\,,
\ee
where $\mathbb{1}_2\coloneqq \left(\begin{matrix}1&0\\0&1\end{matrix}\right)$ and $\sigma_z\coloneqq \left(\begin{matrix}1&0\\0&-1\end{matrix}\right)$.

A quantum channel is said to be Gaussian if it maps Gaussian states into Gaussian states. Later we will focus on three important examples of Gaussian quantum channels: the thermal attenuator, the thermal amplifier, and the additive Gaussian noise.
Before concluding this brief recap of Gaussian quantum information, let us state a lemma which will be useful in the following. The forthcoming Lemma~\ref{ConditionPPT_cov} provides a necessary and sufficient condition on the covariance matrix to assess whether a two-mode Gaussian state is entangled~\cite{Simon00, BUCCO}. This condition is based on the fact that a two-mode Gaussian states is separable (not entangled) if and only if it is PPT~\cite{Simon00, BUCCO}.
\begin{lemma}[\cite{Simon00, BUCCO}]\label{ConditionPPT_cov}
    Let $\rho\in\mathfrak{S}(\HH_{S_1}\otimes \HH_{S_2})$ be a two-mode Gaussian state. Let us write its covariance matrix $V(\rho)$ with respect the ordering $(S_1,S_2)$ as
	\bb
	V\left( \rho \right)= \left(\begin{matrix} V_{S_1} & V_{S_1S_2} \\ V_{S_1S_2}^{\intercal} & V_{S_2}\end{matrix}\right)\,
	\ee 
	and define the function $f:\mathfrak{S}(\HH_{S_1}\otimes \HH_{S_2})\to\mathbb{R}$ as $f(\rho)\coloneqq 1+\det(V(\rho))+2\det(V_{S_1S_2})- \det(V_{S_1})-\det(V_{S_2})$.
    The state $\rho$ is entangled if and only if $f(\rho)<0$.
\end{lemma}
The forthcoming Lemma~\ref{holevo_lemma_eb_gauss} gives necessary and sufficient condition on a Gaussian quantum channel to be entanglement breaking~\cite[Chapter 4.6]{MARK}.
\begin{lemma}\cite{Holevo-EB}\label{holevo_lemma_eb_gauss}
    Let $\Phi:\mathfrak{S}(L^2(\mathbb R^m))\to\mathfrak{S}(L^2(\mathbb R^m))$ be a Gaussian quantum channel. Let $K,\beta\in\mathbb{R}^{2m\times2m}$ and $l\in\mathbb{R}^{2m}$ such that for all $\rho\in\mathfrak{S}(L^2(\mathbb R^m))$ it holds that
\bb 
&\mathbf{m}\left(\Phi(\rho)\right)=K\, \mathbf{m}(\rho)\,,\\
&V\left(\Phi (\rho) \right)=K^{\intercal}\, V(\rho)K+\beta\,.
\ee
Then, $\Phi$ is entanglement breaking if and only if $\beta$ admits the following decomposition:
\bb
\beta=\alpha+\gamma\,,\quad\text{where }\alpha,\gamma\in\mathbb{R}^{2m\times 2m}\text{ with }\alpha\ge i\, \Omega_{m}\,\text{ and }\,\gamma\ge i K^{\intercal}\,\Omega_m K\,.
\ee
    
\end{lemma}

\subsection{Two-way capacities of a quantum channel}
The two-way quantum capacity $Q_2(\Phi)$ and the secret-key capacity $K(\Phi)$ of a quantum channel $\Phi$ are the maximum achievable rate of qubits and secret-key bits, respectively, that can be reliably transmitted through $\Phi$ by assuming that the sender Alice and the receiver Bob have free access to a public, noiseless, two-way classical communication line. The rate of qubits (resp.~secret-key bits) is defined as the ratio between the number of reliably transmitted qubits (resp.~secret-key bits) and the number of uses of $\Phi$~\cite[Chapters 14 and 15]{Sumeet_book}. An ebit is a Bell state $\ket{\psi_{00}}_{AB}$ shared between Alice and Bob. For any $\Phi$, the two-way capacities satisfy
\bb\label{relation_2waycap}
Q_2(\Phi)\le K(\Phi)\,.
\ee
Indeed, by recalling that Alice and Bob can freely send an infinite amount of bits to each other, an ebit can generate a secret-key bit, thanks to E91 protocol~\cite{Ekert91}, and hence $Q_2(\Phi)\le K(\Phi)$. The two-way quantum capacity $Q_2(\Phi)$ and the secret-key capacity $K(\Phi)$ are collectively called the \emph{two-way capacities of $\Phi$}.

In practice, Alice has access to a limited budget ($N_s$) of energy to produce each input signal. Here, by definition, the energy of a signal initialised in a state $\rho$ is equal to its mean photon number $\Tr[\rho\, a^\dagger a]$. Fixed $N_s>0$, the energy-constrained (EC) two-way capacities $Q_2(\Phi,N_s)$ and $K(\Phi,N_s)$ are defined as above but the maximisation of the rate is restricted to the strategies such that the average photon number less or equal to $N_s$. In other words, $N_s$ is the maximum allowed average photon number of the input signals to the channel $\Phi$.
In addition note that the generalisation of~\ref{relation_2waycap} to the EC case  holds, i.e.
\bb\label{relation_2waycapEC}
Q_2(\Phi,N_s) \le K(\Phi,N_s) \,,
\ee
and that any EC capacity is upper bounded by the corresponding unconstrained capacity and tends to it in the limit $N_s\rightarrow\infty$.

\subsection{Entanglement distillation}
The goal of an entanglement distillation protocol is to turn a large number $n$ of copies of a bipartite entangled state $\rho_{AB}\in\mathfrak{S}(\HH_A\otimes\HH_{B})$ shared between Alice and Bob into a smaller number $m$ of ebits by LOCCs (local operations and classical communication). The yield of an entanglement distillation protocol is defined by the ratio $m/n$. The two-way distillable entanglement $E_d(\rho_{AB})\,$ of $\rho_{AB}$ is defined as the maximum yields over all the possible entanglement distillation protocols~\cite{reviewEDP_dur}~\cite[Chapter 8]{Sumeet_book}. The state $\rho_{AB}$ is said to be \emph{distillable} if $E_d(\rho_{AB})>0$. The coherent information of $\rho_{AB}$ is defined by 
\bb
    I_{\text{c}}(\rho_{AB})\coloneqq S(\Tr_A\rho_{AB})-S(\rho_{AB}) 
\ee
and it is a yield achievable by an entanglement distillation protocol which requires classical communication only from Alice to Bob~\cite{devetak2005}.
By exchanging the roles of Alice and Bob in such an entanglement distillation protocol, the reverse coherent information of $\rho_{AB}$, which is defined by 
\bb
    I_{\text{rc}}(\rho_{AB})\coloneqq S(\Tr_B\rho_{AB})-S(\rho_{AB})\,,
\ee
is a yield achievable by an entanglement distillation protocol which only requires classical communication only from Bob to Alice~\cite{devetak2005}. In particular, the following inequality, known as \emph{hashing inequality}, holds:
\begin{equation}\label{hashing_ineq}
    E_d(\rho_{AB})\ge \max\{I_{\text{c}}(\rho_{AB})\,, I_{\text{rc}}(\rho_{AB})\}\,.
\end{equation}
Let us briefly link the notions of distillable entanglement $E_d$ and two-way quantum capacity $Q_2(\Phi,N_s)$. Suppose that Alice produces $n$ copies of a state $\rho_{AA'}$ such that the mean photon number of the half $A'$ is less or equal to $N_s$. Then, she uses $n$ times the channel $\Phi$ to send the halves $A'$, which satisfy the energy constraint, to Bob. Hence, $n$ copies of $\Id_{A}\otimes\Phi(\rho_{AA'})$ are shared between Alice and Bob and can be used to generate ebits by means of an entanglement distillation protocol. Consequently, it holds that
\bb\label{link_D2_D}
    Q_2(\Phi,N_s)\ge E_d\left(\Id_{A}\otimes\Phi(\rho_{AA'}) \right)\,
\ee
for all $N_s\ge0$ and all $\rho_{A'A}$ satisfying $\Tr[a^\dagger a\, \rho_{A'A}]\le N_s$, where $a$ denotes the annihilation operator on $A'$.

If a bipartite state $\rho_{AB}$ is such that the hashing inequality is trivial (i.e.~the right-hand side of~\eqref{hashing_ineq} is negative), in order to obtain a non-trivial lower bound on $E_d(\rho_{AB})$, one can adopt a sufficiently large number of iterations of a \emph{recurrence protocol} on $\rho_{AB}$ prior to apply the hashing inequality. In the context of entanglement distillation, the goal of a recurrence protocol is to transform a certain number of copies of the state $\rho_{AB}$ into fewer copies of another state $\rho'_{AB}$ such that $\bra{\psi_{00}}\rho'_{AB}\ket{\psi_{00}}>\bra{\psi_{00}}\rho_{AB}\ket{\psi_{00}}$~\cite{Bennett-error-correction,Bennett-distillation-mixed,reviewEDP_dur}. 
Examples of recurrence protocols for qubits can be found in~\cite{Bennett-distillation-mixed,DEJMPS,DNMV}, and their generalisations to the case of qudits in~\cite{Horodecki1999,Alber_2001,Dist-Number-Theory}. In the present paper we will exploit the recently introduced P1-or-P2 recurrence protocol~\cite{p1orp2}.   
To achieve a nonzero yield, one may adopt a suitable number of iterations of a recurrence protocol and then apply the hashing or breeding protocol~\cite{Bennett-error-correction,Bennett-distillation-mixed}.
The latter protocols, which exploit only one-way classical communication, achieves the yield of the hashing inequality in~\eqref{hashing_ineq}.
Improvements of the hashing and breeding protocols, which exploit two-way classical communication, have been provided in~\cite{Improvement-Hashing}: the two-way distillable entanglement of a convex combination of Bell states $\rho_{AB}\coloneqq \sum_{ij=0}^1\alpha_{ij}\ketbra{\psi_{ij}}$ is lower bounded by
\bb\label{improv_hashing}
   &Y(\alpha_{00},\alpha_{01},\alpha_{10},\alpha_{11})\coloneqq \max\left(0, 1-H(\{\alpha_{ij}\})+\frac{1}{2}(\alpha_{00}+\alpha_{10})(\alpha_{11}+\alpha_{01})\left[H_2\left(\frac{\alpha_{00}}{\alpha_{00}+\alpha_{10}}\right)+H_2\left(\frac{\alpha_{11}}{\alpha_{01}+\alpha_{11}}\right)\right]\right),
\ee
with $H(\{\alpha_{ij}\})\coloneqq -\sum_{m,n=0}^{1}\alpha_{mn}\log_2\alpha_{mn}$ being the Shannon entropy and $H_2(x)\coloneqq -x\log_2x-(1-x)\log_2(1-x)$ for all $x\in[0,1]$ being the binary entropy. The yield in~\eqref{improv_hashing} is larger than the yield achieved by the hashing protocol, which is $I_{\text{c}}\left(\sum_{ij=0}^1\alpha_{ij}\ketbra{\psi_{ij}}\right)=1-H(\{\alpha_{ij}\})$. Protocols with larger yiels than~\eqref{improv_hashing} may be obtained by exploiting the numerical methods introduced in~\cite{Numerical-Improvement-Hashing}.

Now, let us briefly review the definition, the relevant properties, and the known bounds on the two-way capacities of phase-insensitive bosonic Gaussian channels, namely thermal attenuator, thermal amplifier, and additive Gaussian noise.
\subsection{Thermal attenuator}
Let $\HH_S$ and $\HH_E$ be single-mode systems and let $a$ and $b$ denote their annihilation operators, respectively.
For all $\lambda\in[0,1]$ and $\nu\ge0$, a thermal attenuator $\mathcal{E}_{\lambda,\nu}:\mathfrak{S}(\HH_S)\to\mathfrak{S}(\HH_S)$ is a quantum channel defined by 
\bb\label{def_therm}
    \mathcal{E}_{\lambda,\nu}(\rho)\coloneqq\Tr_E\left[U_\lambda^{SE}\big(\rho^S \otimes\tau_\nu^E \big) {U_\lambda^{SE}}^\dagger\right]\,,
\ee
where $U_\lambda^{SE}$ denotes the unitary operator associated with a beam splitter of transmissivity $\lambda$, i.e.
\bb
	U_{\lambda}^{S E}\coloneqq\exp\left[\arccos\sqrt{\lambda}\left(a^\dagger b-a\, b^\dagger\right)\right]\,,
\ee
and $\tau_\nu\in\mathfrak{S}(\HH_E)$ denotes the thermal state with mean photon number equal to $\nu$. {The beam splitter unitary can be expressed via the following disentangling formula~\cite[Appendix 5]{BARNETT-RADMORE}
\bb
    U_{\lambda}^{SE}=e^{-\sqrt{\frac{1-\lambda}{\lambda}}ab^\dagger}e^{ \frac{1}{2}\ln\lambda\,\left(a^\dagger a -b^\dagger b\right) }e^{\sqrt{\frac{1-\lambda}{\lambda}}a^\dagger b}\,.
\ee}
By writing the quadrature vector $\mathbf{\hat{R}}$ with respect the ordering $(S,E)$, it can be shown that 
\bb\label{transf_r}
\left(U_\lambda^{SE}\right)^\dagger \mathbf{\hat{R}}\, U_{\lambda}^{SE}=S_\lambda\, \mathbf{\hat{R}}\,,
\ee
where 
\bb
S_\lambda\coloneqq	\begin{pmatrix}
		\sqrt{\lambda}\,\mathbb{1}_2 & \sqrt{1-\lambda}\,\mathbb{1}_2 \\
		-\sqrt{1-\lambda}\,\mathbb{1}_2 &\, \sqrt{\lambda}\,\mathbb{1}_2
	\end{pmatrix}\,.
\ee
This implies that for all $\sigma_{SE}\in\mathfrak{S}(\HH_S\otimes H_E)$ it holds that
\bb\label{relation_S}
&\mathbf{m}\left(U^{SE}_\lambda\sigma_{SE} \left(U_\lambda^{SE}\right)^\dagger\right)=S_\lambda\, \mathbf{m}(\sigma_{SE})\,,\\
&V\left(U^{SE}_\lambda\sigma_{SE} \left(U_\lambda^{SE}\right)^\dagger \right)=S_\lambda\,V(\sigma_{SE})\,S_\lambda^{\intercal}\,.
\ee
{In terms of the annihilation operators $a$ and $b$, the transformation in~\eqref{transf_r} reads
\bb
    \left(U_\lambda^{SE}\right)^\dagger a\, U_{\lambda}^{SE}&=\sqrt{\lambda}\,a+\sqrt{1-\lambda}\,b\,,\\
    U_\lambda^{SE} a\, \left(U_{\lambda}^{SE}\right)^\dagger&=\sqrt{\lambda}\,a-\sqrt{1-\lambda}\,b\,,\\
    \left(U_\lambda^{SE}\right)^\dagger b\, U_{\lambda}^{SE}&=-\sqrt{1-\lambda}\,a+\sqrt{\lambda}\,b\,,\\
    U_\lambda^{SE} b\, \left(U_{\lambda}^{SE}\right)^\dagger&=\sqrt{1-\lambda}\,a+\sqrt{\lambda}\,b\,.
\ee}It can be shown that for any single-mode state $\rho$ it holds that 
\bb\label{moment_therm_att}
&\mathbf{m}\left(\mathcal{E}_{\lambda,\nu}(\rho)\right)=\sqrt{
\lambda}\, \mathbf{m}(\rho)\,,\\
&V\left(\mathcal{E}_{\lambda,\nu}(\rho) \right)=\lambda\, V(\rho)+(1-\lambda)(2\nu+1)\mathbb{1}_2\,,
\ee
and, in terms of the characteristic function, for all $\mathbf{r}\in\mathbb{R}^2$ it holds that 
\bb\label{caract_att}
\chi_{\mathcal{E}_{\lambda,\nu}(\rho)}(\mathbf{r})=\chi_\rho(\sqrt{\lambda}\mathbf{r})e^{-\frac{1}{4}(1-\lambda)(2\nu+1)|\mathbf{r}|^2}\,.
\ee
By exploiting~\eqref{caract_att} and the fact that quantum states and characteristic functions are in one-to-one correspondence, for all $\lambda_1,\lambda_2\in[0,1]$ and $\nu\ge0$ the following composition rule holds:\bb\label{composition_them}
\mathcal{E}_{\lambda_1,\nu}\circ\mathcal{E}_{\lambda_2,\nu}=\mathcal{E}_{\lambda_1\lambda_2,\nu}\,.
\ee
{In Theorem~\ref{kraus_comp_thm} we will provide a simple Kraus representation of the thermal attenuator.}

\subsubsection{Bounds on two-way capacities of the thermal attenuator}
The best known upper bound on the two-way capacities of the thermal attenuator, shown by Pirandola-Laurenza-Ottaviani-Banchi (PLOB)~\cite{PLOB}, is
\bb\label{PLOB_Q2}
    K(\mathcal{E}_{\lambda,\nu})&\le \begin{cases}
-h(\nu)-\log_2[(1-\lambda)\lambda^\nu], & \text{if $\lambda\in(\frac{\nu}{\nu+1}, 1]$,} \\
0, & \text{otherwise}
\end{cases}
\ee
where 
\bb\label{bos_ent}
h(\nu)\coloneqq(\nu+1)\log_2(\nu+1)-\nu\log_2\nu\,
\ee
is the so-called bosonic entropy. The parameter region in which such an upper bound vanishes coincides with the parameter region in which the thermal attenuator $\mathcal{E}_{\lambda,\nu}$ is entanglement breaking, i.e.~$\nu\ge0$ and $\lambda\in[0,\frac{\nu}{\nu+1}]$~\cite{Ent_breaking_Gaussian, Holevo-EB}.
The best known lower bound (before our work) on $Q_2(\mathcal{E}_{\lambda,\nu})$ is given by~\cite{Pirandola2009}
\bb\label{lowQ2}
    Q_2(\mathcal{E}_{\lambda,\nu})\ge\max\{0,-h(\nu)-\log_2(1-\lambda)\}\,.
\ee
Although this is also a lower bound on $K(\mathcal{E}_{\lambda,\nu})$, it is not the best among those currently known. Indeed, an improved lower bound on $K(\mathcal{E}_{\lambda,\nu})$ has been shown by Ottaviani et al.~\cite{Ottaviani_new_lower}. In the energy-constrained case, the best known lower bound (before our work) on the EC two-way capacities of the thermal attenuator has been found by Noh-Pirandola-Jiang (NPJ)~\cite{Noh2020}, while the best known upper bound is --- depending on the parameters $\lambda$, $\nu$, and $N_s$ --- the bound found by Davis-Shirokov-Wilde (DSW)~\cite{Davis2018} or the PLOB bound in~\eqref{PLOB_Q2}.

The lower bound in~\eqref{lowQ2} on the two-way capacities of the thermal attenuator can be proved first by applying~\eqref{link_D2_D} with the choice $\rho_{AA'}=\ketbra{\Psi_{N_s}}_{AA'}$, where $\ket{\Psi_{N_s}}$ is the two-mode squeezed vacuum states with local mean photon number equal to $N_s$ defined in~\eqref{two_mode_sq},
second by applying the hashing inequality in~\eqref{hashing_ineq}, and finally by proving that the reverse coherent information satisfies
\bb\label{proof_lower}
   & \lim\limits_{N_s\rightarrow\infty}I_{\text{rc}}\left(\Id_{A}\otimes\mathcal{E}_{\lambda,\nu}(\ketbra{\Psi_{N_s}})\right) = -h(\nu)-\log_2(1-\lambda)\,.
\ee
Analogously, the coherent information  $$I_{\text{c}}\left(\Id_{A}\otimes\mathcal{E}_{\lambda,\nu}(\ketbra{\Psi_{N_s}})\right)$$ and the reverse coherent information $$I_{\text{rc}}\left(\Id_{A}\otimes\mathcal{E}_{\lambda,\nu}(\ketbra{\Psi_{N_s}})\right)$$ are lower bounds on the EC two-way capacities of the thermal attenuator $\mathcal{E}_{\lambda,\nu}$ with energy constraint equal to $N_s$:
\bb
    Q_2(\mathcal{E}_{\lambda,\nu},N_s)&\ge  \max\left\{ I_{\text{c}}\left(\Id_{A}\otimes\mathcal{E}_{\lambda,\nu}(\ketbra{\Psi_{N_s}})\right),\, I_{\text{rc}}\left(\Id_{A}\otimes\mathcal{E}_{\lambda,\nu}(\ketbra{\Psi_{N_s}})\right) \right\}\,\\
    &= \begin{cases}
I_{\text{c}}\left(\Id_{A}\otimes\mathcal{E}_{\lambda,\nu}(\ketbra{\Psi_{N_s}})\right), & \text{if $N_s\le \nu$,} \\
I_{\text{rc}}\left(\Id_{A}\otimes\mathcal{E}_{\lambda,\nu}(\ketbra{\Psi_{N_s}})\right), & \text{otherwise.}
\end{cases}
\ee
It holds that~\cite{holwer,PLOB,Pirandola2009,Noh2020} 
\bb\label{EC_coh_therm_att}
    I_{\text{c}}\left(\Id_{A}\otimes\mathcal{E}_{\lambda,\nu}(\ketbra{\Psi_{N_s}})\right)&=h\left(\lambda N_s +(1-\lambda)\nu \right)-h\left(\frac{D+(1-\lambda)(N_s-\nu)-1}{2}\right)-h\left(\frac{D-(1-\lambda)(N_s-\nu)-1}{2}\right)\,,\\
    I_{\text{rc}}\left(\Id_{A}\otimes\mathcal{E}_{\lambda,\nu}(\ketbra{\Psi_{N_s}})\right)&=h(N_s)-h\left(\frac{D+(1-\lambda)(N_s-\nu)-1}{2}\right)-h\left(\frac{D-(1-\lambda)(N_s-\nu)-1}{2}\right)\,,
\ee
where $D\coloneqq \sqrt{\left( (1+\lambda)N_s+(1-\lambda)\nu +1 \right)^2 -4\lambda N_s(N_s+1)}$. The NPJ lower bound, proved by mixing forward (coherent information) and backward (reverse coherent information) strategies, is~\cite{Noh2020}
\bb\label{npj_bound_therm}
Q_2(\mathcal{E}_{\lambda,\nu},N_s)\ge \sup_{\substack{x\in[0,1],\,N_1,N_2\ge0\\ xN_1+(1-x)N_2=N_s}} \left[  x\,I_{\text{c}}\left(\Id_{A}\otimes\mathcal{E}_{\lambda,\nu}(\ketbra{\Psi_{N_1}})\right)+(1-x)I_{\text{rc}}\left(\Id_{A}\otimes\mathcal{E}_{\lambda,\nu}(\ketbra{\Psi_{N_2}})\right) \right]\,.
\ee
Fixed $\lambda$ and $\nu$, if the energy constraint $N_s$ is sufficiently large, the NPJ lower bound is equal to the reverse coherent information bound (i.e.~the optimal values of the supremum problem in~\ref{npj_bound_therm} are $x=0$, $N_1=0$, and $N_2=N_s$).

\subsection{Thermal amplifier}
Let $\HH_S$ and $\HH_E$ be single-mode systems and let $a$ and $b$ denote their annihilation operators, respectively.
For all $g\ge 1$ and $\nu\ge0$, a thermal amplifier $\Phi_{g,\nu}:\mathfrak{S}(\HH_S)\to\mathfrak{S}(\HH_S)$ is a quantum channel defined by 
\bb\label{def_ampl}
    \Phi_{g,\nu}(\rho)\coloneqq\Tr_E\left[U_g^{SE}\big(\rho^S\otimes\tau_\nu^E\big) {U_g^{SE}}^\dagger\right]\,,
\ee
where $U_g^{SE}$ denotes the unitary operator associated with two-mode squeezing of parameter $g$, i.e.
\bb\label{def_unitary_squeez}
	U_{g}^{S E}\coloneqq\exp\left[\arccosh\sqrt{g}\left(a^\dagger b^\dagger-a\, b\right)\right]\,.
\ee
{The two-mode squeezing unitary can be expressed via the following disentangling formula~\cite[Appendix 5]{BARNETT-RADMORE}
\bb
    U_{g}^{SE}=e^{\sqrt{\frac{g-1}{g}}a^\dagger b^\dagger}e^{ \frac{1}{2}\ln\left(\frac{1}{g}\right)\,\left(a^\dagger a -b^\dagger b+1\right) }e^{-\sqrt{\frac{g-1}{g}}a b}\,.
\ee}By writing the quadrature vector $\mathbf{\hat{R}}$ with respect the ordering $(S,E)$, it can be shown that 
\bb\label{transf_r_amp}
\left(U_g^{SE}\right)^\dagger \mathbf{\hat{R}}\, U_{g}^{SE}=S_g\, \mathbf{\hat{R}}\,,
\ee
where 
\bb
S_g\coloneqq	\begin{pmatrix}
		\sqrt{g}\,\mathbb{1}_2\, & \,\sqrt{g-1}\,\sigma_z \\
		\sqrt{g-1}\,\sigma_z\, &\, \sqrt{g}\,\mathbb{1}_2
	\end{pmatrix}\,.
\ee
This implies that for all $\sigma_{SE}\in\mathfrak{S}(\HH_S\otimes H_E)$ it holds that
\bb\label{relation_S_amp}
&\mathbf{m}\left(U^{SE}_g\sigma_{SE} \left(U_g^{SE}\right)^\dagger\right)=S_g\, \mathbf{m}(\sigma_{SE})\,,\\
&V\left(U^{SE}_g\sigma_{SE} \left(U_g^{SE}\right)^\dagger \right)=S_g\,V(\sigma_{SE})\,S_g^{\intercal}\,.
\ee
{In terms of the annihilation operators $a$ and $b$, the transformation in~\eqref{transf_r_amp} reads
\bb
    \left(U_g^{SE}\right)^\dagger a\, U_{g}^{SE}&=\sqrt{g}\,a+\sqrt{g-1}\,b^\dagger\,,\\
    U_g^{SE} a\, \left(U_{g}^{SE}\right)^\dagger&=\sqrt{g}\,a-\sqrt{g-1}\,b^\dagger\,,\\
    \left(U_g^{SE}\right)^\dagger b\, U_{g}^{SE}&=\sqrt{g-1}\,a^\dagger+\sqrt{g}\,b\,,\\
    U_g^{SE} b\, \left(U_{g}^{SE}\right)^\dagger&=-\sqrt{g-1}\,a^\dagger+\sqrt{g}\,b\,.
\ee}
It can be shown that for any single-mode state $\rho$ it holds that 
\bb\label{moment_therm_amp}
&\mathbf{m}\left(\Phi_{g,\nu}(\rho)\right)=\sqrt{
g}\, \mathbf{m}(\rho)\,,\\
&V\left(\Phi_{g,\nu}(\rho) \right)=g\, V(\rho)+(g-1)(2\nu+1)\mathbb{1}_2\,,
\ee
and, in terms of the characteristic function, for all $\mathbf{r}\in\mathbb{R}^2$ it holds that 
\bb\label{caract_amp}
\chi_{\Phi_{g,\nu}(\rho)}(\mathbf{r})=\chi_\rho(\sqrt{g}\mathbf{r})e^{-\frac{1}{4}(g-1)(2\nu+1)|\mathbf{r}|^2}\,.
\ee
By exploiting~\eqref{caract_amp} and the fact that quantum states and characteristic functions are in one-to-one correspondence, for all $g_1,g_2\ge1$ and $\nu\ge0$ the following composition rule holds:\bb\label{comp_rule_amp}
\Phi_{g_1,\nu}\circ\Phi_{g_2,\nu}=\Phi_{g_1g_2,\nu}\,.
\ee
{In Theorem~\ref{kraus_comp_thm} we will provide a simple Kraus representation of the thermal amplifier.}
\subsubsection{Bounds on two-way capacities of the thermal amplifier}
The best known upper bound on the two-way capacities of the thermal amplifier, shown by PLOB~\cite{PLOB}, is
\bb\label{PLOB_amp}
    K(\Phi_{g,\nu})&\le \begin{cases}
-h(\nu)+\log_2\left(\frac{g^{\nu+1}}{g-1}\right), & \text{if $g\in[1, 1+\frac{1}{\nu})$,} \\
0, & \text{otherwise}
\end{cases}
\ee
where $h(\nu)$ is the bosonic entropy defined in~\eqref{bos_ent}. The parameter region in which such an upper bound vanishes coincides with the parameter region in which the thermal amplifier $\Phi_{g,\nu}$ is entanglement breaking, i.e.~$\nu\ge0$ and $g\ge 1+\frac{1}{\nu}$~\cite{Ent_breaking_Gaussian, Holevo-EB}. The best known lower bound (before our work) on $Q_2(\Phi_{\lambda,\nu})$ is given by~\cite{Pirandola2009}
\bb\label{lowQ2_amp}
    Q_2(\Phi_{g,\nu})\ge\max\left\{0,-h(\nu)+\log_2\left(\frac{g}{g-1}\right)\right\}\,,
\ee
which can be proved, analogously as it has been done in~\eqref{proof_lower}, by showing that the coherent information satisfies
\bb\label{proof_lower_ampl}
   & \lim\limits_{N_s\rightarrow\infty}I_{\text{c}}\left(\Id_{A}\otimes\Phi_{g,\nu}(\ketbra{\Psi_{N_s}})\right)  = -h(\nu)+\log_2\left(\frac{g}{g-1}\right)\,.
\ee
The best known lower bound on the secret-key capacity $K(\Phi_{g,\nu})$ has been shown by Wong-Ottaviani-Guo-Pirandola (WOGP)~\cite{Wang_Q2_amplifier}. In the energy-constrained scenario, the best known lower bound is the NPJ bound~\cite{Noh2020}, which is given by
\bb\label{npj_bound_amp}
Q_2(\Phi_{g,\nu},N_s)\ge \sup_{x\in[0,1]}  x\,I_{\text{c}}\left(\Id_{A}\otimes\Phi_{g,\nu}(\ketbra{\Psi_{\frac{N_s}{x}}})\right) \,,
\ee
where~\cite{holwer,PLOB,Pirandola2009,Noh2020} 
\bb
    &I_{\text{c}}\left(\Id_{A}\otimes\Phi_{g,\nu}(\ketbra{\Psi_{N_s}})\right)\\&=h\left(g N_s +(g-1)(\nu+1) \right)-h\left(\frac{D'+(g-1)(N_s+\nu+1)-1}{2}\right)-h\left(\frac{D'-(g-1)(N_s+\nu+1)-1}{2}\right)\,,
\ee
with $D'\coloneqq \sqrt{\left( (g+1)N_s+(g-1)(\nu+1) +1 \right)^2 -4g\,N_s(N_s+1)}$. Fixed $g$ and $\nu$, if the energy constraint $N_s$ is sufficiently large, the NPJ lower bound is equal to the coherent information bound (i.e.~the optimal value of the supremum problem in~\ref{npj_bound_amp} is $x=1$).

\subsection{Additive Gaussian noise}
Let $\HH_S$ be a single-mode system and let $\{D_\mathbf{r}\}_{\mathbf{r}\in\mathbb{R}^2}$ be its dispacement operators. For all $\xi\ge0$, the additive Gaussian noise $\Lambda_\xi:\mathfrak{S}(\HH_S)\to\mathfrak{S}(\HH_S)$ is a quantum channel defined by
\begin{equation}\label{def_add}
\Lambda_\xi(\rho)\coloneqq\frac{1}{2\pi\xi}\int_{\mathbb R^{2}} \mathrm{d}^2 {\mathbf{r}}\, e^{-\frac{1}{2\xi}\mathbf{r}^{\intercal}\mathbf{r}}  D_{\mathbf{r}}\rho  D_{\mathbf{r}}^\dagger\,.  
\end{equation} 
By using that $\mathbf{\hat R}= (\hat{x},\hat{p})$, $a=\frac{\hat{x}+\hat{p}}{\sqrt{2}}$, and by defining  $\mathbf{r}\coloneqq(x,p)^{\text{T}}$, $z\coloneqq \frac{x+ip}{\sqrt{2}}$, and
\bb
    D(z)\coloneqq\exp{\left[z a^\dagger-z^\ast a\right] }=D_{-\mathbf{r}}\,,
\ee
the additive Gaussian noise can be expressed in the following equivalent form: 
\bb
    \Lambda_\xi(\rho)=\frac{1}{\pi\xi}\int_{\mathbb C} \mathrm{d}^2 {z}\, e^{-\frac{|z|^2}{\xi}}  D(z)\,\rho\,  D(z)^\dagger\,,
\ee
where we have used that $\mathrm{d}^2 {\mathbf{r}}=\mathrm{d}x\,\mathrm{d}p=\frac{\mathrm{d}\text{Re}(z)\,\mathrm{d}\text{Im}(z)}{2}=\frac{\mathrm{d}^2z}{{2}}$ and we have performed the integral variable substitution $z\rightarrow -z$.
It can be shown that for all single-mode states $\rho$ it holds that
\bb\label{relation_add}
&\mathbf{m}\left(\Lambda_\xi(\rho)\right)=\mathbf{m}(\rho)\,,\\
&V\left(\Lambda_\xi(\rho) \right)=V(\rho)+2\xi\,\mathbb{1}_2\,.
\ee
and, in terms of the characteristic function, for all $\mathbf{r}\in\mathbb{R}^2$ it holds that 
\bb\label{caract_add}
\chi_{\Lambda_{\xi}(\rho)}(\mathbf{r})=\chi_\rho(\mathbf{r})e^{-\frac{1}{2}\xi|\mathbf{r}|^2}\,.
\ee
{In Theorem~\ref{kraus_comp_thm} we will provide a simple Kraus representation of the additive Gaussian noise.}
\subsubsection{Additive Gaussian noise as the strong limit of thermal attenuator or thermal amplifier}
For completeness, let us remark that the Additive Gaussian noise $\Lambda_\xi$ is the strong limit of the thermal attenuator $\mathcal{E}_{1-\frac{\xi}{\nu},\nu}$ and thermal amplifier $\Phi_{1+\frac{\xi}{\nu},\nu}$ for $\nu\rightarrow\infty$, i.e.~it holds that \bb \lim\limits_{\nu\rightarrow\infty}\|\mathcal{E}_{1-\frac{\xi}{\nu},\nu}(\rho)-\Lambda_\xi(\rho)\|_1&=0\,,\\\lim\limits_{\nu\rightarrow\infty}\|\Phi_{1+\frac{\xi}{\nu},\nu}(\rho)-\Lambda_\xi(\rho)\|_1&=0\,,\ee
for any single-mode state $\rho$.
Indeed,~\eqref{caract_att},~\eqref{caract_amp}, and~\eqref{caract_add} imply that for any single-mode state $\rho$ and any $\mathbf{r}\in\mathbb{R}^2$ it holds that
\bb\lim\limits_{\nu\rightarrow\infty}\chi_{\mathcal{E}_{1-\frac{\xi}{\nu},\nu}(\rho)}(\mathbf{r})&=\chi_{\Lambda_\xi(\rho)}(\mathbf{r})\,,\\ \lim\limits_{\nu\rightarrow\infty}\chi_{\Phi_{1+\frac{\xi}{\nu},\nu}(\rho)}(\mathbf{r})&=\chi_{\Lambda_\xi(\rho)}(\mathbf{r})\,.\ee
Consequently, by exploiting the fact that a sequence of states $\{\sigma_k\}_{k\in\N}\subseteq\mathfrak{S}(L^2(\mathbb{R}))$ converges in trace norm to a quantum state $\sigma\in\mathfrak{S}(L^2(\mathbb{R}))$ if and only if the sequence of characteristic functions $ \{\chi_{\sigma_k} (\mathbf{r})\}_{k\in\N}$ converges pointwise to the characteristic function $\chi_\sigma(\mathbf{r})$~\cite[Theorem 2]{Cushen1971}, 
the thermal attenuator $\mathcal{E}_{1-\frac{\xi}{\nu},\nu}$ and thermal amplifier $\Phi_{1+\frac{\xi}{\nu},\nu}$ strongly converge to the additive Gaussian noise $\Lambda_\xi$ for $\nu\rightarrow\infty$.

\subsubsection{Bounds on two-way capacities of the additive Gaussian noise}
The best known upper bound on the two-way capacities of the additive Gaussian noise, shown by PLOB~\cite{PLOB}, is
\bb\label{PLOB_add}
K(\Lambda_\xi)\le  \begin{cases}
\frac{\xi-1}{\ln2}-\log_2(\xi), & \text{if $\xi< 1$,} \\
0, & \text{otherwise}
\end{cases}
\ee
where $h(\nu)$ is the bosonic entropy defined in~\eqref{bos_ent}. The parameter region in which such an upper bound vanishes coincides with the parameter region in which the Additive Gaussian noise $\Lambda_\xi$ is entanglement breaking, i.e.~$\xi\ge 1$~\cite{Ent_breaking_Gaussian, Holevo-EB}.
The best known lower bound (before our work) on $Q_2(\Lambda_\xi)$ is given by~\cite{Pirandola2009}
\bb\label{lowQ2_add}
Q_2(\Lambda_\xi)\ge \max\{0,-\log_2(e\,\xi)\}\,,
\ee
which can be proved, analogously as it has been done in~\eqref{proof_lower}, by showing that the coherent information satisfies
\bb\label{proof_lower_additive}
   \lim\limits_{N_s\rightarrow\infty}I_{\text{c}}\left(\Id_{A}\otimes\Lambda_{\xi}(\ketbra{\Psi_{N_s}}_{AA'})\right) =\log_2(e\,\xi)\,.
\ee
In the energy-constrained scenario, the best known lower bound is the NPJ bound~\cite{Noh2020}, which is given by
\bb\label{npj_bound_add}
Q_2(\Lambda_\xi,N_s)\ge \sup_{x\in[0,1]}  x\,I_{\text{c}}\left(\Id_{A}\otimes\Lambda_\xi(\ketbra{\Psi_{\frac{N_s}{x}}})\right) \,,
\ee
where~\cite{holwer,PLOB,Pirandola2009,Noh2020} 
\bb
    &I_{\text{c}}\left(\Id_{A}\otimes\Lambda_\xi(\ketbra{\Psi_{N_s}})\right)  =h\left(N_s +\xi \right)-h\left(\frac{D''+\xi-1}{2}\right)-h\left(\frac{D''-\xi-1}{2}\right)\,,
\ee
with $D''\coloneqq \sqrt{\left( 2N_s+\xi +1 \right)^2 -4N_s(N_s+1)}$. Fixed $\xi$, if the energy constraint $N_s$ is sufficiently large, the NPJ lower bound is equal to the coherent information bound (i.e.~the optimal value of the supremum problem in~\ref{npj_bound_add} is $x=1$).

\section{Action of phase-insensitive bosonic Gaussian channels on generic operators}
{In this section we establish properties of the channel composition between pure amplifier channel and pure loss channel.}
\begin{definition}\label{def1}
For all $\lambda\in[0,1]$ and $g\ge1$ let us define the channel $\mathcal{N}_{g,\lambda}$ as the composition between pure amplifier channel $\Phi_{g,0}$ and pure loss channel $\mathcal{E}_{\lambda,0}$, i.e.
\bb\label{def_comp}
    \mathcal{N}_{g,\lambda}\coloneqq\Phi_{g,0}\circ \mathcal{E}_{\lambda,0}\,.
\ee
\end{definition}
\begin{lemma}\label{lemma_eb_comp}
    The channel $\mathcal{N}_{g,\lambda}$ is entanglement breaking if and only if $(1-\lambda)g\ge 1$.
\end{lemma}
\begin{proof}
    First, let us determine the parameter region of $g$ and $\lambda$ where the channel $\mathcal{N}_{g,\lambda}$ is entanglement breaking. Since $\mathcal{N}_{g,\lambda}$ is a Gaussian channel, we can apply Lemma~\ref{holevo_lemma_eb_gauss}. By using~\eqref{transf_caract}, one can show that $\mathcal{N}_{g,\lambda}$ transforms the first moment and the covariance matrix as
\bb\label{relation_comp}
    &\mathbf{m}\left(\mathcal{N}_{g,\lambda}(\rho)\right)=\sqrt{g\lambda}\,\mathbf{m}(\rho)\,,\\
    &V\left(\mathcal{N}_{g,\lambda}(\rho) \right)=g\lambda\, V(\rho)+(2g-1-g\lambda)\,\mathbb{1}_2\,,
\ee
for all quantum states $\rho$. Hence, Lemma~\ref{holevo_lemma_eb_gauss} establishes that $\Phi$ is entanglement breaking if and only if there exists $\alpha,\gamma\in\mathbb{R}^{2\times 2}\text{ with }\alpha\ge i\, \Omega_{1}\,\text{ and }\,\gamma\ge i \lambda g\,\Omega_1$ such that
\bb\label{conditionEB_comp}
(2g-1-g\lambda)\,\mathbb{1}_2=\alpha+\gamma\,.
\ee
The condition in~\eqref{conditionEB_comp} is equivalent to
\bb\label{conditionEB_comp2}
    (2g-1-g\lambda)\,\mathbb{1}_2\ge i\,(1+\lambda g)\,\Omega_1\,.
\ee
Indeed, if the condition in~\eqref{conditionEB_comp} is satisfied, then $(2g-1-g\lambda)\,\mathbb{1}_2=\alpha+\gamma\ge i\,(1+\lambda g)\,\Omega_1$, i.e.~also the condition in~\eqref{conditionEB_comp2} is satisfied. Conversely, assume that the condition in~\eqref{conditionEB_comp2} is satisfied. Then, the fact that 
\bb\label{fact_diag}
x\,\mathbb{1}_2\ge i\, \Omega_1 \quad\text{if and only if}\quad x\ge1\,,
\ee
implies that $(1-\lambda)g\ge1$. Consequently, by choosing $\alpha\coloneqq\mathbb{1}_2$ and $\gamma\coloneqq (2g-2-g\lambda)\mathbb{1}_2$ and by using~\eqref{fact_diag}, it holds that the condition in~\eqref{conditionEB_comp} is satisfied with $\alpha\ge i\, \Omega_{1}\,\text{ and }\,\gamma\ge i \lambda g\,\Omega_1$. By exploiting~\eqref{fact_diag}, we deduce that $\mathcal{N}_{g,\lambda}$ is entanglement breaking if and only if $(1-\lambda)g\ge 1$. 
\end{proof}
\begin{lemma}\label{lemma_comp_bos}
    Let $\nu\ge 0$, $\lambda\in[0,1]$, $g\ge 1$, and $\xi\ge0$. The thermal attenuator $\mathcal{E}_{\lambda,\nu}$, the thermal amplifier $\Phi_{g,\nu}$, and the additive Gaussian noise $\Lambda_\xi$ can be expressed in terms of the composition between pure amplifier channel and pure loss channel as
    \bb\label{channels_as_composition}
        \mathcal{E}_{\lambda,\nu}&=\mathcal{N}_{1+(1-\lambda)\nu\,,\,\frac{\lambda}{1+(1-\lambda)\nu}}\,\,,\\
        \Phi_{g,\nu}&=\mathcal{N}_{g+(g-1)\nu\,,\,\frac{g}{g+(g-1)\nu}}\,\,,\\
        \Lambda_{\xi}&=\mathcal{N}_{1+\xi\,,\,\frac{1}{1+\xi}}\,\,.
    \ee
\end{lemma}
\begin{proof}
    Let $\rho$ be a single-mode state. The characteristic function of $\mathcal{N}_{g,\lambda}(\rho)$ is
    \bb\label{transf_caract}
        \chi_{\mathcal{N}_{g,\lambda}(\rho)}(\mathbf{r})=\chi_{\Phi_{g,0}\circ\mathcal{E}_{\lambda,0}(\rho)}(\mathbf{r})=\chi_{\mathcal{E}_{\lambda,0}(\rho)}(\sqrt{g}\mathbf{r})e^{-\frac{1}{4}(g-1)|\mathbf{r}|^2}=\chi_{\rho}(\sqrt{g\lambda}\,\mathbf{r})e^{-\frac{1}{4}\left[2g-1-g\lambda\right]|\mathbf{r}|^2}\,
    \ee
    for all $\mathbf{r}\in\mathbb{R}^2$, where we have used~\eqref{caract_att} and~\eqref{caract_amp}. Consequently, by exploiting~\eqref{caract_att},~\eqref{caract_amp}, and~\eqref{caract_add}, one can check that for all $\mathbf{r}\in\mathbb{R}^2$ it holds that
    \bb
        \chi_{\mathcal{E}_{\lambda,\nu}(\rho)}(\mathbf{r})&= \chi_{\mathcal{N}_{1+(1-\lambda)\nu\,,\,\frac{\lambda}{1+(1-\lambda)\nu}}(\rho)}(\mathbf{r}) \,,\\
        \chi_{\Phi_{g,\nu}(\rho)}(\mathbf{r})&= \chi_{\mathcal{N}_{g+(g-1)\nu\,,\,\frac{g}{g+(g-1)\nu}}(\rho)}(\mathbf{r}) \,,\\
        \chi_{\Lambda_{\xi}(\rho)}(\mathbf{r})&= \chi_{\mathcal{N}_{1+\xi\,,\,\frac{1}{1+\xi}}(\rho)}(\mathbf{r}) \,.\\
    \ee
    Hence, by exploiting the fact that quantum states and characteristic functions are in one-to-one correspondence,~\eqref{channels_as_composition} is proved.
\end{proof}
The forthcoming Theorem~\ref{kraus_comp_thm} provides a simple Kraus representation of $\mathcal{N}_{g,\lambda}$ and allows one to easily calculate the output of $\mathcal{N}_{g,\lambda}$ for generic input operators.
\begin{thm}\label{kraus_comp_thm}
Let $\lambda\in[0,1]$ and $g\ge 1$. The quantum channel $\mathcal{N}_{g,\lambda}$, defined in Definition~\eqref{def1}, admits the following Kraus representation:
\bb\label{krausform_comp}
    \mathcal{N}_{g,\lambda}(\rho)=\sum_{k,m=0}^\infty M_{k,m}^{\text{(comp)}}(g,\lambda)\,\rho \left(M_{k,m}^{\text{(comp)}}(g,\lambda)\right)^\dagger\,,
\ee
where
\bb\label{kraus_expl_comp}
    M_{k,m}^{(comp)}(g,\lambda)\coloneqq M_{k}^{\text{(pure amp)}}(g)\,M_{m}^{\text{(pure loss)}}(\lambda)=\sqrt{ \frac{ (g-1)^k\, (1-\lambda)^m}{k!\,m!\, g^{k+1}  } }(a^\dagger)^k \left(\sqrt{\frac{\lambda}{g}}\right)^{  a^\dagger a }a^m
\ee
and where we have introduced the Kraus operators of pure loss channel and pure amplifier channel:
\begin{align}
    M_{k}^{\text{(pure amp)}}(g)&\coloneqq \frac{1}{\sqrt{g\,k!}}\left(\sqrt{\frac{g-1}{g}}\right)^k (a^\dagger)^k \left(\frac{1}{\sqrt{g}}\right)^{  a^\dagger a }\,,\\
    M_{m}^{\text{(pure loss)}}(\lambda)&\coloneqq\sqrt{\frac{{(1-\lambda)^m}}{m!}}  (\sqrt{\lambda})^{ a^\dagger a }\, a^m\,.
\end{align}
In particular, by letting $\ket{n}$ and $\ket{i}$ two Fock states, it holds that
\begin{align}\label{action_comp_chan}
    \mathcal{N}_{g,\lambda}(\ketbraa{n}{i})&=\sum_{l=\max(i-n,0)}^\infty f_{n,i,l}(g,\lambda)\ketbraa{l+n-i}{l}\,.
\end{align}
where
\begin{align}\label{def_f_comp}
    f_{n,i,l}(g,\lambda)&\coloneqq \sum_{m=\max(i-l,0)}^{\min(n,i)}\frac{\sqrt{n!i!l!(l+n-i)!}}{(n-m)!(i-m)!m!(l+m-i)!} \frac{(g-1)^{l+m-i}(1-\lambda)^m\lambda^{\frac{n+i-2m}{2}}}{g^{l+1+\frac{n-i}{2}} } \,.
\end{align}
\end{thm}
\begin{proof}
By using~\eqref{def_therm}, the pure loss channel can be written as
    \bb\label{kraus_pure_def}
        \mathcal{E}_{\lambda,0}(\rho)=\sum_{m=0}^\infty M_{m}^{\text{(pure loss)}}(\lambda)\,\rho \left(M_{m}^{\text{(pure loss)}}(\lambda)\right)^\dagger\,,
    \ee
where for all $m\in\N$ the Kraus operator $M_{m}^{\text{(pure loss)}}(\lambda)$ is
    \bb
        M_{m}^{\text{(pure loss)}}(\lambda)\coloneqq (-1)^m\bra{m}_E U_{\lambda}^{SE}\ket{0}_E\,.
    \ee
Hence, by using the disentangling formula for beam splitter unitary~\cite[Appendix 5]{BARNETT-RADMORE}
\bb
    U_{\lambda}^{SE}=e^{-\sqrt{\frac{1-\lambda}{\lambda}}ab^\dagger}e^{ \frac{1}{2}\ln\lambda\,\left(a^\dagger a -b^\dagger b\right) }e^{\sqrt{\frac{1-\lambda}{\lambda}}a^\dagger b}\,
\ee
and the fact that
\bb
    e^{ -\frac{1}{2}\ln\lambda\, a^\dagger a }\,a^m\, e^{ \frac{1}{2}\ln\lambda\, a^\dagger a }=\lambda^{m/2} a\,,
\ee
it holds that
\bb
    M_{m}^{\text{(pure loss)}}(\lambda)=\frac{1}{\sqrt{m!}}\left(\sqrt{\frac{1-\lambda}{\lambda}}\right)^m a^m e^{ \frac{1}{2}\ln\lambda\, a^\dagger a }=\sqrt{\frac{{(1-\lambda)^m}}{m!}}  (\sqrt{\lambda})^{ a^\dagger a }\, a^m\,.
\ee
By using~\eqref{def_ampl}, the pure amplifier channel can be written as
    \bb\label{kraus_amp_def}
        \Phi_{g,0}(\rho)=\sum_{k=0}^\infty M_{k}^{\text{(pure amp)}}(g)\,\rho \left( M_{k}^{\text{(pure amp)}}(g) \right)^\dagger\,,
    \ee
where for all $k\in\N$ the Kraus operator $M_{k}^{\text{(pure amp)}}(g)$ is
    \bb
        M_{k}^{\text{(pure amp)}}(g)\coloneqq \bra{k}_E U_{g}^{SE}\ket{0}_E\,.
    \ee
Hence, by using the disentangling formula for the two-mode squeezing unitary~\cite[Appendix 5]{BARNETT-RADMORE}
\bb
    U_{g}^{SE}=e^{\sqrt{\frac{g-1}{g}}a^\dagger b^\dagger}e^{ \frac{1}{2}\ln\left(\frac{1}{g}\right)\,\left(a^\dagger a -b^\dagger b+1\right) }e^{-\sqrt{\frac{g-1}{g}}a b}\,,
\ee
it holds that
\bb
    M_{k}^{\text{(pure amp)}}(g)=\frac{1}{\sqrt{g\,k!}}\left(\sqrt{\frac{g-1}{g}}\right)^k (a^\dagger)^k \left(\frac{1}{\sqrt{g}}\right)^{  a^\dagger a }\,.
\ee
By using~\eqref{kraus_pure_def},~\eqref{kraus_amp_def}, and the fact $\mathcal{N}_{g,\lambda}=\Phi_{g,0}\circ\mathcal{E}_{\lambda,0}$,~\eqref{krausform_comp} is proved. Now, let us calculate $\mathcal{N}_{g,\lambda}(\ketbraa{n}{i})=\sum_{m,n=0}^\infty M_{k,m}^{\text{(comp)}}(g,\lambda)\ketbraa{n}{i}( M_{k,m}^{\text{(comp)}}(g,\lambda) )^\dagger$ in order to prove~\eqref{action_comp_chan}. By exploiting the following formulae
\bb
a^m\ket{n}&=
\begin{cases}
	\sqrt{\frac{n!}{(n-m)!}}\ket{n-m}, & \text{if $n\ge m$,} \\
	0, & \text{otherwise}
\end{cases}\\
\left(a^\dagger\right)^k\ket{n-m}&=
\sqrt{\frac{(n-m+k)!}{(n-m)!}}\ket{n-m+k}\,,
\ee
 for $m>n$ it holds that $M_{k,m}\ket{n}=0$, otherwise for $m\le n$ it holds that
\bb
     M_{k,m}^{(comp)}(g,\lambda)\ket{n}&=\frac{1}{(n-m)!}\sqrt{\frac{n!(n-m+k)!}{k!m!}}  \sqrt{\frac{(g-1)^{k}(1-\lambda)^m\lambda^{n-m}}{g^{k+1+n-m}}}\ket{n-m+k}\,.
\ee
Consequently, we conclude that
\bb
    \mathcal{N}_{g,\lambda}(\ketbraa{n}{i})&=\sum_{k=0}^{\infty}\sum_{m=0}^{\min(n,i)}\frac{\sqrt{n!(n-m+k)!i!(i-m+k)!}}{(n-m)!(i-m)!k!m!}  \sqrt{\frac{(g-1)^{2k}(1-\lambda)^{2m}\lambda^{n+i-2m}}{g^{2k+2+n+i-2m}}}\ketbraa{n-m+k}{i-m+k}\\&=\sum_{l=\max(i-n,0)}^\infty f_{n,i,l}(g,\lambda)\ketbraa{l+n-i}{l}   \,.
\ee
Hence,~\eqref{action_comp_chan} is proved.
\end{proof}
Calculating the action of Gaussian channels on non-Gaussian states is cumbersome in general. The forthcoming Theorem~\ref{gen_master_eq_trick} overcomes this difficulty and allows one to easily calculate the output of all piBGCs for generic input operators. 
\begin{thm}\label{gen_master_eq_trick}
Let $\nu\ge 0$, $\lambda\in[0,1]$, $g\ge 1$, and $\xi\ge0$. The thermal attenuator $\mathcal{E}_{\lambda,\nu}$, the thermal amplifier $\Phi_{g,\nu}$, and the additive Gaussian noise $\Lambda_\xi$ admit the following Kraus representations:
\bb\label{krausform_att}
    \mathcal{E}_{\lambda,\nu}(\rho)=\sum_{k,m=0}^\infty M_{k,m}^{\text{(att)}}(\lambda,\nu)\,\rho \left(M_{k,m}^{\text{(att)}}(\lambda,\nu)\right)^\dagger\,,
\ee
\bb\label{krausform_amp}
    \Phi_{g,\nu}(\rho)=\sum_{k,m=0}^\infty M_{k,m}^{\text{(amp)}}(g,\nu)\,\rho \left(M_{k,m}^{\text{(amp)}}(g,\nu)\right)^\dagger\,,
\ee
\bb\label{krausform_add}
    \Lambda_{\xi}(\rho)=\sum_{k,m=0}^\infty M_{k,m}^{\text{(add)}}(\xi)\,\rho \left(M_{k,m}^{\text{(add)}}(\xi)\right)^\dagger\,,
\ee
where 
\begin{align}
    M_{k,m}^{\text{(att)}}(\lambda,\nu)&\coloneqq M_{k,m}^{\text{(comp)}}\left(1+(1-\lambda)\nu,\frac{\lambda}{1+(1-\lambda)\nu}\right)\,,\label{kraus_op_att}\\
    M_{k,m}^{\text{(amp)}}(g,\nu)&\coloneqq M_{k,m}^{\text{(comp)}}\left(g+(g-1)\nu,\frac{g}{g+(g-1)\nu}\right)\,,\label{kraus_op_ampl}\\
    M_{k,m}^{\text{(add)}}(\xi)&\coloneqq M_{k,m}^{\text{(comp)}}\left(1+\xi,\frac{1}{1+\xi}\right)\,,\label{kraus_op_add}
\end{align}
with $M_{k,m}^{\text{(comp)}}$ being defined in~\eqref{kraus_expl_comp}.
In particular, by letting $\ket{n}$ and $\ket{i}$ two Fock states, it holds that
\begin{align}
    \mathcal{E}_{\lambda,\nu}(\ketbraa{n}{i})&=\sum_{l=\max(i-n,0)}^\infty f_{n,i,l}\left(1+(1-\lambda)\nu,\frac{\lambda}{1+(1-\lambda)\nu}\right)\,\ketbraa{l+n-i}{l}\,,\label{action_ni_att}\\
    \Phi_{g,\nu}(\ketbraa{n}{i})&=\sum_{l=\max(i-n,0)}^\infty f_{n,i,l}\left(g+(g-1)\nu,\frac{g}{g+(g-1)\nu}\right)\,\ketbraa{l+n-i}{l}\,,\label{action_ni_amp}\\
    \Lambda_{\xi}(\ketbraa{n}{i})&=\sum_{l=\max(i-n,0)}^\infty f_{n,i,l}\left(1+\xi,\frac{1}{1+\xi}\right)\,\ketbraa{l+n-i}{l}\,,\label{action_ni_add}
\end{align}
with $f_{n,i,l}$ being defined in~\eqref{def_f_comp}.
\end{thm}
\begin{proof}
    Theorem~\ref{gen_master_eq_trick} is a direct consequence of Lemma~\ref{lemma_comp_bos} and Theorem~\ref{kraus_comp_thm}.
\end{proof}
We observe here that the Kraus representation in~\eqref{kraus_op_att} of the thermal attenuator is precisely the one obtained in~\cite{Die-Hard-2-PRA} via the ``master equation trick".

\section{Results}
In this section we expound our results. In subsection~\ref{sub_preliminary} first we prove preliminary results on the two-way capacities of generic quantum channels and second we apply them to the composition between pure amplifier channel and pure loss channel. In subsection~\ref{sub_res_twoway} we specialise these results to the case of piBGCs (thermal attenuator, thermal amplifier, and additive Gaussian noise) and we find the following two main results:
\begin{itemize}
    \item The parameter regions where the (EC) two-way capacities of piBGCs are strictly positive are precisely those where these channels are not entanglement breaking;
    \item We find a new lower bound on the secret-key and two-way quantum capacity of piBGCs, which constitutes a significant improvement with respect the state-of-the-art lower bounds~\cite{Ottaviani_new_lower,Pirandola2009,Pirandola18,Wang_Q2_amplifier,Noh2020} in many parameter regions.
\end{itemize}
\subsection{Preliminary results}\label{sub_preliminary}
Let us begin by introducing the concept of a \emph{generalized Choi state} of a quantum channel.
\begin{Def}
[Generalised Choi state of a quantum channel~\cite{Holevo-CJ,Holevo-CJ-arXiv}]\label{gen_choi_thm}
Let $\HH_{B}$ be a possibly infinite-dimensional Hilbert space. Let $\HH_A,\HH_{A'}$ be isomorphic (possibly infinite dimensional) Hilbert spaces. Let $\ket{\psi}_{A'A}$ be a pure state of the form
\bb\label{state_psi_a'a}
    \ket{\psi}_{A'A}=\sum_{i}\sqrt{\lambda_i}\ket{e_i}_{A}\otimes\ket{e_i}_{A'}\,,
\ee
where $(\lambda_i)_{i}$ are strictly positive numbers such that $\sum_i\lambda_i=1$, and $(\ket{e_i}_A)_{i}$ and $(\ket{e_i}_{A'})_{i}$ form an orthonormal basis of $\HH_A$ and $\HH_{A'}$, respectively. Let $\Phi_{A'\to B}$ be a quantum channel from $\HH_{A'}$ to $\HH_B$. Then, the state 
\bb
    C_{AB}\coloneqq\Id_{A}\otimes\Phi_{A'\to B}(\ketbra{\psi}_{AA'})
\ee
is said to be a generalised Choi state of $\Phi$.
\end{Def}
In finite dimensions, if the state $\ket{\psi}$ in \eqref{state_psi_a'a} is a \emph{maximally entangled state}, then the state $C_{AB}\coloneqq\Id_{A}\otimes\Phi_{A'\to B}(\ketbra{\psi}_{AA'})$ is simply referred to as the Choi state of the channel $\Phi$. Additionally, it is well known that a quantum channel is completely characterised by its Choi state~\cite{Sumeet_book}. The following lemma extends this result, showing that a quantum channel can also be completely characterised in terms of its \emph{generalised} Choi state.
\begin{lemma}\label{lemma_choichoi}
    Let $\Phi_{A'\to B}$ be a quantum channel and let be $C_{AB}$ a generalised Choi state. Then, it holds that
    \bb
    \Phi_{A'\to B}(X_{A'})&=\Tr_{A}\!\left[\left(X_A\otimes\mathbb{1}_{B}\right)\,(D_A\otimes\mathbb{1}_B)\,C_{AB}^{t_A}\,(D_A\otimes\mathbb{1}_B)\right]   \quad\text{for all linear operators }X_{A'}\,,
    \ee
    where, by using the notation introduced in Definition~\ref{gen_choi_thm},  $D_A\coloneqq\sum_{i}\lambda_i^{-1/2}\ketbraa{e_i}{e_i}_{A}$ and $t_A$ is the partial transpose on $A$, i.e.~$(\ketbraa{e_i}{e_j}_{A})^{t_A}=\ketbraa{e_j}{e_i}_{A}$ for all $i,j$.

\end{lemma}
\begin{proof}
    By exploiting that $(\ket{e_i}_A)_{i}$ are orthonormal, we have that
    \bb\label{generalised_choi_eq_def}
    \Phi_{A'\to B}\!\left(\ketbraa{e_i}{e_j}_{A'}\right) = \frac{1}{\sqrt{\lambda_i \lambda_j}}\Tr_{A}\!\left[\left(\ketbraa{e_j}{e_i}_{A}\otimes\mathbb{1}_{B}\right)\,C_{AB}\right],\quad\forall\,i,j\,.
    \ee
    By writing 
    \bb
        X_{A'}=\sum_{ij}\bra{e_i}X\ket{e_j} \ketbraa{e_i}{e_j}_{A'}
    \ee
    and by exploiting the linearity of $\Phi_{A'\to B}$, it thus follows that
    \bb
    \Phi_{A'\to B}(X_{A'})&=\Tr_{A}\!\left[\left((D_AX_{A} D_A)^{t_A}\otimes\mathbb{1}_{B}\right)\,C_{AB}\right]\\
    &=\Tr_{A}\!\left[\left(D_AX_{A} D_A\otimes\mathbb{1}_{B}\right)\,C_{AB}^{t_A}\right]\\
    &=\Tr_{A}\!\left[\left(X_A\otimes\mathbb{1}_{B}\right)\,(D_A\otimes\mathbb{1}_B)\,C_{AB}^{t_A}\,(D_A\otimes\mathbb{1}_B)\right]\,.
    \ee
\end{proof}
In finite dimensions, it is well known that a quantum channel is entanglement breaking if and only if its Choi state is separable. The following lemma generalises this result, demonstrating that a quantum channel is entanglement breaking if and only if its \emph{generalised} Choi state is separable.
\begin{lemma}\label{lemma_eb_gen_choi}
    Let $\Phi$ be a quantum channel and let $C_{AB}$ be a generalised Choi state of $\Phi$. Then, $\Phi$ is entanglement breaking if and only if $C_{AB}$ is separable.
\end{lemma}
\begin{proof}
    By definition, if $\Phi$ is entanglement breaking, then $\Id_{R}\otimes\Phi_{A'\to B}(\rho_{RA'})$ is separable for all bipartite states $\rho_{RA'}$. In particular, any generalised Choi state of an entanglement breaking channel is separable.

    Conversely, let us assume that the generalised Choi state $C_{AB}$ is separable, that is there exists a probability distribution $p_x$ and states $(\rho_{A}^{(x)})_x,(\sigma_{B}^{(x)})_x$ such that
    \bb
        C_{AB}=\sum_x p_x\, \rho_{A}^{(x)}\otimes \sigma_{B}^{(x)}\,.
    \ee
    Let us show that $\Phi$ is entanglement breaking. To this end let us consider an arbitrary bipartite state $\rho_{RA'}$ and let us show that $\Id_{R}\otimes\Phi_{A'\to B}(\rho_{RA'})$ is separable. By exploiting Lemma~\ref{lemma_choichoi}, it holds that
    \bb\label{conti_choi_gen}
        \Id_{R}\otimes\Phi_{A'\to B}(\rho_{RA'})&=\Tr_{A}\!\left[\left(\rho_{RA}\otimes\mathbb{1}_{B}\right)\,\left(\mathbb{1}_R\otimes(D_A\otimes\mathbb{1}_B)\,C_{AB}^{t_A}\,(D_A\otimes\mathbb{1}_B)\right)\right]\\
        &= \sum_x p_x\, \Tr_{A}\!\left[\rho_{RA}\,\left(\mathbb{1}_R\otimes D_A\,(\rho_{A}^{(x)})^{t_A}\,D_A\right)\right] \otimes \sigma_{B}^{(x)}\,.
    \ee 
    Since the operator $D_A\,(\rho_{A}^{(x)})^{t_A}\,D_A$ is positive semidefinite, we can write its spectral decomposition as
    \bb
        D_A\,(\rho_{A}^{(x)})^{t_A}\,D_A=\sum_{i} \eta^{(x)}_{i}\ketbra{\phi^{(x)}_i}
    \ee
    with the eigenvalues $(\eta^{(x)}_{i})_i$ being positive. This implies that the operator $ \Tr_{A}\!\left[\rho_{RA}\,\left(\mathbb{1}_R\otimes D_A\,(\rho_{A}^{(x)})^{t_A}\,D_A\right)\right]$ is positive semidefinite, as it can be written as
    \bb
        \Tr_{A}\!\left[\rho_{RA}\,\left(\mathbb{1}_R\otimes D_A\,(\rho_{A}^{(x)})^{t_A}\,D_A\right)\right]= \sum_i \eta^{(x)}_{i} \bra{\phi^{(x)}_i}_A \rho_{RA} \ket{\phi^{(x)}_i}_A\,.
    \ee 
   and $\bra{\phi^{(x)}_i}_A \rho_{RA} \ket{\phi^{(x)}_i}_A$ is positive semidefinite. In particular, the trace of the operator $\Tr_{A}\!\left[\rho_{RA}\,\left(\mathbb{1}_R\otimes D_A\,(\rho_{A}^{(x)})^{t_A}\,D_A\right)\right]$ vanishes if and only if it is the zero operator. Consequently, \eqref{conti_choi_gen} implies that
    \bb
        \Id_{R}\otimes\Phi_{A'\to B}(\rho_{RA'})&= \sum_{x:\,q_x\ne 0} q_x\, \omega_{R}^{(x)} \otimes \sigma_{B}^{(x)}\,,
    \ee
    where we defined
    \bb
        q_x&\coloneqq p_x\Tr_{RA}\!\left[\rho_{RA}\,\left(\mathbb{1}_R\otimes D_A\,(\rho_{A}^{(x)})^{t_A}\,D_A\right)\right]\,,\\
        \omega_{R}^{(x)} &\coloneqq \frac{\Tr_{A}\!\left[\rho_{RA}\,\left(\mathbb{1}_R\otimes D_A\,(\rho_{A}^{(x)})^{t_A}\,D_A\right)\right]}{\Tr_{RA}\!\left[\rho_{RA}\,\left(\mathbb{1}_R\otimes D_A\,(\rho_{A}^{(x)})^{t_A}\,D_A\right)\right]}\,.
    \ee
    The fact that the operator $\Tr_{A}\!\left[\rho_{RA}\,\left(\mathbb{1}_R\otimes D_A\,(\rho_{A}^{(x)})^{t_A}\,D_A\right)\right]$ is positive semidefinite implies that $(q_x)_x$ is a probability distribution and that $\omega_{R}^{(x)}$ is a quantum state. Hence, we conclude that $\Id_{R}\otimes\Phi_{A'\to B}(\rho_{RA'})$ is separable.
\end{proof}

The following theorem establishes that the energy-constrained two-way capacities of a single-mode Gaussian channel are strictly positive if and only if the channel is not entanglement breaking.
\begin{thm} 
    Let $\Phi$ be a single-mode Gaussian channel and let $N_s>0$. The energy-constrained two-way quantum capacity $Q_2(\Phi,N_s)$ and secret-key capacity $K(\Phi,N_s)$ are strictly positive if and only if $\Phi$ is not entanglement breaking. 
\end{thm}
\begin{proof}[Proof] 
    Since any entanglement-breaking channel has vanishing two-way capacities~\cite{Sumeet_book}, it suffices to consider the case where $\Phi$ is not entanglement breaking. Assume that Alice prepares many copies of the two-mode squeezed vacuum state $\ket{\Psi_{N_s}}_{AA'}$ with mean local photon number $N_s$ and sends the systems $A'$ through the Gaussian channel $\Phi_{A'\to B}$. Now Alice and Bob share many copies of the two-mode Gaussian state $C_{N_s}\coloneqq \big(\Id_{A}\otimes\,\Phi_{A'\to B}\big)(\ketbra{\Psi_{N_s}}_{AA'})$, which is a generalised Choi state of $\Phi$~\cite{Holevo-CJ,BUCCO}. As such, $C_{N_s}$ is entangled, as established by Lemma~\ref{lemma_eb_gen_choi}. By exploiting the fact that a two-mode Gaussian state is entangled if and only if it is not PPT~\cite{PeresPPT,Simon00, BUCCO}, it thus follows that $C_{N_s}$ is not PPT. Since any two-mode Gaussian state that it is not PPT is also \emph{distillable}~\cite{Giedke01} --- i.e.~it can be converted into ebits with a strictly positive rate --- we conclude that $K(\Phi,N_s)\ge Q_2(\Phi,N_s)>0$.  
\end{proof}
In the following, we provide an alternative, more explicit proof of the above result. We start by proving the following lemma.

\begin{lemma}\label{lemma_Q2_positive}
    Let $\Phi:\mathfrak{S}(L^2(\mathbb{R}))\to\mathfrak{S}(L^2(\mathbb{R}))$ be a single-mode Gaussian quantum channel and let $N_s>0$. Suppose that $f\left( \Id\otimes\Phi(\ketbra{\Psi_{N_s}})\right)<0$, where $\ket{\Psi_{N_s}}$ is the two-mode squeezed vacuum state defined in~\eqref{two_mode_sq} and $f$ is the function defined in Lemma~\ref{ConditionPPT_cov}.
    The energy-constrained two-way capacities $Q_2(\Phi,N_s)$ and $K(\Phi,N_s)$ are strictly positive. In particular, the (unconstrained) two-way capacities $Q_2(\Phi)$ and $K(\Phi)$ are strictly positive.
\end{lemma}
\begin{proof} 
Since the state $\Id \otimes\Phi(\ketbra{\Psi_{N_s}})$ is a two-mode Gaussian state, we can apply Lemma~\ref{ConditionPPT_cov} to conclude that it is entangled. Consequently, since any two-mode Gaussian entangled state is distillable~\cite{Giedke01}, then  $\Id\otimes\Phi(\ketbra{\Psi_{N_s}})$ is distillable. 
Hence, by exploiting~\eqref{link_D2_D}, we deduce that $Q_2(\Phi,N_s)>0$. In addition,~\eqref{relation_2waycapEC} implies that $K(\Phi,N_s)>0$. Finally, since the energy-constrained capacities are lower bounds on the corresponding unconstrained capacities, we conclude that the unconstrained two-way capacities of $\Phi$ are strictly positive.
\end{proof}

The forthcoming Theorem~\ref{th1} determines the parameter region of $g\ge1$ and $\lambda\in[0,1]$ where the composition $\mathcal{N}_{g,\lambda}\coloneqq\Phi_{g,0}\circ \mathcal{E}_{\lambda,0}$ between pure amplifier channel $\Phi_{g,0}$ and pure loss channel $\mathcal{E}_{\lambda,0}$ has strictly positive (EC) two-way capacities. In particular, we show that the (EC) two-way capacities of $\mathcal{N}_{g,\lambda}$ are strictly positive if and only if $\mathcal{N}_{g,\lambda}$ is not entanglement breaking. 
\begin{thm}\label{th1}
    Let $\lambda\in[0,1]$, $g\ge1$, and $N_s>0$. The energy-constrained two-way capacities $Q_2(\mathcal{N}_{g,\lambda},N_s)$ and $K(\mathcal{N}_{g,\lambda},N_s)$ are strictly positive if and only if $(1-\lambda)g< 1$, i.e.~if and only if $\mathcal{N}_{g,\lambda}$ is not entanglement breaking. In particular, the same holds for the unconstrained two-way capacities.
\end{thm}
\begin{proof}
Suppose that $(1-\lambda)g\ge 1$. Then Lemma~\ref{lemma_eb_comp} implies that $\mathcal{N}_{g,\lambda}$ is entanglement breaking and hence~\cite{Goodenough16} its two way-capacities vanish.

Now, suppose that $(1-\lambda)g< 1$. Let us check that the hypothesis of Lemma~\ref{lemma_Q2_positive} is fulfilled, i.e.~we need to check that $f\left( \Id\otimes\mathcal{N}_{g,\lambda}(\ketbra{\Psi_{N_s}})\right)<0$, where $\ket{\Psi_{N_s}}$ is the two-mode squeezed vacuum state defined in~\eqref{two_mode_sq} and $f$ is the function defined in Lemma~\ref{ConditionPPT_cov}. Let us calculate the covariance matrix of the state 
\bb 
&\Id_{A}\otimes\mathcal{N}_{g,\lambda}(\ketbra{\Psi_{N_s}}_{AA'}) \\&=  \Tr_{E_1E_2}\left[\left(\mathbb{1}_A\otimes U_g^{A'E_1}\otimes U_\lambda^{A'E_2}\right)\, \ketbra{\Psi_{N_s}}_{AA'}\otimes\ketbra{0}_{E_1}\otimes\ketbra{0}_{E_2}\text{} \left(\mathbb{1}_A\otimes U_g^{A'E_1}\otimes U_\lambda^{A'E_2}\right)^\dagger\right]\,	
\ee
with respect the ordering $(A,A',E_1,E_2)$.
By using~\eqref{relation_S} and~\eqref{relation_S_amp}, one can show that the covariance matrix of 
\bb
\left(\mathbb{1}_A\otimes U_g^{A'E_1}\otimes U_\lambda^{A'E_2}\right)\, \ketbra{\Psi_{N_s}}_{AA'}\otimes\ketbra{0}_{E_1}\otimes\ketbra{0}_{E_2}\text{} \left(\mathbb{1}_A\otimes U_g^{A'E_1}\otimes U_\lambda^{A'E_2}\right)^\dagger
\ee
with respect the ordering $(A,A',E_1,E_2)$ is
\bb\label{cov_AA'EE_comp}
\left(\mathbb{1}_2\oplus S_g\oplus \mathbb{1}_2\right) 
\left(\mathbb{1}_2\oplus \bar{S}_\lambda\right) \left(V(\ketbra{\Psi_{N_s}}_{AA'})\oplus V(\ketbra{0})\oplus V(\ketbra{0})\right) \left(\mathbb{1}_2\oplus \bar{S}_\lambda^{\intercal}\right)\left(\mathbb{1}_2\oplus S_g^{\intercal}\oplus \mathbb{1}_2\right)\,,
\ee
where
\bb
&\bar{S}_\lambda \coloneqq \left(\begin{matrix} \sqrt{\lambda}\,\mathbb{1}_2 & 0_{2\times 2} & \sqrt{1-\lambda}\,\mathbb{1}_2 \\ 0_{2\times 2} & \mathbb{1}_{2} & 0_{2\times 2} \\
-\sqrt{1-\lambda}\,\mathbb{1}_2 & 0_{2\times 2} & \sqrt{\lambda}\,\mathbb{1}_2 \end{matrix}\right) 
\ee 
and $0_{2\times 2}\coloneqq\left(\begin{matrix}0&0\\0&0\end{matrix}\right)$. Hence, since $V\left( \Id_{A}\otimes\mathcal{N}_{g,\lambda}(\ketbra{\Psi_{N_s}}_{AA'}) \right)$ is the $4\times4$ upper-left block of the covariance matrix in~\eqref{cov_AA'EE_comp}, one can show that
\bb\label{cov_choi_comp}
&V\left( \Id_{A}\otimes\mathcal{N}_{g,\lambda}(\ketbra{\Psi_{N_s}}_{AA'}) \right) = \left(\begin{matrix} (2N_s+1)\mathbb{1}_2 & 2\sqrt{g\lambda N_s(N_s+1)}\sigma_z \\ 2\sqrt{g\lambda N_s(N_s+1)}\sigma_z  &[2g\left( 1+\lambda N_s\right)-1]\mathbb{1}_2\end{matrix}\right)\,,
\ee 
where we used~\eqref{moments_thermal} and~\eqref{moments_squeezed}.
Consequently, since 
\bb
&f\left( \Id_{A}\otimes\mathcal{N}_{g,\lambda}(\ketbra{\Psi_{N_s}}_{AA'} \right) = -16 N_s (1 + N_s) g \left(1 - (1-\lambda)g \right)\,,
\ee 
and since $(1-\lambda)g < 1$, we have that $f\left( \Id_{A}\otimes\mathcal{N}_{g,\lambda}(\ketbra{\Psi_{N_s}}_{AA'} \right)<0$, i.e.~the hypothesis of Lemma~\ref{lemma_Q2_positive} is fulfilled. Hence, Lemma~\ref{lemma_Q2_positive} implies that the energy-constrained two-way capacities of $\mathcal{N}_{g,\lambda}$ are strictly positive. {This concludes the proof of Theorem~\ref{th1}. In Remark~\ref{remark_alternative_proof} we will provide an alternative proof.}
\end{proof}

In the forthcoming Theorem~\ref{th_lower_Q2} we obtain a lower bound on the two-way capacities of a quantum channel $\Phi:\mathfrak{S}(L^2(\mathbb{R}))\to\mathfrak{S}(L^2(\mathbb{R}))$ by introducing a protocol to distribute ebits though $\Phi$. The idea of such a protocol is the following. First, Alice prepares states of the form
\bb\label{initial_state}
    \ket{\Psi_{M,c}}_{AA'}\coloneqq c\ket{0}_A\ket{0}_{A'}+\sqrt{1-c^2}\ket{M}_A\ket{M}_{A'}\,,
\ee
where $M\in\N^+$ and $c\in(0,1)$. Then, she sends the halves $A'$ to Bob trough $\Phi$, who makes a measurement on each half in order to project his half onto the span of $\{\ket{0},\ket{M}\}$. Then, Alice and Bob run $k$ times the P1-or-P2 recurrence protocol~\cite{p1orp2} on the resulting states. After this, Alice and Bob run the improved hashing protocol introduced in~\cite{Improvement-Hashing} in order to generate ebits. Let $R(\Phi,M,c,k)$ be the rate of distributed ebits of this protocol. A lower bound on $Q_2(\Phi)$ (and hence on $K(\Phi)$) can be obtained by maximising $R(\Phi,M,c,k)$ over $M\in\N^+$, $c\in(0,1)$, and $k\in\N$.

\begin{thm}\label{th_lower_Q2}
Let $\Phi:\mathfrak{S}(L^2(\mathbb{R}))\to\mathfrak{S}(L^2(\mathbb{R}))$ be a quantum channel which maps a single-mode system $A'$ into another single-mode system $B$. The EC two-way capacities $Q_2(\Phi,N_s)$ and $K(\Phi,N_s)$ satisfy the following lower bound
\bb\label{lowQ2_genEC}
	K(\Phi,N_s)&\ge Q_2(\Phi,N_s) \ge \sup_{\substack{c\in(0,1),\, M\in\N^+,\, k\in\N\\ (1-c^2)M\le N_s}} R(\Phi,M,c,k)\,,
\ee
and, in particular, the unconstrained two-way capacities satisfy
\bb\label{lowQ2_gen}
	K(\Phi)&\ge Q_2(\Phi) \ge \sup_{c\in(0,1),\, M\in\N^+,\, k\in\N} R(\Phi,M,c,k)\,,
\ee
where 
\begin{equation}\label{rate_gen}
	R(\Phi,M,c,k)\coloneqq \mathcal{C}(\Phi,c,M)  \frac{\prod_{t=0}^{k-1}P_t}{2^{k}}\mathcal{I}(  \alpha^{(k)}_{00},\alpha^{(k)}_{01},\alpha^{(k)}_{10},\alpha^{(k)}_{11})\,.
\end{equation} 
Fixed $c\in(0,1)$, $M\in\N^+$, and $k\in\N$, the quantities present in~\eqref{rate_gen} are defined as follows. $\mathcal{C}(\Phi,c,M)$ is defined as
\bb\label{formula_gen_norm}
&\mathcal{C}(\Phi,c,M)  \coloneqq   \Tr\left[\mathbb{1}_A\otimes\Pi_M\, \Id_{A}\otimes\Phi(\ketbra{\Psi_{M,c}}_{AA'}) \right]\,,
\ee
where $\Pi_M\coloneqq \ketbra{0}_{B}+\ketbra{M}_{B}$ and the state $\ket{\Psi_{M,c}}_{AA'}$ is defined in~\eqref{initial_state}.
Let us define for all $m,n\in\{0,1\}$ the coefficients $\alpha_{mn}^{(0)}$ as
    \bb \label{formula_gen_alpha0}
        &\alpha_{mn}^{(0)}\coloneqq \frac{\bra{\psi^{(M)}_{mn}}_{AB}\mathbb{1}_A\otimes\Pi_M\, \Id_{A}\otimes\Phi(\ketbra{\Psi_{M,c}})\, \mathbb{1}_A\otimes\Pi_M\ket{\psi^{(M)}_{mn}}_{AB}}{\mathcal{C}(\Phi,c,M)}\,,
    \ee    
where $\ket{\psi^{(M)}_{mn}}_{AB}$ is defined as 
\bb\label{Bell_states_delta}
\ket{\psi^{(M)}_{mn}}_{AB}\coloneqq \frac{1}{\sqrt{2}}\sum_{j=0}^1 (-1)^{mj}\ket{jM}_A\otimes\ket{(j\oplus n)M}_{B}\,.
\ee 
For all $t\in\{0,1,\ldots,k-1\}$ and all $m,n\in\{0,1\}$ the coefficients $\alpha_{mn}^{(t+1)}$ and $P_t$ are defined in the following way:
\begin{itemize}
    \item If $\alpha^{(t)}_{10}<\alpha^{(t)}_{01}$, then
    \bb\label{coeff_p1}
            \alpha^{(t+1)}_{mn}\coloneqq \frac{1}{P_t}\sum_{\substack{m_1,m_2=0\\m_1\oplus m_2=m}}^{1}\alpha_{m_1n}^{(t)}\alpha_{m_2n}^{(t)}\,,
    \ee
    where
    \bb\label{probk_p1}
        P_t \coloneqq \sum_{n=0}^{1}\left(\sum_{m=0}^{1}\alpha^{(t)}_{m n}\right)^2\,.
    \ee
    \item If $\alpha^{(t)}_{10}\ge \alpha^{(t)}_{01}$, then
    \bb\label{coeff_p2}
    \alpha^{(t+1)}_{mn}\coloneqq \frac{1}{P_t}\sum_{\substack{n_1,n_2=0\\n_1\oplus n_2=n}}^{1}\alpha_{mn_1}^{(t)}\alpha_{mn_2}^{(t)}\,,
    \ee    
    where
    \bb\label{probk_p2}
        P_t \coloneqq \sum_{m=0}^{1}\left(\sum_{n=0}^{1}\alpha^{(t)}_{m n}\right)^2\,.
    \ee
\end{itemize}
For all $\alpha_{00},\alpha_{01},\alpha_{10},\alpha_{11}\ge0$ with $\alpha_{00}+\alpha_{01}+\alpha_{10}+\alpha_{11}=1$, the quantity $\mathcal{I}(  \alpha_{00},\alpha_{01},\alpha_{10},\alpha_{11})$ is defined as 
\bb\label{yield_final_protocol}
&\mathcal{I}(  \alpha_{00},\alpha_{01},\alpha_{10},\alpha_{11})\coloneqq \max\left(\,Y(  \alpha_{00},\alpha_{01},\alpha_{10},\alpha_{11}),  Y(  \alpha_{00},\alpha_{10},\alpha_{01},\alpha_{11} ),\, Y(  \alpha_{01},\alpha_{00},\alpha_{10},\alpha_{11})\right)\,,
\ee
where the function $Y$ is defined in~\eqref{improv_hashing}.
\end{thm}
\begin{proof} We introduce a protocol to distribute ebits through the channel $\Phi$, which depends on three parameters: $M\in\N^+$, $c\in(0,1)$,  $k\in\N$.
Our lower bound on $Q_2(\Phi,N_s)$ in~\eqref{lowQ2_genEC} can be obtained by optimising over these parameters the rate of ebits of such a protocol. The lower bound on the other EC two-way capacities follows from~\eqref{relation_2waycapEC}.
The steps of the protocol are the following.

\textbf{-Step 1}: Alice prepares $n_0$ copies of the state $\ket{\Psi_{M,c}}_{AA'}$ in~\eqref{initial_state} and she sends the halves $A'$ to Bob through the channel $\Phi$. Hence, Alice and Bob share $n_0$ copies of the state $\Id_{A}\otimes\Phi(\ketbra{\Psi_{M,c}})$.

\textbf{-Step 2}: Bob performs the local POVM $\{\Pi_M, \mathbb{1}-\Pi_M\}$ on each pair $\Id_{A}\otimes\Phi(\ketbra{\Psi_{M,c}})$, where $\Pi_M\coloneqq \ketbra{0}+\ketbra{M}$. If Bob finds the outcome which corresponds to $\Pi_M$, then Alice and Bob keep the pair, otherwise they discard it. They keep the pair with probability
    \bb
            \mathcal{C}(\Phi,c,M)&\coloneqq \Tr\left[\mathbb{1}_A\otimes\Pi_M\, \Id_{A}\otimes\Phi(\ketbra{\Psi_{M,c}}) \right]\,.
    \ee
At this point, Alice and Bob shares $\approx n_0\,\mathcal{C}(\Phi,c,M)$ pairs.
    Each of these pairs are in the state $\rho'$ given by 
\bb\label{rho_primo_distil}
        \rho'&= \frac{\mathbb{1}_A\otimes\Pi_M\, \Id_{A}\otimes\Phi(\ketbra{\Psi_{M,c}})\, \mathbb{1}_A\otimes\Pi_M}{\mathcal{C}(\Phi,c,M)} \,.
\ee
Note that the support of $\rho'$ is equal to $\text{Span}\{\ket{0}\otimes\ket{0},\ket{0}\otimes\ket{M},\ket{M}\otimes\ket{0},\ket{M}\otimes\ket{M}\}$.
For simplicity, in the following we will use the notation $\ket{1}\equiv \ket{M}$. This formally corresponds to consider the state $\rho''\coloneqq U_M\otimes U_M\, \rho'\, U_M^\dagger \otimes U_M^\dagger$, which is obtained once both Alice and Bob have applied the unitary 
\bb
U_M\coloneqq \sum_{i\ne\{0, M\}}^\infty\ketbra{i}+\ketbraa{1}{M}+\ketbraa{M}{1}
\ee
on the remaining state $\rho'$. Hence, since the support of $\rho''$ is equal to $\text{Span}\{\ket{0}\otimes\ket{0},\ket{0}\otimes\ket{1},\ket{1}\otimes\ket{0},\ket{1}\otimes\ket{1}\}$, in the following we consider transformations which act on qubit systems.

\textbf{-Step 3}:
    For each of the $\approx n_0\,\mathcal{C}(\Phi,c,M)$ pairs, Alice and Bob choose randomly two bits $\mu,\nu\in\{0,1\}$ and they both apply the unitary $\sigma_{\mu\nu}$ defined by 
    \bb\label{def_pauli}
        \sigma_{\mu\nu}\coloneqq\sum_{i=0}^{1} (-1)^{\mu i}\ketbraa{i\oplus \nu}{i}\,
    \ee
    (in terms of the Pauli matrices it holds that $\sigma_{00}=\mathbb{1}_2$, $\sigma_{01}=\sigma_x$,  $\sigma_{10}=\sigma_z$, and $\sigma_{11}=i\sigma_y$).
    Hence, each pair is transformed into the state $\rho_0$ defined by
    \bb
        \rho_0\coloneqq \frac{1}{4}\sum_{\mu,\nu=0}^{1} \left(\sigma_{\mu\nu}\otimes\sigma_{\mu\nu}\right)\,\rho''\,\left(\sigma_{\mu\nu}\otimes\sigma_{\mu\nu}\right)^\dagger\,.
    \ee
    By exploiting the fact that the Bell states defined in~\ref{Bell_states} form an orthonormal basis, one can show that $\rho$ is diagonal in the Bell basis:
    \bb
        \rho_0=\sum_{m,n=0}^{1}\alpha_{mn}^{(0)}\ketbra{\psi_{mn}}\,,
    \ee
    where the coefficients $\alpha_{mn}^{(0)}$ are given by 
    \bb
        \alpha_{mn}^{(0)}&= \bra{\psi_{mn}}\rho''\ket{\psi_{mn}} = \bra{\psi^{(M)}_{mn}}\rho'\ket{\psi^{(M)}_{mn}}\,,
    \ee
    with $\ket{\psi^{(M)}_{mn}}$ being defined in~\ref{Bell_states_delta}.

\textbf{-Step 4}: 
 Alice and Bob run the following sub-routine, which is a recurrence protocol dubbed \emph{P1-or-P2}~\cite{p1orp2}.
     \begin{itemize}
         \item \textbf{Step 4.0}: Let $t=0$.
         \item \textbf{Step 4.1}: 
         At this point, all the pairs are in the state $\rho_t$.  Alice and Bob collect all the pairs in groups of two pairs. Let $\rho_t^{(A_1B_1)}$ denote the first pair of each group and let $\rho_t^{(A_2B_2)}$ denote the second one. If $\alpha^{(t)}_{10}<\alpha^{(t)}_{01}$, then Alice and Bob apply the bi-local unitary $U_{1}$ defined as
         \bb
	         U_{1}\coloneqq U_{\text{CNOT}}^{(A_1A_2)}\otimes U_{\text{CNOT}}^{(B_1B_2)}\,,
         \ee
         where for all $S=A,B$ the operator $U_{\text{CNOT}}^{(S_1S_2)}$ is the CNOT gate on $S_1$ and $S_2$ with control qubit $S_1$, i.e.
        \bb
          U_{\text{CNOT}}^{(S_1S_2)}\ket{i}_{S_1}\otimes\ket{j}_{S_2}=\ket{i}_{S_1}\otimes\ket{i\oplus j}_{S_2}\,.
         \ee
          Otherwise if $\alpha^{(t)}_{10}\ge\alpha^{(t)}_{01}$, they apply the bi-local unitary $U_{2}$ defined as
          	\bb
         	U_{2}&\coloneqq(H^{(A_1)}\otimes H^{(B_1)})(U_{\text{CNOT}}^{(A_1A_2)}\otimes U_{\text{CNOT}}^{(B_1B_2)}) (H^{(A_1)}\otimes H^{(A_2)}\otimes H^{(B_1)}\otimes H^{(B_2)})\,,
         	\ee
 where for all $S=A_1,A_2,B_1,B_2$ the operator $H^{(S)}$ on $S$ is the Hadamard gate, i.e.
  \bb \label{hadamard}
 H^{(S)}=\frac{1}{\sqrt{2}}\sum_{m,n=0}^{1}(-1)^{mn}\ketbraa{n}{m}_S\,.
 \ee
At this point, the state of $A_1A_2B_1B_2$ is 
\bb
            \rho_t^{(A_1A_2B_1B_2)}\coloneqq U_p\,\rho_t^{(A_1B_2)}\otimes\rho_k^{(A_2B_2)}\, U_p^\dagger\,.
\ee
with $p=1$ if $\alpha^{(t)}_{10}<\alpha^{(t)}_{01}$, and $p=2$ otherwise.
         \item \textbf{Step 4.2}:  Alice and Bob measure the pair $A_2B_2$ of each group with respect to the local POVM $\{M_{i,j}\}_{i,j\in\{0,1\}}$ with $M_{i,j}\coloneqq\ketbra{i}_{A_2}\otimes\ketbra{j}_{B_2}$ for all $i,j\in\{0,1\}$. Then they discard the pair $A_2B_2$. 
         They discard also the pair $A_1B_1$ if the outcome of the previous measurement corresponds to $M_{i,j}$ with $i\ne j$. The probability that a pair $A_1B_1$ is not discarded is given by 
          \bb
         P_t &\coloneqq \sum_{i=0}^{1}\bra{i}_{A_2}\bra{i}_{B_2}\Tr_{A_1B_1}\left[\rho_t^{(A_1A_2B_1B_2)}\right]\ket{i}_{A_2}\ket{i}_{B_2}\,.
         \ee
         By using that for all $k_1,k_2,j_1,j_2\in\{0,1\}$ it holds that
         \bb\label{formula_bi_cnot}
         &U_{\text{CNOT}}^{(A_1A_2)}\otimes U_{\text{CNOT}}^{(B_1B_2)}\ket{\psi_{k_1j_1}}_{A_1B_1}\otimes\ket{\psi_{k_2j_2}}_{A_2B_2} =\ket{\psi_{k_1\oplus k_2,\,j_1}}_{A_1B_1}\otimes\ket{\psi_{k_2,\,j_1\oplus j_2}}_{A_2B_2}
         \ee
         and that 
          \bb\label{formula_bi_hadamard}
         H^{(A)}\otimes H^{(B)}\ket{\psi_{k_1j_1}}_{AB}=(-1)^{k_1j_1}\ket{\psi_{j_1k_1}}_{AB}\,,
         \ee
		one can show that $P_t$ can be expressed as in~\eqref{probk_p1}  if $\alpha^{(t)}_{10}<\alpha^{(t)}_{01}$, and as in~\eqref{probk_p2} otherwise.
        At this point, the number of remaining pairs is 
         \begin{equation}
         	\approx n_0\,\mathcal{C}(\Phi,c,M) \frac{1}{2^{t+1}} \prod_{m=0}^{t}P_m\,
         \end{equation}
     	and each of these is in the state $\rho_{t+1}$ given by
        \bb
\rho_{t+1}&=\frac{1}{2}\sum_{i=0}^1\frac{\bra{i}_{A_2}\bra{i}_{B_2}\rho_t^{(A_1A_2B_1B_2)}\ket{i}_{A_2}\ket{i}_{B_2}}{\Tr_{A_1B_1}\left[ \bra{i}_{A_2}\bra{i}_{B_2}\rho_t^{(A_1A_2B_1B_2)}\ket{i}_{A_2}\ket{i}_{B_2} \right]}\,.
\ee
By using~\eqref{formula_bi_cnot} and~\eqref{formula_bi_hadamard}, one can show that
\bb
\rho_{t+1}=\sum_{m,n=0}^{1}\alpha^{(t+1)}_{mn}\ketbra{\psi_{mn}}\,,
\ee
where the coefficients $\alpha^{(t+1)}_{mn}$ are given by~\eqref{coeff_p1} if $\alpha^{(t)}_{10}<\alpha^{(t)}_{01}$, and by~\eqref{coeff_p2} otherwise.

         \item \textbf{Step 4.3}: Let $t=t+1$.
         \item \textbf{Step 4.4}: If the condition $t<k$ is satisfied, then go back to Step 4.1.
     \end{itemize}
 	Before introducing Step 5, let us recall that if the improved hashing protocol of~\cite{Improvement-Hashing} is applied on states of the form $\rho= \sum_{ij=0}^1\alpha_{ij}\ketbra{\psi_{ij}}$ then it can generate ebits with a yield $Y(\alpha_{00},\alpha_{01},\alpha_{10},\alpha_{11})$ given by~\eqref{improv_hashing}.
   	Note that such a yield is not invariant under permutations of the variables $\alpha_{00},\alpha_{01},\alpha_{10},\alpha_{11}$. Hence, one may achieve a yield which is larger than $Y(\alpha_{00},\alpha_{01},\alpha_{10},\alpha_{11})$ by applying suitable bi-local unitaries, which suitably permutes the Bell states, just before running the improved hashing protocol. Since the yield function $Y(\alpha_{00},\alpha_{01},\alpha_{10},\alpha_{11})$ satisfies
   	\bb
   	Y(\alpha_{00},\alpha_{01},\alpha_{10},\alpha_{11})&=Y(\alpha_{10},\alpha_{01},\alpha_{00},\alpha_{11})\,,\\
   	Y(\alpha_{00},\alpha_{01},\alpha_{10},\alpha_{11})&=Y(\alpha_{00},\alpha_{11},\alpha_{10},\alpha_{01})\,,\\
   	Y(\alpha_{00},\alpha_{01},\alpha_{10},\alpha_{11})&=Y(\alpha_{01},\alpha_{00},\alpha_{11},\alpha_{10})\,,\\
   	\ee
   	then by permuting the four variables $\alpha_{ij}$ it is possible to obtain at most three different values of the rate function, which are: $Y(  \alpha_{00},\alpha_{01},\alpha_{10},\alpha_{11})$, $Y(  \alpha_{00},\alpha_{10},\alpha_{01},\alpha_{11})$, and $Y(  \alpha_{01},\alpha_{00},\alpha_{10},\alpha_{11})$.
   	Let us define the function $\mathcal{I}$ as
   	\bb
   		&\mathcal{I}(  \alpha_{00},\alpha_{01},\alpha_{10},\alpha_{11})\coloneqq \max\left(\,Y(  \alpha_{00},\alpha_{01},\alpha_{10},\alpha_{11}),   Y(  \alpha_{00},\alpha_{10},\alpha_{01},\alpha_{11} ),\, Y(  \alpha_{01},\alpha_{00},\alpha_{10},\alpha_{11})\right)\,.
   	\ee
Note that at the beginning of Step 5, the number of remaining pairs is
\begin{equation}
	\approx n_0\,\mathcal{C}(\Phi,c,M) \frac{1}{2^{k}} \prod_{t=0}^{k-1}P_t\,
\end{equation}
and each of these is in $\rho_{k}=\sum_{m,n=0}^{1}\alpha^{(k)}_{mn}\ketbra{\psi_{mn}}$.

 \textbf{-Step 5}:  If $\mathcal{I}(  \alpha^{(k)}_{00},\alpha^{(k)}_{01},\alpha^{(k)}_{10},\alpha^{(k)}_{11})= Y(  \alpha^{(k)}_{00},\alpha^{(k)}_{10},\alpha^{(k)}_{01},\alpha^{(k)}_{11})$, then both Alice and Bob apply the Hadamard gate defined by~\eqref{hadamard}. Therefore, in this case, the state of each of pairs becomes           \bb
\left(H\otimes H\right)\rho_{k} \left(H\otimes H\right)^\dagger&=\alpha^{(k)}_{00}\ketbra{\psi_{00}}+\alpha^{(k)}_{10}\ketbra{\psi_{01}} +\alpha^{(k)}_{01}\ketbra{\psi_{10}}+\alpha^{(k)}_{11}\ketbra{\psi_{11}}\,,
\ee
where we have exploited~\eqref{formula_bi_hadamard}.
If $\mathcal{I}(  \alpha^{(k)}_{00},\alpha^{(k)}_{01},\alpha^{(k)}_{10},\alpha^{(k)}_{11})= Y(  \alpha^{(k)}_{01},\alpha^{(k)}_{00},\alpha^{(k)}_{10},\alpha^{(k)}_{11})$, then both Alice and Bob apply $B_x\coloneqq \frac{\mathbb{1}_2-i\sigma_{01}}{\sqrt{2}}$, where $\sigma_{01}$ is defined by~\eqref{def_pauli}, and hence the state becomes
\bb
\left(B_x\otimes B_x\right)\rho_{k} \left(B_x\otimes B_x\right)^\dagger&=\alpha^{(k)}_{01}\ketbra{\psi_{00}}+\alpha^{(k)}_{00}\ketbra{\psi_{01}} +\alpha^{(k)}_{10}\ketbra{\psi_{10}}+\alpha^{(k)}_{11}\ketbra{\psi_{11}}\,.
\ee

\textbf{-Step 6}: Alice and Bob run the improved hashing protocol of~\cite{Improvement-Hashing}, which can achieve the yield $\mathcal{I}(  \alpha^{(k)}_{00},\alpha^{(k)}_{01},\alpha^{(k)}_{10},\alpha^{(k)}_{11})$. Hence, in the end, Alice and Bob can generate a number of ebits equal to 
\begin{equation}
	\approx n_0\,\mathcal{C}(\Phi,c,M)  \frac{\prod_{t=0}^{k-1}P_t}{2^{k}}\mathcal{I}(  \alpha^{(k)}_{00},\alpha^{(k)}_{01},\alpha^{(k)}_{10},\alpha^{(k)}_{11})\,.
\end{equation}
Since the channel $\Phi$ is used $n_0$ times (during Step 1) to send the $n_0$ halves of the state $\ket{\Psi_{M,c}}_{AA'}$, the rate of distributed ebits of the presented protocol is 
\begin{equation}\label{rate_gen2}
\mathcal{C}(\Phi,c,M)  \frac{\prod_{t=0}^{k-1}P_t}{2^{k}}\mathcal{I}(  \alpha^{(k)}_{00},\alpha^{(k)}_{01},\alpha^{(k)}_{10},\alpha^{(k)}_{11})\,.
\end{equation}    
Since the local mean photon number of $\ket{\Psi_{M,c}}_{AA'}$ is $(1-c^2)M$, the rate in~\eqref{rate_gen2} is a lower bound on the energy-constrained two-way quantum capacity $Q_2(\Phi,N_s)$ for all $M\in\N^+$, $c\in(0,1)$, $k\in\N$ such that $(1-c^2)M\le N_s$. The optimisation over these parameters of the rate in~\eqref{rate_gen2} leads to the lower bound on $Q_2(\Phi,N_s)$ in~\eqref{lowQ2_genEC}. In addition, since $K(\Phi,N_s)\ge Q_2(\Phi,N_s)$ thanks to~\eqref{relation_2waycapEC}, we have proved~\eqref{lowQ2_genEC}. By taking the limit $N_s\rightarrow\infty$ of~\eqref{lowQ2_genEC}, the lower bound on the unconstrained two-way capacities in~\eqref{lowQ2_gen} is also proved.
\end{proof}
In the forthcoming Theorem~\ref{th_comp_Q2} we apply Theorem~\ref{th_lower_Q2} to the composition $\mathcal{N}_{g,\lambda}\coloneqq\Phi_{g,0}\circ \mathcal{E}_{\lambda,0}$ between pure amplifier channel $\Phi_{g,0}$ and pure loss channel $\mathcal{E}_{\lambda,0}$.
\begin{thm}\label{th_comp_Q2}
 Let $g\ge1$, $\lambda\in[0,1]$, and $N_s\ge0$. The EC two-way capacities $Q_2(\mathcal{N}_{g,\lambda},N_s)$ and $K(\mathcal{N}_{g,\lambda},N_s)$ of the composition $\mathcal{N}_{g,\lambda}\coloneqq\Phi_{g,0}\circ \mathcal{E}_{\lambda,0}$ between pure amplifier channel and pure loss channel satisfy the following lower bound
\bb\label{lowQ2_deltaEC_comp}
	K(\mathcal{N}_{g,\lambda},N_s)& \ge  Q_2(\mathcal{N}_{g,\lambda},N_s)\ge \sup_{\substack{c\in(0,1),\, M\in\N^+,\, k\in\N\\ (1-c^2)M\le N_s}} \mathcal{R}(g,\lambda,M,c,k)\,,
\ee
and, in particular, the unconstrained two-way capacities satisfy
\bb\label{lowQ2_delta_comp}
	K(\mathcal{N}_{g,\lambda})& \ge  Q_2(\mathcal{N}_{g,\lambda})\ge \sup_{c\in(0,1),\, M\in\N^+,\, k\in\N} \mathcal{R}(g,\lambda,M,c,k)\,,
\ee
where
\bb\label{def_mathcal_R}
    \mathcal{R}(g,\lambda,M,c,k)\coloneqq R(\mathcal{N}_{g,\lambda},M,c,k)\,,
\ee
with the quantity $R(\mathcal{N}_{g,\lambda},M,c,k)$ being defined in Theorem~\ref{th_lower_Q2}. The quantities $\mathcal{C}(\mathcal{N}_{g,\lambda},c,M)$ and $\alpha_{mn}^{(0)}$, which appear in the definition of $R(\mathcal{N}_{g,\lambda},M,c,k)$ in Theorem~\ref{th_lower_Q2}, can be expressed as
\bb\label{formula_N_alpha_0_comp}
\mathcal{C}(\mathcal{N}_{g,\lambda},c,M) &\coloneqq \sum_{n,l=0}^{1} c_n^2\, f_{M n,M n,M l}(g,\lambda)\,,\\        \alpha_{mn}^{(0)}&\coloneqq\frac{1}{2\mathcal{C}(\mathcal{N}_{g,\lambda},c,M)}\sum_{x,y=0}^{1}\sum_{l=\max(y-x,0)}^{1+\min(y-x,0)} \delta_{x\oplus n,l+x-y}\,\delta_{y\oplus n,l}\,  (-1)^{m(x+y)} c_x c_y\, f_{ M x,M y,M l}(g,\lambda)\,,
    \ee
where $c_0\coloneqq c$, $c_1\coloneqq \sqrt{1-c^2}$, $f_{n,i,l}(g,\lambda)$ is defined in~\eqref{def_f_comp}, and $\delta_{x,y}$ denotes the Kronecker delta.
\end{thm}
\begin{proof} 
\eqref{lowQ2_deltaEC_comp} and~\eqref{lowQ2_delta_comp} follows by applying Theorem~\ref{th_lower_Q2} to $\mathcal{N}_{g,\lambda}$. We only need to show the expressions of $\mathcal{C}(\mathcal{N}_{g,\lambda},c,M)$ and $\alpha_{mn}^{(0)}$ in~\eqref{formula_N_alpha_0_comp}. In this proof we use the notation introduced in the statement of Theorem~\ref{th_lower_Q2}. By using~\eqref{action_comp_chan}, we deduce that
\bb\label{formula_choi}
&\Id_{A}\otimes\mathcal{N}_{g,\lambda}(\ketbra{\Psi_{M,c}})=\sum_{n,i=0}^{1}\sum_{l=M\max(i-n,0)}^\infty c_n c_i f_{M n,M i,l}(g,\lambda)\ketbraa{M n}{M i}_{A}\otimes\ketbraa{l+M(n-i)}{l}_B.
\ee
Consequently, it holds that
\bb\label{step_proof}
        &\mathbb{1}_A\otimes\Pi_M\, \Id_{A}\otimes\mathcal{N}_{g,\lambda}(\ketbra{\Psi_{M,c}})\, \mathbb{1}_A\otimes\Pi_M=\sum_{n,i=0}^{1}\sum_{l=\max(i-n,0)}^{1+\min(i-n,0)} c_n c_i  f_{M n,M i,M l}(g,\lambda)\ketbraa{M n}{M i}_{A}\otimes\ketbraa{M(l+n-i)}{M l}_B.
\ee
By inserting this into the definition of $\mathcal{C}(\mathcal{N}_{g,\lambda},c,M)$ in~\eqref{formula_gen_norm} and of $\alpha_{mn}^{(0)}$ in~\eqref{formula_gen_alpha0}, one obtains the expressions in~\eqref{formula_N_alpha_0_comp}.
\end{proof}

\subsection{Remarks}
Let us consider the entanglement distribution protocol shown in the proof of Theorem~\ref{th_lower_Q2} applied to the the composition $\mathcal{N}_{g,\lambda}\coloneqq\Phi_{g,0}\circ \mathcal{E}_{\lambda,0}$ between pure amplifier channel $\Phi_{g,0}$ and pure loss channel $\mathcal{E}_{\lambda,0}$. After completing Step 2 of this protocol, the entanglement distribution process is reduced to an entanglement distillation protocol on the two-qubit state reported in~\eqref{rho_primo_distil}. We will denote this two-qubit state as $\rho^{(g,\lambda,M,c)}_{AB}$, where $M\in\N^+$ and $c\in(0,1)$ correspond to the constants appearing in the state in~\eqref{initial_state} that Alice produces during Step 1. The natural question that arises is: "Under what conditions is $\rho^{(g,\lambda,M,c)}_{AB}$ distillable?" In Remark~\ref{remark_alternative_proof} we answer this question.
\begin{remark}\label{remark_alternative_proof}
 $\rho^{(g,\lambda,M,c)}_{AB}$ is distillable if and only if $\lambda$ and $g$ satisfy the inequality $(1-\lambda) g<1$, meaning that $\mathcal{N}_{g,\lambda}$ is not entanglement breaking. This provides an alternative proof of Theorem~\ref{th1}.
\end{remark}
\begin{proof}
    By exploiting~\eqref{step_proof}, for all $g>1,\lambda\in(0,1),M\in\N^+,c\in(0,1)$ the state in~\eqref{rho_primo_distil} can be expressed as
\bb
    \rho^{(g,\lambda,M,c)}_{AB}\coloneqq\frac{ \sum_{n,i=0}^{1}\sum_{l=\max(i-n,0)}^{1+\min(i-n,0)} c_n c_i  f_{M n,M i,M l}(g,\lambda)\ketbraa{M n}{M i}_{A}\otimes\ketbraa{M(l+n-i)}{M l}_B }{ \sum_{n,l=0}^{1} c_n^2\, f_{M n,M n,M l}(g,\lambda)}\,,
\ee
where $c_0\coloneqq c$, $c_1\coloneqq \sqrt{1-c^2}$, and $f_{n,i,l}(g,\lambda)$ is defined in~\eqref{def_f_comp}. Consequently, it holds that
\bb\label{explicit_state_f}
    \rho^{(g,\lambda,M,c)}_{AB}=\frac{ 1}{ \sum_{n,l=0}^{1} c_n^2\, f_{M n,M n,M l}(g,\lambda)}\big[&c^2f_{0,0,0}(g,\lambda)\ketbraa{0}{0}_{A}\otimes\ketbraa{0}{0}_B\\&+c^2f_{0,0,M}(g,\lambda)\ketbraa{0}{0}_{A}\otimes\ketbraa{M}{M}_B\\&+c\sqrt{1-c^2}f_{0,M,M}(g,\lambda)\ketbraa{0}{M}_{A}\otimes\ketbraa{0}{M }_B\\&+c\sqrt{1-c^2}f_{M,0,0}(g,\lambda)\ketbraa{M }{0}_{A}\otimes\ketbraa{M}{0}_B\\&+(1-c^2)f_{M,M,0}(g,\lambda)\ketbraa{M}{M}_{A}\otimes\ketbraa{0}{0}_B\\&+(1-c^2)f_{M,M,M}(g,\lambda)\ketbraa{M}{M}_{A}\otimes\ketbraa{M}{M}_B\big]\,,
\ee
Hence, the matrix associated with the partial transpose on $B$ of $\rho^{(g,\lambda,M,c)}_{AB}$, written with respect the basis $\{\ket{0}_A\otimes\ket{0}_B, \ket{0}_A\otimes\ket{M}_B, \ket{M}_A\otimes\ket{0}_B, \ket{M}_A\otimes\ket{M}_B\}$, is 
\bb\nonumber
&\frac{ 1}{ \sum_{n,l=0}^{1} c_n^2\, f_{M n,M n,M l}(g,\lambda)}\left(\begin{matrix} c^2f_{0,0,0}(g,\lambda)\quad & 0& 0& 0\\ 0 & c^2f_{0,0,M}(g,\lambda)\quad & c\sqrt{1-c^2}f_{0,M,M}(g,\lambda)\quad& 0\\
0 & c\sqrt{1-c^2}f_{M,0,0}(g,\lambda)\quad & (1-c^2)f_{M,M,0}(g,\lambda)\quad & 0\\ 0 & 0 & 0 &(1-c^2)f_{M,M,M}(g,\lambda)\quad \end{matrix}\right) \,.
\ee 
It follows that $\rho^{(g,\lambda,M,c)}_{AB}$ is not PPT if and only if 
\bb\label{condition_det}
    f_{M,0,0}(g,\lambda)\,f_{M,M,0}(g,\lambda)< f_{0,M,M}(g,\lambda)\, f_{0,0,M}(g,\lambda)\,.
\ee
The definition of $f_{\cdot,\cdot,\cdot}(g,\lambda)$ in~\eqref{def_f_comp} yields
\bb\label{formulae_f_g_lam}
    f_{0,0,M}(g,\lambda)&=\frac{(g-1)^M}{g^{1+M}}\,,\\
    f_{M,M,0}(g,\lambda)&=\frac{(1-\lambda)^M}{g}\,,\\
    f_{0,M,M}(g,\lambda)&=f_{M,0,0}(g,\lambda)=\frac{\lambda^{\frac{M}{2}}}{g^{1+\frac{M}{2}}}\,.
\ee
Consequently,~\eqref{condition_det} establishes that $\rho^{(g,\lambda,M,c)}_{AB}$ is not PPT if and only if $(1-\lambda)g<1$, independentely of $c$ and $M$. The fact that any two-qubit state is distillable if and only if it is not PPT~\cite{2-qubit-distillation} implies that $\rho^{(g,\lambda,M,c)}_{AB}$ is distillable if and only if $(1-\lambda)g<1$ for all $c\in(0,1)$ and all $M\in\N^+$. 

Let us now show that this fact constitutes an alternative proof of Theorem~\ref{th1}, i.e.~let us show that the energy-constrained two-way capacities $Q_2(\mathcal{N}_{g,\lambda},N_s)$ and $K(\mathcal{N}_{g,\lambda},N_s)$ are strictly positive if and only if $(1-\lambda)g< 1$, i.e.~if and only if $\mathcal{N}_{g,\lambda}$ is not entanglement breaking. The entanglement distribution protocol's Steps S1 and S2 imply that for all $N_s\ge 0$ it holds that $Q_2(\mathcal{N}_{g,\lambda},N_s)\ge E_d\left(  \rho^{(g,\lambda,M,c)}_{AB} \right)$, for any $c\in(0,1)$ and $M\in\N^+$ satisfying $(1-c^2)M\le N_s$. Here, $E_d(\cdot)$ denotes the distillable entanglement. As we have proved above, if $(1-\lambda)g<1$ then the state $\rho^{(g,\lambda,M,c)}_{AB}$ is distillable, i.e.~$E_d(\rho^{(g,\lambda,M,c)}_{AB})>0$. This implies that if $(1-\lambda)g<1$, then the energy-constrained two-way capacities of $\mathcal{N}_{g,\lambda}$ are strictly positive, i.e.~$K(\mathcal{N}_{g,\lambda},N_s)\ge Q_2(\mathcal{N}_{g,\lambda},N_s)>0$. 
Conversely, by exploiting Theorem~\ref{lemma_eb_comp} and the fact that any entanglement-breaking channel has vanishing two-way capacities, it follows that if $(1-\lambda)g\ge1$ then $K(\mathcal{N}_{g,\lambda})= Q_2(\mathcal{N}_{g,\lambda})=0$ and hence $K(\mathcal{N}_{g,\lambda},N_s)= Q_2(\mathcal{N}_{g,\lambda},N_s)=0$.
\end{proof}

{Remark~\ref{remark_alternative_proof} ensures that if the channel $\mathcal{N}_{g,\lambda}$ is not entanglement breaking, then the state $\rho^{(g,\lambda,M,c)}_{AB}$ obtained at the end of Step~2 is distillable for any $M\in\N^+$ and $c\in(0,1)$. We now turn our attention to the state, denoted as $\sigma^{(g,\lambda,M,c)}_{AB}$, which is obtained at the end of Step 3 through Pauli-based twirling of $\rho^{(g,\lambda,M,c)}_{AB}$. It is possible for this operation to map distillable states to undistillable states, so we ask the question: "Under what conditions is $\sigma^{(g,\lambda,M,c)}_{AB}$ distillable?" In Remark~\ref{remark_step3} we will demonstrate that for any $\lambda\in(0,1)$ and $g>1$, if $\mathcal{N}_{g,\lambda}$ is not entanglement breaking, then for all $M\in\N^+$ the state $\sigma^{(g,\lambda,M,\bar{c})}_{AB}$ is distillable, where $\bar{c}\coloneqq\frac{1}{\sqrt{1+(g-1)^M}}$. This means that Alice and Bob can choose the value of $c$ appropriately such that the Pauli-based twirling does not affect the distillability of the shared state.}
{\begin{remark}\label{remark_step3}
Let $M\in\N^+$, $\lambda\in(0,1)$, and $g>1$ with $(1-\lambda)g<1$ (meaning that $\mathcal{N}_{g,\lambda}$ is not entanglement breaking). Then, the state $\sigma^{(g,\lambda,M,\bar{c})}_{AB}$ is distillable, where $\bar{c}\coloneqq\frac{1}{\sqrt{1+(g-1)^M}}$.
    \end{remark}
    \begin{proof}
    After applying the Pauli-based twirling on the state $\rho^{(g,\lambda,M,\bar{c})}_{AB}$, the resulting state $\sigma^{(g,\lambda,M,\bar{c})}_{AB}$ is transformed into a Bell-diagonal form, that is
    \bb
        \sigma^{(g,\lambda,M,\bar{c})}_{AB}=\sum_{i,j=0}^1 p_{ij} \ketbraa{\psi^{(M)}_{ij}}{\psi^{(M)}_{ij}}_{AB}\,,
    \ee
    where $p_{ij}\coloneqq \bra{\psi^{(M)}_{ij}}\rho^{(g,\lambda,M,\bar{c})}_{AB}\ket{\psi^{(M)}_{ij}}$ and $\{\ket{\psi^{(M)}_{ij}}_{AB}\}_{i,j\in\{0,1\}}$ are the Bell states defined in~\eqref{Bell_states_delta}. In particular, it holds that
    \bb\label{prob_bell_diag_state}
        p_{00}+p_{10}&=\bra{0}_{A}\bra{0}_{B}\rho^{(g,\lambda,M,c)}_{AB}\ket{0}_{A}\ket{0}_{B}+\bra{M}_{A}\bra{M}_{B}\rho^{(g,\lambda,M,c)}_{AB}\ket{M}_{A}\ket{M}_{B}\,,\\
        p_{01}-p_{11}&=2\bra{0}_{A}\bra{0}_{B}\rho^{(g,\lambda,M,c)}_{AB}\ket{M}_{A}\ket{M}_{B}\,,\\
        p_{01}+p_{11}&=\bra{0}_{A}\bra{M}_{B}\rho^{(g,\lambda,M,c)}_{AB}\ket{0}_{A}\ket{M}_{B}+\bra{M}_{A}\bra{0}_{B}\rho^{(g,\lambda,M,c)}_{AB}\ket{M}_{A}\ket{0}_{B}\,,\\
        p_{00}-p_{10}&=2\bra{0}_{A}\bra{0}_{B}\rho^{(g,\lambda,M,c)}_{AB}\ket{M}_{A}\ket{M}_{B}\,.
    \ee
    Lemma~\ref{lemma_bell_diag} guarantees that if $p_{01}+p_{11}-|p_{00}-p_{10}|<0$ then the state $\sigma^{(g,\lambda,M,\bar{c})}_{AB}$ is distillable.  By using~\eqref{explicit_state_f} and~\eqref{prob_bell_diag_state}, the condition $p_{01}+p_{11}-|p_{00}-p_{10}|<0$ is satisfied if and only if
    \bb
        \bar{c}^2f_{0,0,M}(g,\lambda)+(1-\bar{c}^2)f_{M,M,0}(g,\lambda)-2\bar{c}\sqrt{1-\bar{c}^2}f_{M,0,0}(g,\lambda)<0\,,
    \ee
    that is
    \bb
        \bar{c}^2(g-1)^M+(1-\bar{c}^2)(1-\lambda)^Mg^M-2\bar{c}\sqrt{1-\bar{c}^2}(\lambda g)^{M/2}<0\,,
    \ee
    where we have exploited~\eqref{formulae_f_g_lam}. By hypothesis, the channel $\mathcal{N}_{g,\lambda}$ is entanglement breaking and hence $(1-\lambda)g<1$, as established by Lemma~\ref{lemma_comp_bos}. Consequently, for all $g>1$ and $\lambda\in(0,1)$ it holds that 
    \bb
        \bar{c}^2(g-1)^M+(1-\bar{c}^2)(1-\lambda)^Mg^M-2\bar{c}\sqrt{1-\bar{c}^2}(\lambda g)^{M/2}&<\bar{c}^2(g-1)^M+(1-\bar{c}^2)-2\bar{c}\sqrt{1-\bar{c}^2}(g-1)^{M/2}\\&=\left( \bar{c}(g-1)^{M/2}-\sqrt{1-\bar{c}^2} \right)^2=0\,,
    \ee
    where we have used that $\bar{c}\coloneqq\frac{1}{\sqrt{1+(g-1)^M}}$. Hence, for all $M\in \N^+$, $\lambda\in(0,1)$, and $g>1$ with $(1-\lambda)g<1$, it holds that $\sigma^{(g,\lambda,M,\bar{c})}_{AB}$ is distillable.
    \end{proof}

\begin{lemma}\label{lemma_bell_diag}
    Let $\{\ket{\psi_{ij}}\}_{i,j\in\{0,1\}}$ be the Bell states defined in~\eqref{Bell_states}. A convex combination of Bell states $\rho_{AB}=\sum_{i,j=0}^1p_{ij}\ketbraa{\psi_{ij}}{\psi_{ij}}$ is distillable if and only if $p_{00}+p_{10}<|p_{01}-p_{10}|$ or $p_{01}+p_{11}<|p_{00}-p_{10}|$.
\end{lemma}
\begin{proof}
    The matrix associated with $\rho_{AB}=\sum_{i,j=0}^1p_{ij}\ketbraa{\psi_{ij}}{\psi_{ij}}$, written with respect the basis $\{\ket{0}_A\otimes\ket{0}_B, \ket{0}_A\otimes\ket{1}_B, \ket{1}_A\otimes\ket{0}_B, \ket{1}_A\otimes\ket{1}_B\}$, is 
\bb\nonumber
&\frac{ 1}{ 2}\left(\begin{matrix} p_{00}+p_{10}\quad & 0& 0& p_{00}-p_{10}\\ 0 & p_{10}+p_{11}\quad & p_{10}-p_{11}\quad& 0\\
0 & p_{10}-p_{11}\quad & p_{10}+p_{11}\quad & 0\\ p_{00}-p_{10} & 0 & 0 &p_{00}+p_{10}\quad \end{matrix}\right) \,.
\ee 
Its partial transpose on $B$ is
\bb\nonumber
&\frac{ 1}{ 2}\left(\begin{matrix} p_{00}+p_{10}\quad & 0& 0& p_{10}-p_{11}\\ 0 & p_{10}+p_{11}\quad & p_{00}-p_{10}\quad& 0\\
0 & p_{00}-p_{10}\quad & p_{10}+p_{11}\quad & 0\\ p_{10}-p_{11} & 0 & 0 &p_{00}+p_{10}\quad \end{matrix}\right) \,.
\ee 
Hence, the state $\rho_{AB}$ is PPT if and only if $p_{00}+p_{10}\ge|p_{01}-p_{10}|$ and $p_{01}+p_{11}\ge|p_{00}-p_{10}|$. Consequently, the fact that any two-qubit state is distillable if and only if it is not PPT~\cite{2-qubit-distillation} implies the validity of the thesis.
\end{proof}

}

\subsection{Experimental challenges regarding our protocol}\label{SM_experimental_challenges}
As demonstrated in the main text, applying our main result regarding the maximum tolerable excess noise to the current Internet infrastructure shows that continuous-variable quantum key distribution is feasible if and only if the fibre length is approximately less than $1000 $ kilometres. Hence, any practical QKD protocol, which is based on the existing Internet infrastructure, must adhere to this fundamental limit. Furthermore, this limit of $1000 $ kilometres can now serve as a benchmark for evaluating the quality of any new CV-QKD protocol, underscoring the significant impact of our results on practical implementations.

The potential benefit of our protocol (presented both in the main text and in the proof of Theorem~\ref{th_lower_Q2} above) lies in its \emph{faithfulness} --- it can distil entanglement (and hence generate secret keys) whenever the channel is not entanglement breaking. With the current Internet infrastructure based on optical fibres, our protocol could \emph{theoretically} achieve the ultimate limit set by quantum physics of transmitting entanglement and secret keys over distances up to $1000$ kilometres. This is a unique feature of our protocol, which stands in stark contrast with \emph{all} existing entanglement distribution and key distribution protocols.

While there exist CV-QKD protocols that are relatively easy to implement with current technology (capable of distributing secret keys across optical fibres of at most $200$ kilometres~\cite{Record1,Record2,Record3,Record4,Record5}), this is not the case for entanglement distribution. Indeed, \emph{all} known entanglement-distribution protocols are experimentally challenging with current technology. For example, the best known entanglement-distribution protocol prior to our work~\cite{Pirandola2009} --- i.e.~the hashing protocol applied to the Choi state of the channel --- is not experimentally feasible. 

Our protocol is an entanglement-distribution protocol and, as such, is experimentally challenging at present. We emphasise that this limitation is not unique to our protocol but is a common challenge faced by \emph{all} entanglement distribution protocols due to current technological constraints.  A major factor is the lack of a noiseless quantum memory, which makes it challenging to perform even a few iterations of a recurrence entanglement distillation protocol. Nevertheless, given the significant recent experimental advancements regarding quantum memories~\cite{q_memory1, q_memory2} and entanglement distillation~\cite{Kalb_2017,Hu_2021,Ecker_2021}, we are optimistic about the future experimental viability of our protocol. Hence, we stress that, although our protocol is experimentally challenging with current technology, there is no way that a protocol as simple as our ours will \emph{not} be realisable in a few decades at worst. 
 
Let us provide further details about a possible practical realisation of our protocol. To perform Step 1, it suffices that Alice produces the state
\bb\label{state}
    \ket{\Psi}_{\!AA'}\coloneqq \frac{\ket{0}_{\!A}\!\otimes\!\ket{0}_{\!A'}+\ket{1}_{\!A}\!\otimes\!\ket{1}_{\!A'}}{\sqrt{2}}\,
\ee
in order to make the rate of the protocol faithful. However, without changing the rate (as explained below), Alice can instead produce the \emph{NOON state}~\cite{Sanders1989}
\bb\label{noon}
    \ket{\Psi'}_{\!AA'}\coloneqq \frac{\ket{0}_{\!A}\!\otimes\!\ket{1}_{\!A'}+\ket{1}_{\!A}\!\otimes\!\ket{0}_{\!A'}}{\sqrt{2}}\,,
\ee
which can be experimentally prepared~\cite{exp_noon1,exp_noon2}.
After Alice has sent the sub-system $A'$ through the channel to Bob, Step 2 involves performing a non-demolition measurement with the POVM operator $\ketbra{0}_B+\ketbra{1}_B$. Although this measurement appears challenging to implement experimentally, fortunately the problem has been studied already, and several promising approaches do exist. Specifically, one may exploit either: the pre-certification scheme employed in \cite{measurement1}; the coupling scheme between optical signals and trapped cold atomic gas designed in \cite{measurement2}; single photon filters based on Rydberg blockade~\cite{rydberg1,rydberg2} to implement single photon subtraction~\cite{rydberg3,rydberg4}; single atoms inside an optical cavity to perform single photon subtraction~\cite{substraction1,substraction2}. After this non-demolition measurement, Alice and Bob share a two-mode state in the subspace spanned by $\{\ket{0}\otimes\ket{0},\ket{0}\otimes\ket{1},\ket{1}\otimes\ket{0}, \ket{1}\otimes\ket{1}\}$. At this point, they can transfer their state from the optical modes to a qubit solid-state platform (e.g., superconducting or trapped ion platform). This transfer can be experimentally performed in several ways, for example by exploiting: the quantum-memory based approaches introduced in~\cite{q_memory1,q_memory2}; the aforementioned single photon subtraction methods~\cite{rydberg1,rydberg2,rydberg3,rydberg4,substraction1,substraction2}, which map the photonic state onto the atomic state; quantum transduction from optical to microwave photons that are compatible with the superconducting qubits~\cite{transf1,transf2,transf3}.  This means that only the first two steps of the protocol involve optical platforms, which is advantageous because all the two-qubit unitaries used in the subsequent steps of our protocol are much easier to experimentally implement in a qubit platform. To address the fact that Alice has sent the NOON state in \eqref{noon} instead of the state in \eqref{state}, she simply needs to apply the Pauli $\sigma_x$ before initiating Step~3.

\section{Multi-rail strategies}\label{multiplerail_section}
In this section we introduce an additional protocol for distributing ebits across the piBGC $\mathcal{N}_{g,\lambda}$ by combining and optimising the multi-rail protocol introduced in~\cite{Winnel} and the qudit P1-or-P2 protocol introduced in~\cite{p1orp2}. To begin, we will establish some notation and we will prove a useful lemma. For any $K\in\N$ with $K\geq 2$ and any $\textbf{n}\coloneqq(n_1,\ldots,n_K)\in\N^K$, we denote as $\ket{\textbf{n}}_{A_1\ldots A_K}$ the following $K$-mode Fock state with total photon number equal to $\|\textbf{n}\|_1$:
\bb
        \ket{\textbf{n}}_{A_1\ldots A_K}\coloneqq \ket{n_1}_{A_1}\otimes\ket{n_2}_{A_2}\otimes\ldots\otimes\ket{n_K}_{A_K}\,,
    \ee
    where we have used the notation $\|\textbf{n}\|_1\coloneqq \sum_{j=1}^K n_j$. For any $N,K\in\N^+$ with $K\ge2$, let us order the set 
    \bb
    \{\ket{\textbf{n}}_{A_1\ldots A_K}:\,\textbf{n}\in\N^K\,,  \|\textbf{n}\|_1=N\}
    \ee
    according to the restricted lexicographic ordering. More formally, the relation $\preceq$ is defined as
    \bb
\ket{\textbf{n}}_{A_1\ldots A_K}\preceq\ket{\textbf{m}}_{A_1\ldots A_K}   \Longleftrightarrow \sum_{j=1}^{K} n_{j}\, (N+1)^j<\sum_{j=1}^{K} m_{j} \,(N+1)^j\,.
    \ee
     The set has $\binom{N+K-1}{N}$ elements, and for all $n=0,1,\ldots,\binom{N+K-1}{N}-1$, we define the state $\ket{\phi^{(N)}_n}_{A_1\ldots A_K}$ as the $n$th element of the ordered set. For example, if $N=2$ and $K=3$, we have that
    \bb
        \ket{\phi^{(2)}_0}_{A_1 A_2 A_3}&\coloneqq \ket{0}_{A_1}\otimes\ket{0}_{A_2}\otimes\ket{2}_{A_3}\,,\\
        \ket{\phi^{(2)}_1}_{A_1 A_2 A_3}&\coloneqq \ket{0}_{A_1}\otimes\ket{1}_{A_2}\otimes\ket{1}_{A_3}\,,\\
        \ket{\phi^{(2)}_2}_{A_1 A_2 A_3}&\coloneqq \ket{0}_{A_1}\otimes\ket{2}_{A_2}\otimes\ket{0}_{A_3}\,,\\
        \ket{\phi^{(2)}_3}_{A_1 A_2 A_3}&\coloneqq \ket{1}_{A_1}\otimes\ket{0}_{A_2}\otimes\ket{1}_{A_3}\,,\\
        \ket{\phi^{(2)}_4}_{A_1 A_2 A_3}&\coloneqq \ket{1}_{A_1}\otimes\ket{1}_{A_2}\otimes\ket{0}_{A_3}\,,\\
        \ket{\phi^{(2)}_5}_{A_1 A_2 A_3}&\coloneqq \ket{2}_{A_1}\otimes\ket{0}_{A_2}\otimes\ket{0}_{A_3}\,.
    \ee
    In addition, for all $N,K\in\N^+$ with $K\ge2$ let us define the following state of $K+K$ modes $A_1,\ldots,A_k,A_1',\ldots,A_K'$:
    \bb
        \ket{\Psi_{N,K}}_{A_1\ldots A_K, A'_1,\ldots,A'_K}&\coloneqq \frac{1}{\sqrt{\binom{N+K-1}{N}}}\sum_{n=0}^{\binom{N+K-1}{N}-1}\ket{\phi^{(N)}_n}_{A_1\ldots A_k}\otimes\ket{\phi^{(N)}_n}_{A'_1\ldots A'_k}\\&=\frac{1}{\sqrt{\binom{N+K-1}{N}}} \sum_{\substack{\textbf{n}\in\N^K\\ \|\textbf{n}\|_1=N}} \ket{\textbf{n}}_{A_1\ldots A_K}\otimes\ket{\textbf{n}}_{A'_1\ldots A'_K}\,,
    \ee
    which is a $\binom{N+K-1}{N}$-dimensional maximally entangled state 
    that corresponds to the subspace of the Hilbert space of $K$ modes with total photon number equal to $N$. Moreover, let us define for all $K,F\in\N$ the projector $\Pi^{(K)}_F$ onto the subspace of $K$ modes $A_1,\ldots,A_K$ whose total photon number equals $F$, i.e.
    \bb\label{PROJECTOR_N_photon}
        \Pi^{(K)}_F\coloneqq \sum_{\substack{\textbf{m}\in\N^K\\ \|\textbf{m}\|_1=F}} \ketbra{\textbf{m}}_{A_1\ldots A_K} \,.
    \ee
    The following lemma will be useful in order to calculate the rate of our entanglement distribution protocol.
    \begin{lemma}\label{Lemma_P_F}
        Let $\lambda\in[0,1],g\ge1, N\in\N,K\in\N^+$, and $\textbf{n}\in\N^K$ such that $\|\textbf{n}\|_1=N$. Assume that Alice transmits the $K$-mode Fock state $\ket{\textbf{n}}$ to Bob via $K$ parallel uses of the piBGC $\mathcal{N}_{g,\lambda}\coloneqq \Phi_{g,0}\circ\mathcal{E}_{\lambda,0}$ and suppose further that Bob measures the total photon number of the $K$ received modes. The probability $\mathcal{P}_F$ that Bob gets the outcome $F\in\N$ is 
        \bb\label{expr_pf_values}
            \mathcal{P}_F\coloneqq \Tr\left[\mathcal{N}_{g,\lambda}^{\otimes K}(\ketbra{\textbf{n}})\,\Pi^{(K)}_F\right]=\sum_{P=0}^{\min(F,N)}\binom{N}{P}\binom{K+F-1}{F-P}\lambda^{P}(1-\lambda)^{N-P}\frac{(g-1)^{F-P}}{g^{K+F}} \,.
        \ee
        In particular, note that $\mathcal{P}_F$ depends on $\textbf{n}$ only through the total photon number $\|\textbf{n}\|_1=N$.
        Specifically, if the communication channel is the pure loss channel $\mathcal{E}_{\lambda,0}=\mathcal{N}_{1,\lambda}$, the probability of getting the outcome $F\in\N$ is
        \bb
            \Tr\left[\mathcal{E}_{\lambda,0}^{\otimes K}(\ketbra{\textbf{n}})\,\Pi^{(K)}_F\right]&=\binom{N}{F}\lambda^{F}(1-\lambda)^{N-F} \Theta(N-F)\,,
         \ee   
         where we have introduced the Heaviside function $\Theta(x)$ defined as $\Theta(x)=1$ if $x\ge 0$, and $\Theta(x)=0$ if $x<0$.
         In addition, if the communication channel is the pure amplifier channel $\Phi_{g,0}=\mathcal{N}_{g,1}$, the probability of getting the outcome $F\in\N$ is
         \bb
            \Tr\left[\Phi_{g,0}^{\otimes K}(\ketbra{\textbf{n}})\,\Pi^{(K)}_F\right]&=\binom{K+F-1}{F-N} \frac{(g-1)^{F-N}}{g^{K+F}}\Theta(F-N)\,.
        \ee
        Therefore, the probability $\mathcal{P}_F$ in \eqref{expr_pf_values} of getting $F$ photons at the output of $K$ parallel uses of the composition between pure loss channel and pure amplifier channel can be expressed as the sum over $P\in\N$ of the conditional probability of getting $F$ photons at the output of the $K$ pure amplifier channels conditioned on the event of getting $P$ photons at the output of the $K$ pure loss channels, multiplied by the probability of the latter event.
    \end{lemma}
    \begin{proof}
        As a consequence of \eqref{action_ni_att}, for all $n\in\N$ it holds that
        \bb
            \mathcal{E}_{\lambda,0}(\ketbra{n})=\sum_{l=0}^n\binom{n}{l}\lambda^l(1-\lambda)^{n-l}\ketbra{l}
        \ee
        and hence 
        \bb
        \mathcal{E}_{\lambda,0}^{\otimes K}(\ketbra{\textbf{n}})=\sum_{\substack{\textbf{l}\in\N^K\\ \textbf{l}\le \textbf{n}}}\left(\prod_{j=1}^K\binom{n_j}{l_j}\right)\lambda^{\|\textbf{l}\|_1}(1-\lambda)^{N-\|\textbf{l}\|_1}\ketbra{\textbf{l}}\,,
        \ee
        where the inequality between vectors $\textbf{a}\ge \textbf{b}$ means that $a_j \ge b_j$ for all $j=1,\ldots,K$.
        Consequently, by using that $\sum_{P=0}^\infty \Pi^{(K)}_P=\mathbb{1}$, it holds that 
        \bb\label{calculation_P_F}
            \mathcal{P}_F&\coloneqq \Tr\left[\Phi_{g,0}^{\otimes K}\left(\mathcal{E}_{\lambda,0}^{\otimes K}(\ketbra{\textbf{n}})\right)\,\Pi^{(K)}_F\right]=\sum_{P,P'=0}^\infty\Tr\left[\Phi_{g,0}^{\otimes K}\left(\Pi^{(K)}_P\mathcal{E}_{\lambda,0}^{\otimes K}(\ketbra{\textbf{n}})\Pi^{(K)}_{P'}\right)\,\Pi^{(K)}_F\right]\\&=\sum_{P=0}^{N}\lambda^{P}(1-\lambda)^{N-P}\sum_{\substack{\textbf{l}\in\N^K\\ \textbf{l}\le \textbf{n}\\ \|\mathbf{l}\|_1=P}}\left(\prod_{j=1}^K\binom{n_j}{l_j}\right)\Tr\left[\Phi_{g,0}^{\otimes K}\left(\ketbra{\textbf{l}}\right)\,\Pi^{(K)}_F\right]\,.
        \ee
        Moreover, \eqref{action_ni_amp} implies that for all $l\in\N$ it holds that
        \bb
            \Phi_{g,0}(\ketbra{l})=\frac{1}{g^{l+1}}\sum_{m=0}^\infty \binom{l+m}{l}\left(\frac{g-1}{g}\right)^m\ketbra{m+l}
        \ee
        and hence
        \bb
            \Phi_{g,0}^{\otimes K}(\ketbra{\textbf{l}})=\frac{1}{g^{P+K}}\sum_{\textbf{m}\in\N^K} \left(\prod_{j=1}^K\binom{l_j+m_j}{l_j}\right)\left(\frac{g-1}{g}\right)^{\|\textbf{m}\|_1}\ketbra{\textbf{m}+\textbf{l}}\,.
        \ee
    Consequently, it holds that
    \bb\label{expression_with_sum}
        \Tr\left[\Phi_{g,0}^{\otimes K}\left(\ketbra{\textbf{l}}\right)\,\Pi^{(K)}_F\right]=\frac{(g-1)^{F-P}}{g^{K+F}}\sum_{\substack{\textbf{m}\in\N^K\\ \|\textbf{m}\|_1=F-P}} \prod_{j=1}^K\binom{l_j+m_j}{l_j}\,.
    \ee
    The sum 
    \bb
         \sum_{\substack{\textbf{m}\in\N^K\\ \|\textbf{m}\|_1=F-P}} \prod_{j=1}^K\binom{l_j+m_j}{l_j}\,,\qquad\quad
    \ee
    which appears in \eqref{expression_with_sum}, is the coefficient of the term $x^{F-P}$ of the power series $Q(x)$ in the variable $x\in(0,1)$ defined as
    \bb
        Q(x)\coloneqq\sum_{\substack{\textbf{m}\in\N^K}} \left(\prod_{j=1}^K\binom{l_j+m_j}{l_j}\right)x^{\|\textbf{m}\|_1}\,.
    \ee
    By exploiting that for all $l\in\N$ it holds that
    \bb\label{identity_pol}
    	\sum_{m=0}^\infty \binom{m+l}{m}x^m=\frac{1}{(1-x)^{l+1}}\,,
    \ee
    one obtains that
    \bb
        Q(x)=\sum_{\substack{\textbf{m}\in\N^K}} \left(\prod_{j=1}^K\binom{l_j+m_j}{l_j}\right)x^{\|\textbf{m}\|_1}&=\prod_{j=1}^K \left(\sum_{m=0}^\infty\binom{l_j+m}{m}x^m\right)=\prod_{j=1}^K\frac{1}{(1-x)^{l_j+1}}=\frac{1}{(1-x)^{P+K}}\\&=\sum_{m=0}^\infty\binom{P+K+m-1}{m}x^m\,.
    \ee
    It follows that
    \bb
        \sum_{\substack{\textbf{m}\in\N^K\\ \|\textbf{m}\|_1=F-P}} \prod_{j=1}^K\binom{l_j+m_j}{l_j}=\binom{K+F-1}{F-P}\Theta(F-P)
    \ee
    and hence
    \bb
        \Tr\left[\Phi_{g,0}^{\otimes K}\left(\ketbra{\textbf{l}}\right)\,\Pi^{(K)}_F\right]=\frac{(g-1)^{F-P}}{g^{K+F}}\binom{K+F-1}{F-P}\Theta(F-P)\,,
    \ee
    where we have introduced the Heaviside function $\Theta(x)$ defined as $\Theta(x)=1$ if $x\ge 0$, and $\Theta(x)=0$ if $x<0$.
    Consequently, \eqref{calculation_P_F} implies that
    \bb
        \mathcal{P}_F &= \sum_{P=0}^{\min(F,N)}\lambda^{P}(1-\lambda)^{N-P}\frac{(g-1)^{F-P}}{g^{K+F}}\binom{K+F-1}{F-P}\sum_{\substack{\textbf{l}\in\N^K\\ \textbf{l}\le \textbf{n}\\ \|\mathbf{l}\|_1=P}}\left(\prod_{j=1}^K\binom{n_j}{l_j}\right)\\&=\sum_{P=0}^{\min(F,N)}\binom{N}{P}\lambda^{P}(1-\lambda)^{N-P}\frac{(g-1)^{F-P}}{g^{K+F}}\binom{K+F-1}{F-P} \,,
    \ee
    where in the last equality we have exploited that
    \bb\label{identity_binomial_constr}
        \sum_{\substack{\textbf{l}\in\N^K\\ \textbf{l}\le \textbf{n}\\ \|\mathbf{l}\|_1=P}}\left(\prod_{j=1}^K\binom{n_j}{l_j}\right)=\binom{N}{P}\,.
    \ee
    This follows from the fact that the sum in \eqref{identity_binomial_constr} is equal to the coefficient of the term 
    $x^{P}$ of the following polynomial in the variable $x\in\mathbb{R}$:
    \bb
        \sum_{\substack{\textbf{l}\in\N^K\\ \textbf{l}\le \textbf{n}}}\left(\prod_{j=1}^K\binom{n_j}{l_j}\right)x^{\|\textbf{l}\|_1}=\prod_{j=1}^K(1+x)^{n_j}=(1+x)^N=\sum_{l=0}^N\binom{N}{l}x^{l}\,.
    \ee
    \end{proof}

    \begin{remark}
        Here we present an alternative method to calculate the probability $\mathcal{P}_F$ reported in \eqref{expr_pf_values}. For all $x\in(0,1)$ let us consider the tensor product of $K$ thermal states with mean photon number $\frac{x}{1-x}$, i.e.
        \bb
            \tau_{\frac{x}{1-x}}^{\otimes K}=(1-x)^K\sum_{\textbf{l}\in\N^K}x^{\|\textbf{l}\|_1}\ketbra{\textbf{l}}\,.
        \ee
        Consequently, the quantity
        \bb
            \mathcal{P}_F= \Tr\left[\mathcal{N}_{g,\lambda}^{\otimes K}(\ketbra{\textbf{n}})\,\Pi^{(K)}_F\right]
        \ee
        is the coefficient of the term $x^{F}$ of the power series $P(x)$ in the variable $x\in(0,1)$ defined as
    \bb\label{eq_p_x}
        P(x)\coloneqq \frac{1}{(1-x)^K}\Tr\left[\mathcal{N}_{g,\lambda}^{\otimes K}(\ketbra{\textbf{n}})\,\tau_{\frac{x}{1-x}}^{\otimes K}\right]\,.
    \ee
    By using the characteristic function properties reported in~\eqref{def_charact_func}, \eqref{inverse_fourier_displacement}, \eqref{transf_caract}, and the fact that the characteristic function of a thermal state $\tau_\nu$ is $\chi_{\tau_{\nu}}(\mathbf r)=e^{ -\frac{1}{4}(2\nu+1)|\mathbf{r}|^2}$, one obtains that for any single-mode state $\rho$ it holds that
    \bb
         \Tr\left[\mathcal{N}_{g,\lambda}(\rho)\,\tau_{\frac{x}{1-x}}\right]&=\int_{\mathbb{R}^{2}}\frac{\mathrm{d}^{2}\mathbf{r}}{2\pi}\chi_{\mathcal{N}_{g,\lambda}(\rho)}(\mathbf r)\,\chi_{\tau_{\frac{x}{1-x}}}(\mathbf r)=\int_{\mathbb{R}^{2}}\frac{\mathrm{d}^{2}\mathbf{r}}{2\pi} \chi_{\rho}(\sqrt{g\lambda}\,\mathbf{r})e^{-\frac{1}{4}\left(2g-g\lambda+2\frac{x}{1-x}\right)|\mathbf{r}|^2}\\&=\frac{1}{g\lambda}\int_{\mathbb{R}^{2}}\frac{\mathrm{d}^{2}\mathbf{r}}{2\pi} \chi_{\rho}(\mathbf{r})e^{-\frac{1}{4g\lambda}\left(2g-g\lambda+2\frac{x}{1-x}\right)|\mathbf{r}|^2}=\frac{1}{g\lambda}\Tr\left[\rho\,\,\tau_{\frac{g-g\lambda+(1+g\lambda-g)x}{g\lambda(1-x)}}\right]\,.
    \ee
    Hence, by exploiting \eqref{identity_pol} and the fact that $\|\textbf{n}\|_1=N$, the power series $P(x)$ can be expressed as 
    \bb
    		P(x)&=\frac{1}{(1-x)^K (g\lambda)^K}\Tr\left[ \ketbra{\textbf{n}} \tau_{\frac{g-g\lambda+(1+g\lambda-g)x}{g\lambda(1-x)}}^{\otimes K}  \right] =\frac{\left[g(1-\lambda)+(1+g\lambda-g)x\right]^{N}}{[g-(g-1)x]^{N+K}}\\&=\sum_{P=0}^N \binom{N}{P}(1+g\lambda-g)^P(1-\lambda)^{N-P}g^{-P-K}x^P\sum_{l=0}^\infty \binom{N+K-1+l}{l}\left(\frac{g-1}{g}\right)^l x^l\,.
    \ee
    It follows that
    \bb\label{expr_pf_values2}
    	\mathcal{P}_F=\sum_{P=0}^{\min(F,N)}\binom{N}{P}\binom{N+K+F-P-1}{F-P}(1+g\lambda-g)^P(1-\lambda)^{N-P}\frac{(g-1)^{F-P}}{g^{F+K}}\,.
    \ee
      Incidentally, by comparing the two expressions of $\mathcal{P}_F$ in \eqref{expr_pf_values} and \eqref{expr_pf_values2}, one deduces the following identity:
    \bb
    \sum_{P=0}^{\min(F,N)}\binom{N}{P}\binom{K+F-1+N-P}{F-P}\left(\frac{g\lambda -(g-1)}{(1-\lambda)(g-1)}\right)^P=\sum_{P=0}^{\min(F,N)}\binom{N}{P}\binom{K+F-1}{F-P}\left(\frac{\lambda}{(1-\lambda)(g-1)}\right)^P\,.
    \ee
    \end{remark}

    Let us now introduce an additional entanglement distribution protocol to distribute ebits across any piBGC $\mathcal{N}_{g,\lambda}$. The protocol depends on two parameters, $K,N\in\N^+$ with $K\ge2$, and it is composed of five steps named S1-S5, which we now outline.

\begin{enumerate}[\bf S1:]
    \item Alice prepares the state $\ket{\Psi_{N,K}}_{A_1\ldots A_K, A'_1,\ldots,A'_K}$ of $K+K$ modes $A_1,\ldots,A_k,A_1',\ldots,A_K'$, sending the systems $A'_1,\ldots,A'_K$ to Bob through $K$ uses of the channel $\mathcal{N}_{g,\lambda}$. Now Alice and Bob share the state $\Id_{A_1\ldots A_k}\otimes\mathcal{N}_{g,\lambda}^{\otimes K}(\ketbra{\Psi_{N,K}})$. By using~\eqref{action_comp_chan}, such a state can be expressed as
    \bb
        &\Id_{A_1\ldots A_k}\otimes\mathcal{N}_{g,\lambda}^{\otimes K}(\ketbra{\Psi_{N,K}})\\&=\frac{1}{\binom{N+K-1}{N}}\sum_{\substack{\textbf{n}\in\N^K\\ \|\textbf{n}\|_1=N}}\,\sum_{\substack{\textbf{i}\in\N^K\\ \|\textbf{i}\|_1=N}}\,\sum_{\substack{\textbf{l}\in\N^K\\ \textbf{l}\ge \max(\textbf{i}-\textbf{n},\textbf{0})}}\left(\prod_{j=1}^K f_{n_j,i_j,l_j}(g,\lambda)\right)\ketbraa{\textbf{n}}{\textbf{i}}_{A_1\ldots A_K} \otimes \ketbraa{\textbf{l}+\textbf{n}-\textbf{i}}{\textbf{l}}_{B_1\ldots B_K}\,,
    \ee
    where $\textbf{0}\in\N^K$ is the zero vector and the inequality between vectors $\textbf{a}\ge \textbf{b}$ means that $a_j \ge b_j$ for all $j=1,\ldots,K$.

    \item Bob performs the local POVM $\{\Pi^{(K)}_F\}_{F\in\N}$, where $\Pi^{(K)}_F$ is the projector onto the subspace whose total photon number equals $F$ (see \eqref{PROJECTOR_N_photon}), on the $K$ modes he has received.
    The probability of getting the outcome $F$ is denoted by $\mathcal{P}_F$ and it can be calculated as
    \bb
            \mathcal{P}_F&\coloneqq \Tr\left[\left(\mathbb{1}_{A_1\ldots A_k}\otimes\Pi^{(K)}_F\right)\, \left(\Id_{A_1\ldots A_k}\otimes\mathcal{N}_{g,\lambda}^{\otimes K}(\ketbra{\Psi_{N,K}})\right) \right]= \frac{1}{\binom{N+K-1}{N}}\sum_{\substack{\textbf{n}\in\N^K\\ \|\textbf{n}\|_1=N}}\Tr\left[\Pi^{(K)}_F\,\mathcal{N}_{g,\lambda}^{\otimes K}(\ketbra{\textbf{n}})\right]\\&=\sum_{P=0}^{\min(F,N)}\binom{N}{P}\binom{K+F-1}{F-P}\lambda^{P}(1-\lambda)^{N-P}\frac{(g-1)^{F-P}}{g^{K+F}} \,,
    \ee
    where we have exploited Lemma~\ref{Lemma_P_F}. The post-measurement state $\rho_{A_1\ldots A_kB_1\ldots B_k}^{(F)}$ conditioned on the outcome $F\in\N$ is given by
\bb\label{rho_step2_multirail}
        &\rho^{(F)}_{A_1\ldots A_kB_1\ldots B_k} \\
        &\quad = \frac{1}{\mathcal{P}_F}\left(\mathbb{1}_{A_1\ldots A_k}\otimes\Pi^{(K)}_F\right)\left( \Id_{A_1\ldots A_k}\otimes\mathcal{N}_{g,\lambda}^{\otimes K}(\ketbra{\Psi_{N,K}})\right)\left( \mathbb{1}_{A_1\ldots A_k}\otimes\Pi^{(K)}_F\right)\\
        &\quad =\frac{1}{\mathcal{P}_F\binom{N+K-1}{N}}\sum_{\substack{\textbf{n}\in\N^K\\ \|\textbf{n}\|_1=N}}\,\sum_{\substack{\textbf{i}\in\N^K\\ \|\textbf{i}\|_1=N}}\,\sum_{\substack{\textbf{l}\in\N^K\\ \|\textbf{l}\|_1=F\\\textbf{l}\ge \max(\textbf{i}-\textbf{n},\textbf{0})}}\left(\prod_{j=1}^K f_{n_j,i_j,l_j}(g,\lambda)\right)\ketbraa{\textbf{n}}{\textbf{i}}_{A_1\ldots A_K} \otimes \ketbraa{\textbf{l}+\textbf{n}-\textbf{i}}{\textbf{l}}_{B_1\ldots B_K}\,\\
        &\quad =\sum_{n,i=0}^{\binom{N+K-1}{N}-1} \sum_{h,l=0}^{\binom{F+K-1}{F}-1}c_{n,i,h,l}\ketbraa{\phi^{(N)}_n}{\phi^{(N)}_i}_{A_1\ldots A_k}\otimes\ketbraa{\phi^{(F)}_h}{\phi^{(F)}_l}_{B_1\ldots B_k}\,,
\ee
where for all $n,i=0,1,\ldots, \binom{N+K-1}{N}-1$ and all $h,l=0,1,\ldots, \binom{F+K-1}{F}-1$ the coefficient $c_{n,i,h,l}$ is defined as follows. Let $\textbf{n},\textbf{i}, \textbf{h}, \textbf{l}\in\N^K$ such that $\ket{\phi^{(N)}_n}=\ket{\textbf{n}}$, $\ket{\phi^{(N)}_i}=\ket{\textbf{i}}$, $\ket{\phi^{(F)}_h}=\ket{\textbf{h}}$, and $\ket{\phi^{(F)}_l}=\ket{\textbf{l}}$. If $\textbf{l}\ge \max(\textbf{i}-\textbf{n},\textbf{0})$ and $\textbf{h}=\textbf{l}+\textbf{n}-\textbf{i}$, then 
\bb
    c_{n,i,h,l}\coloneqq\frac{\left(\prod_{j=1}^K f_{n_j,i_j,l_j}(g,\lambda)\right)}{\mathcal{P}_F\binom{N+K-1}{N}}\,,
\ee
    otherwise $c_{n,i,h,l}=0$. By setting 
    \bb
    d\coloneqq \max\left(\binom{N+K-1}{N}, \binom{F+K-1}{F}\right)\,,
    \ee
    the resulting state in~\eqref{rho_step2_multirail} can be seen as a bipartite two-qu$d$it state $\rho^{(F)}_{AB}\in\mathfrak{S}(\HH_d\otimes\HH_d)$ of the form
    \bb
        \rho^{(F)}_{AB}=\sum_{n,i,h,l=0}^{d-1} \eta_{n,i,h,l}\ketbraa{n}{i}_A\otimes\ketbraa{h}{l}_B\,,
    \ee
    where $\HH_d$ is the qu$d$it Hilbert space with $\{\ket{0},\ket{1},\ldots, \ket{d-1}\}$ as an orthonormal basis, and where the coefficients $\eta_{n,i,h,l}$ are defined as follows:
    \begin{itemize}
        \item if $n,i\le \binom{N+K-1}{N}-1$ and $h,l\le \binom{F+K-1}{F}-1$, then $\eta_{n,i,h,l}\coloneqq c_{n,i,h,l}$;
        \item otherwise, $\eta_{n,i,h,l}\coloneqq 0$.
    \end{itemize}
    Consequently, Alice and Bob have reduced the problem in distilling ebits from the two-qu$d$it state $\rho^{(F)}_{AB}$.

    \item Now Alice and Bob decide whether or not to run the reverse hashing protocol, which can distil ebits from $\rho^{(F)}_{AB}$ with a rate equal to its reverse coherent information, i.e. 
    \bb
        I_{\text{rc}}(\rho^{(F)}_{AB})=S(\Tr_B\rho^{(F)}_{AB})-S(\rho^{(F)}_{AB})\,,
    \ee
    where $S(\cdot)$ denotes the von Neumann entropy. By exploiting that
    \bb
        \Tr_B\rho^{(F)}_{AB}=\frac{1}{\binom{N+K-1}{N}}\sum_{n=0}^{\binom{N+K-1}{N}-1}\ketbra{n},
    \ee
    as guaranteed by \eqref{rho_step2_multirail} and Lemma~\ref{Lemma_P_F}, it follows that the reverse coherent information can be calculated as
    \bb
        I_{\text{rc}}(\rho^{(F)}_{AB})=\log_2\binom{N+K-1}{N}-S\left( \sum_{n,i,h,l=0}^{d-1} \eta_{n,i,h,l}\ketbraa{n}{i}\otimes\ketbraa{h}{l}\right)\,.
    \ee
    If Alice and Bob choose to run the reverse hashing protocol, the protocol terminates. Otherwise, they apply the qu$d$it Pauli-based twirling reported in~\cite[Eq.~(18)]{p1orp2} in order to transform their state in a Bell-diagonal state of the form 
    \bb
    \rho'^{(F)}_{AB}=\sum_{m,n=0}^d \alpha^{(F,0)}_{mn}\ketbra{\psi^{(d)}_{mn}}_{AB}\,,
    \ee
    where 
    \bb
        \ket{\psi^{(d)}_{mn}}_{AB}\coloneqq\frac{1}{\sqrt{d}}\sum_{r=0}^{d-1}e^{i\frac{2\pi m r}{d}}\ket{r}_A\otimes\ket{(r-n)\text{ mod } d}_B
    \ee
    and 
    \bb
        \alpha^{(F,0)}_{mn}\coloneqq \bra{\psi^{(d)}_{mn} }\rho^{(F)}_{AB}\ket{\psi^{(d)}_{mn}}=\frac{1}{d}\sum_{r_1,r_2=0}^{d-1} \cos\left(\frac{2\pi m (r_2-r_1)}{d}\right)\eta_{r_1,r_2,\,(r_1-n)\text{ mod } d,\, (r_2-n)\text{ mod } d}\,.
    \ee

    \item Alice and Bob run $\bar{k}$ times the P1-or-P2 sub-routine for qu$d$its~\cite{p1orp2}, where $\bar{k}$ is chosen in order to maximise the ebit rate. The goal of this step is to bring the shared state closer to the $d$-dimensional maximally-entangled state $\ket{\psi_{00}^{(d)}}$. This step is successful, i.e.~the protocol is not aborted, with a probability of success equal to $\prod_{t=0}^{\bar{k}-1}P^{(F)}_t$ and it allows Alice and Bob to transform $2^{\bar{k}}$ copies of $\rho'^{(F)}_{AB}=\sum_{m,n=0}^d \alpha^{(F,0)}_{mn}\ketbra{\psi^{(d)}_{mn}}_{AB}$ in a state of the form
    \bb
         \rho'^{(F,\bar{k})}_{AB}\coloneqq\sum_{m,n=0}^d \alpha^{(F,\bar{k})}_{mn}\ketbra{\psi^{(d)}_{mn}}_{AB}\,.
    \ee
    For all $t\in\{0,1,\ldots,\bar{k}-1\}$ and all $m,n\in\{0,1,\ldots,d-1\}$ the coefficients $\alpha_{mn}^{(F,t+1)}$ and the probabilities $P^{(F)}_t$ are recursively defined in the following way~\cite{p1orp2}:
\begin{itemize}
    \item If $\sum_{m_1=0}^{d-1}\alpha^{(F,t)}_{m_10}<\sum_{n_1=0}^{d-1}\alpha^{(F,t)}_{0n_1}$, then
    \bb\label{coeff_p1_multirail}
            \alpha^{(F,t+1)}_{mn}\coloneqq \frac{1}{P^{(F)}_t}\sum_{\substack{m_1,m_2=0\\(m_1+ m_2)\text{ mod }d = m}}^{d-1}\alpha_{m_1n}^{(F,t)}\alpha_{m_2n}^{(F,t)}\,,
    \ee
    where
    \bb\label{probk_p1_multirail}
        P^{(F)}_t \coloneqq \sum_{m_1,m_2,n=0}^{d-1}\alpha_{m_1n}^{(F,t)}\alpha_{m_2n}^{(F,t)} \,.
    \ee
    \item Otherwise,
    \bb\label{coeff_p2_multirail}
        \alpha^{(F,t+1)}_{mn}\coloneqq \frac{1}{P^{(F)}_t}\sum_{\substack{n_1,n_2=0\\(n_1+ n_2)\text{ mod }d = n}}^{d-1}\alpha_{mn_1}^{(F,t)}\alpha_{mn_2}^{(F,t)}\,,
    \ee    
    where
    \bb\label{probk_p2_multirail}
        P^{(F)}_t \coloneqq \sum_{m,n_1,n_2=0}^{d-1}\alpha_{mn_1}^{(F,t)}\alpha_{mn_2}^{(F,t)} \,.
    \ee
\end{itemize}

     \item Alice and Bob distil ebits from the state $\rho'^{(F,\bar{k})}_{AB}=\sum_{m,n=0}^d \alpha^{(F,\bar{k})}_{mn}\ketbra{\psi^{(d)}_{mn}}_{AB}$ with a yield denoted as $\mathcal{I}_d(\alpha^{(F,\bar{k})})$ by running the following protocol:
    \begin{itemize}
    
        \item If $d=2$, then Alice and Bob run the Step~5 and Step~6 of the entanglement distribution protocol introduced in the proof of Theorem~\ref{th_lower_Q2} in order to distil ebits from $\rho'^{(F,\bar{k})}_{AB}$ with a yield equal to 
        \bb
            \mathcal{I}_2(\alpha^{(F,\bar{k})})\coloneqq\mathcal{I}(  \alpha^{(F,\bar{k})}_{00},\alpha^{(F,\bar{k})}_{01},\alpha^{(F,\bar{k})}_{10},\alpha^{(F,\bar{k})}_{11})\,,
        \ee
    where $\mathcal{I}$ is defined in~\eqref{yield_final_protocol}. 

        \item If $d>2$, then Alice and Bob run the hashing protocol on $\rho'^{(F,\bar{k})}_{AB}$ and thus they distil ebits with a yield equal to the coherent information of $\rho'^{(F,\bar{k})}_{AB}$, i.e.
        \bb
            \mathcal{I}_d(\alpha^{(F,\bar{k})})\coloneqq I_{\text{c}}(\rho'^{(F,\bar{k})}_{AB})= \log_2 d+ \sum_{m,n=0}^{d-1}\alpha^{(F,\bar{k})}_{mn}\log_2\alpha^{(F,\bar{k})}_{mn}\,.
        \ee
    \end{itemize}

\end{enumerate}

    The ebit rate of the protocol is given by
    \bb\label{rate_multirail}
        R(g,\lambda,N,K)\coloneqq \frac{1}{K}\sum_{F=1}^\infty \mathcal{P}_F \max\left(I_{\text{rc}}(\rho^{(F)}_{AB})\,,\,\sup_{\bar{k}\in\N} \frac{\prod_{t=0}^{\bar{k}-1}P^{(F)}_t}{2^{\bar{k}}}\mathcal{I}_d(  \alpha^{(F,\bar{k})})\right)\,.
    \ee
    The term $\frac{1}{K}$ in the expression~\eqref{rate_multirail} arises from the fact that Alice uses the channel $K$ times during step S1, and the variable $F$ corresponds to the outcome of the total photon number measurement in step S2, with associated probability $\mathcal{P}_F$. The sum over $F$ equals the expected value of the yield of ebits that can be distilled from the post-measurement state $\rho^{(F)}_{AB}$ by running steps S3, S4, and S5. The maximum comes from the fact that during step S3 Alice and Bob choose whether or not to run the reverse hashing protocol, which can distil ebits with a rate equal to $I_{\text{rc}}(\rho^{(F)}_{AB})$. The supremum over $\bar{k}$ comes from the fact that Alice and Bob choose the number of iterations $\bar{k}$ of the P1-or-P2 subroutine in order to maximise the rate. The rate in~\eqref{rate_multirail} is a lower bound on the two-way quantum capacity of the piBGC $\mathcal{N}_{g,\lambda}$ for all $N,K\in\N^+$ with $K\ge 2$. Therefore, we have
    \bb
        K(\mathcal{N}_{g,\lambda}) \ge  Q_2(\mathcal{N}_{g,\lambda})\ge \sup_{\substack{N,K\in\N^+\\ K\ge2 }} R(g,\lambda,N,K)\,.
    \ee
    Let us summarise this result in the following theorem.
    \begin{thm}\label{thm_multirail}
        For all $\lambda\in[0,1]$ and $g\ge 1$ the secret-key capacity $K(\mathcal{N}_{g,\lambda})$ and the two-way quantum capacity $Q_2(\mathcal{N}_{g,\lambda})$ of the piBGC $\mathcal{N}_{g,\lambda}$ satisfy
        \bb
            K(\mathcal{N}_{g,\lambda}) \ge  Q_2(\mathcal{N}_{g,\lambda})\ge \sup_{\substack{N,K\in\N^+\\ K\ge2 }} R(g,\lambda,N,K)\,,
        \ee
        where
        \bb\label{rate_expr_multirail}
        R(g,\lambda,N,K)\coloneqq \frac{1}{K}\sum_{F=1}^\infty \mathcal{P}_F \max\left(I_{\text{rc}}^{(F)}\,,\,\sup_{\bar{k}\in\N} \frac{\prod_{t=0}^{\bar{k}-1}P^{(F)}_t}{2^{\bar{k}}}\mathcal{I}_d(  \alpha^{(F,\bar{k})})\right)\,.
        \ee
        The quantities present in~\eqref{rate_expr_multirail} are defined as follows. For all $F\in\N$ the dimension $d$ is defined as 
        \bb
        d\coloneqq \max\left(\binom{N+K-1}{N}, \binom{F+K-1}{F}\right)
        \ee
        and the probability $\mathcal{P}_F$ is defined as
        \bb
            \mathcal{P}_F&\coloneqq \sum_{P=0}^{\min(F,N)}\binom{N}{P}\binom{K+F-1}{F-P}\lambda^{P}(1-\lambda)^{N-P}\frac{(g-1)^{F-P}}{g^{K+F}}\,.
        \ee
        Moreover, the probabilities $P^{(F)}_{\bar{k}}$ and the coefficients $\{\alpha_{mn}^{(F,\bar{k})}\}_{m,n\in\{0,1,\ldots,d-1\}}$ are recursively defined as follows.
        For all $t\in\{0,1,\ldots,\bar{k}-1\}$ and all $m,n\in\{0,1,\ldots,d-1\}$ it holds that:
\begin{itemize}
    \item If $\sum_{m_1=0}^{d-1}\alpha^{(F,t)}_{m_10}<\sum_{n_1=0}^{d-1}\alpha^{(F,t)}_{0n_1}$, then
    \bb\label{coeff_p1_multirail2}
            \alpha^{(F,t+1)}_{mn}&\coloneqq \frac{1}{P^{(F)}_t}\sum_{\substack{m_1,m_2=0\\(m_1+ m_2)\text{ mod }d = m}}^{d-1}\alpha_{m_1n}^{(F,t)}\alpha_{m_2n}^{(F,t)}\,,\\
        P^{(F)}_t &\coloneqq \sum_{m_1,m_2,n=0}^{d-1}\alpha_{m_1n}^{(F,t)}\alpha_{m_2n}^{(F,t)} \,.
    \ee
    \item Otherwise,
    \bb\label{coeff_p2_multirail2}
        \alpha^{(F,t+1)}_{mn}&\coloneqq \frac{1}{P^{(F)}_t}\sum_{\substack{n_1,n_2=0\\(n_1+ n_2)\text{ mod }d = n}}^{d-1}\alpha_{mn_1}^{(F,t)}\alpha_{mn_2}^{(F,t)}\,,\\
        P^{(F)}_t &\coloneqq \sum_{m,n_1,n_2=0}^{d-1}\alpha_{mn_1}^{(F,t)}\alpha_{mn_2}^{(F,t)} \,.
    \ee
\end{itemize}
Moreover, for all $m,n\in\{0,1,\ldots,d-1\}$ the coefficient $\alpha^{(F,0)}_{mn}$ is defined as
    \bb
        \alpha^{(F,0)}_{mn}\coloneqq \frac{1}{d}\sum_{r_1,r_2=0}^{d-1} \cos\left(\frac{2\pi m (r_2-r_1)}{d}\right)\eta_{r_1,r_2,(r_1-n)\text{ mod } d, (r_2-n)\text{ mod } d}\,.
    \ee
    In addition, for all $n,i\in{0,1,\ldots,\binom{N+K-1}{N}-1}$, we define $(n_1,\ldots,n_K)$ and $(i_1,\ldots,i_K)$ as the $n$th and $i$th element of the ordered set $S_{K,N}$, where $S_{K,N}$ is defined as  \bb
    S_{K,N}\coloneqq\{(f_1,\ldots,f_K)\in\N^K:\,  \sum_{j=1}^{K} f_j=N\}
    \ee
    and it is ordered according to the relation $\preceq_{K,N}$, given by     \bb
        (f_1,\ldots,f_K)\,\preceq_{K,N}\,(g_1,\ldots,g_K) \, \Longleftrightarrow \sum_{j=1}^{K} f_{j}\, (N+1)^j<\sum_{j=1}^{K} g_{j}\, (N+1)^j\,.
    \ee
    Additionally, for all $h,l\in\{0,1,\ldots,\binom{F+K-1}{F}-1\}$, we define $(h_1,\ldots,h_K)$ and $(l_1,\ldots,l_K)$ as the $h$th and $l$th element of the set $S_{K,F}$ ordered according to the relation $\preceq_{F,N}$. Furthermore, for all $n,i,h,l\in\{0,1,\ldots,d-1\}$ the coefficients $\eta_{n,i,h,l}$ are defined as follows:
    \begin{itemize}
        \item If 
        \bb
        n,i&\le \binom{N+K-1}{N}-1\,,\\
        h,l&\le \binom{F+K-1}{F}-1\,,\\ 
        l_j&\ge \max(i_j-n_j,0)\,\quad \text{for all } j=1,2,\ldots,K\,,\\  
        h_j&=l_j+n_j-i_j\,\quad \text{for all } j=1,2,\ldots,K\,,  
    \ee
    then 
    \bb
    \eta_{n,i,h,l}\coloneqq\frac{\left(\prod_{j=1}^K f_{n_j,i_j,l_j}(g,\lambda)\right)}{\mathcal{P}_F\binom{N+K-1}{N}}\,,
    \ee
    where $f_{n,i,l}(g,\lambda)$ is defined in~\eqref{def_f_comp}.
    \item Otherwise, $\eta_{n,i,h,l}=0$. 
    \end{itemize}
    Moreover, the quantity $I_{\text{rc}}^{(F)}$ is defined as
    \bb
        I_{\text{rc}}^{(F)}\coloneqq \log_2\binom{N+K-1}{N}-S\left( \sum_{n,i,h,l=0}^{d-1} \eta_{n,i,h,l}\ketbraa{n}{i}\otimes\ketbraa{h}{l}\right)\,,
    \ee
    where $S(\cdot)$ denotes the von Neumann entropy.
    Finally, the term $\mathcal{I}_d(\alpha^{(F,\bar{k})})$ is defined differently depending on the value of $d$:
    \begin{itemize}
        \item If $d=2$, then
        \bb
            \mathcal{I}_2(\alpha^{(F,\bar{k})})\coloneqq\mathcal{I}(  \alpha^{(F,\bar{k})}_{00},\alpha^{(F,\bar{k})}_{01},\alpha^{(F,\bar{k})}_{10},\alpha^{(F,\bar{k})}_{11})\,,
        \ee
    where $\mathcal{I}$ is defined in~\eqref{yield_final_protocol}. 
        \item If $d>2$, then 
        \bb
            \mathcal{I}_d(\alpha^{(F,\bar{k})})\coloneqq \log_2 d+ \sum_{m,n=0}^{d-1}\alpha^{(F,\bar{k})}_{mn}\log_2\alpha^{(F,\bar{k})}_{mn}\,.
        \ee
    \end{itemize}
    \end{thm}

\section{Results on the two-way capacities of piBGCs}\label{sub_res_twoway}
In this subsection, for each of the piBGCs, first we determine the parameter region where the two-way capacities vanish, second we find a new lower bound on the two-way capacities, and finally we compare our results with the existing literature.
\subsection{Results on the two-way capacities of the thermal attenuator}
Let us consider the thermal attenuator $\mathcal{E}_{\lambda,\nu}$ of transmissivity $\lambda\in[0,1]$ and thermal noise $\nu\ge0$. Since the PLOB bound in~\eqref{PLOB_Q2} vanishes for $\lambda\le \frac{\nu}{\nu+1}$, it is already known that the two-way capacities of $\mathcal{E}_{\lambda,\nu}$ vanish for $\lambda<\frac{\nu}{\nu+1}$. The following theorem establishes that also the vice-versa is true.
\begin{thm}\label{th1_therm}
Let $\lambda\in[0,1]$, $\nu\ge0$, and $N_s>0$. The energy-constrained two-way capacities of the thermal attenuator $Q_2(\mathcal{E}_{\lambda,\nu},N_s)$ and $K(\mathcal{E}_{\lambda,\nu},N_s)$ vanish if and only if $\lambda\le \frac{\nu}{\nu+1}$, i.e.~if and only if $\mathcal{E}_{\lambda,\nu}$ is entanglement breaking. In particular, the same holds for the unconstrained two-way capacities.
\end{thm}
\begin{proof}
    Theorem~\ref{th1_therm} is a direct consequence of Lemma~\ref{lemma_comp_bos} and Theorem~\ref{th1}.
\end{proof}
The validity of Theorem~\ref{th1_therm} was not known before the present work. Indeed, in~\cite{Pirandola18,Pirandola20} the authors says that it is an open problem to determine the exact value of the maximum tolerable excess noise, which is defined by 
\begin{equation}\label{excess_noise}
   \epsilon(\lambda)\coloneqq \frac{1-\lambda}{\lambda}\max\{\nu\ge0\,:\,K(\mathcal{E}_{\lambda,\nu})>0\}\,.
\end{equation}
Theorem~\ref{th1} implies that $\varepsilon(\lambda)=1$ for all $\lambda\in(0,1)$. Hence, we have answered to the question, which was deemed ``crucial" in~\cite[Section 7]{Pirandola18}, ``What is the maximum
excess noise that is tolerable in QKD? I.e., optimizing over all QKD protocols?"
In~\cite{Pirandola18,Pirandola20} the authors showed, by applying the PLOB bound, the upper bound $\varepsilon(\lambda)\le1$ and provided also a lower bound on $\varepsilon(\lambda)$ which was far from $1$.

Except for the special case $\nu=0$, it is an open question whether the reverse coherent information lower bound in Eq.~\ref{lowQ2} equals the true two-way quantum capacity of the thermal attenuator $Q_2(\mathcal{E}_{\lambda,\nu})$: Theorem~\ref{th1_therm} provides a negative answer to this question. Indeed, although $Q_2(\mathcal{E}_{\lambda,\nu})=0$ if and only if $\lambda\le \frac{\nu}{\nu+1}$ (thanks to Theorem~\ref{th1_therm}), the reverse coherent information lower bound vanishes for all $\lambda \le 1-2^{-h(\nu)}$. Hence, since $1-2^{-h(\nu)}>\frac{\nu}{\nu+1}$ for all $\nu>0$, the reverse coherent information lower bound is not equal to $Q_2(\mathcal{E}_{\lambda,\nu})$ at least in the region $\nu>0$ and $\lambda\in (\frac{\nu}{\nu+1},1-2^{-h(\nu)}]$. In the following theorem we obtain an improved lower bound on the two-way capacities of the thermal attenuator. 
\begin{thm}\label{th_delta}
 Let $\lambda\in[0,1]$, $\nu\ge0$, and $N_s\ge0$. The EC two-way capacities $Q_2(\mathcal{E}_{\lambda,\nu},N_s)$ and $K(\mathcal{E}_{\lambda,\nu},N_s)$ of the thermal attenuator $\mathcal{E}_{\lambda,\nu}$ satisfy the following lower bound
\bb\label{lowQ2_deltaEC}
	K(\mathcal{E}_{\lambda,\nu},N_s)& \ge  Q_2(\mathcal{E}_{\lambda,\nu},N_s)\ge \sup_{\substack{c\in(0,1),\, M\in\N^+,\, k\in\N\\ (1-c^2)M\le N_s}} \mathcal{R}\left(1+(1-\lambda)\nu,\frac{\lambda}{1+(1-\lambda)\nu},M,c,k\right)\,,
\ee
and, in particular, the unconstrained two-way capacities satisfy
\bb\label{lowQ2_delta}
	K(\mathcal{E}_{\lambda,\nu})& \ge  Q_2(\mathcal{E}_{\lambda,\nu})\ge \sup_{c\in(0,1),\, M\in\N^+,\, k\in\N} \mathcal{R}\left(1+(1-\lambda)\nu,\frac{\lambda}{1+(1-\lambda)\nu},M,c,k\right)\,,
\ee
where the quantity $\mathcal{R}$ is defined in~\eqref{def_mathcal_R}. 
\end{thm}
\begin{proof} 
     Theorem~\ref{th_delta} is a direct consequence of Lemma~\ref{lemma_comp_bos} and Theorem~\ref{th_comp_Q2}.
\end{proof}
Theorem~\ref{th_delta} shows a new lower bound, reported in~\eqref{lowQ2_delta}, on the two-way capacities of the thermal attenuator $\mathcal{E}_{\lambda,\nu}$. Our new lower bound outperforms all the previous known lower bounds in a large region of the parameters $\lambda$ and $\nu$.
In Fig.~\ref{figure_nu2_delta}a and in Fig.~\ref{figure_nu2_delta}b we plot our new bound and its ratio with the PLOB bound, respectively, with respect to $\nu$ where the transmissivity is chosen to be equal to $\lambda(\nu)\coloneqq 1-2^{-h(\nu)}$, which is the upper endpoint for the $\lambda$-range for which the best known lower bound on $Q_2(\mathcal{E}_{\lambda,\nu})$ (i.e.~the reverse coherent information lower bound reported in~\eqref{lowQ2}) vanishes. From Fig.~\ref{figure_nu2_delta}a and Fig.~\ref{figure_nu2_delta}b we see that for these choices of $\nu$ and $\lambda(\nu)$, our new lower bound is now the best lower bound on $Q_2(\mathcal{E}_{\lambda,\nu})$ and it achieves the $\simeq 14\%$ of the PLOB bound for $\nu\gg 1$. For example, if $\nu=1$ and if the transmissivity is equal to $\lambda=1-2^{-h(1)}=0.75$, our new lower bound is $\simeq 0.033$, its ratio with the PLOB bound is $\simeq 0.08$, and the optimal parameters of the supremum present in the expression of our new bound in~\eqref{lowQ2_delta} are $c\simeq0.703$, $M=2$, and $k=2$. In Fig.~\ref{rate_vs_lambda1} we plot our new bound on $Q_2(\mathcal{E}_{\lambda,\nu})$ with respect to $\lambda$ for $\nu=1$ and $\nu=10$.

Our new bound can outperform also the best known lower bound (before our work) on the secret-key capacity $K(\mathcal{E}_{\lambda,\nu})$ found by Ottaviani et al.~\cite{Ottaviani_new_lower}. To demonstrate that our new bound can be strictly tighter than the Ottaviani et al.~lower bound, in Fig.~\ref{secret_key_nu} we plot the latter bound and our new bound with respect to $\lambda$ for $\nu=1$ and $\nu=10$. From Fig.~\ref{secret_key_nu}, we note that the Ottaviani et al.~lower bound vanishes for larger transmissivities than our bound. In particular, fixed $\nu>0$, we numerically observe that our new bound is strictly positive for all $\lambda>\frac{\nu}{\nu+1}$, which is the region where the two-way capacities of $\mathcal{E}_{\lambda,\nu}$ are strictly positive, as established by Theorem~\ref{th1}. As an example, for $\nu=1$, in Fig.~\ref{log1} we plot the ratio between our bound and the PLOB bound in logarithmic scale and we see that our bound is strictly positive for $\lambda\gtrsim \frac{\nu}{\nu+1}=0.5$. In addition, fixed $\nu>0$, we numerically observe that the optimal value of $k$ of the supremum present in the expression of our new bound in~\eqref{lowQ2_delta} increases as $\lambda$ decreases and tends to infinity as $\lambda$ tends to $\frac{\nu}{\nu+1}$, where we recall that $k$ represents the number of iterations of the P1-or-P2 sub-routine~\cite{p1orp2} in the entanglement distribution protocol we have introduced in the proof of Theorem~\ref{th_lower_Q2}.

We numerically observe that for all $\lambda$ and $\nu$ the optimal choice of $M$ of the supremum present in the expression of our bound in~\eqref{lowQ2_delta} is always less or equal to $3$. Hence, since the mean photon number of each signal sent by Alice is $\Tr[a^\dagger a\ketbra{\Psi_{M,c}}]=(1-c^2)M$ (see~\eqref{initial_state}), the entanglement distribution protocol we have presented in the proof of Theorem~\ref{th_delta} exploits a mean photon number per channel use which is strictly lower than $3$. On the contrary, the entanglement distribution protocol which leads to the reverse coherent information lower bound in~\eqref{lowQ2} requires infinite mean photon number per channel use, as we reviewed in~\ref{proof_lower}.

Theorem~\ref{th_delta} shows also the bound in~\eqref{lowQ2_deltaEC}, which constitutes a new lower bound on the EC two-way capacities of the thermal attenuator $\mathcal{E}_{\lambda,\nu}$. 
This new lower bound can outperform the NPJ lower bound~\cite{Noh2020} reported in~\eqref{npj_bound_therm}, which is the best known lower bound on the EC two-way capacities of the thermal attenuator, as we show in Fig.~\ref{ECfigures} where we plot our new bound in~\eqref{lowQ2_deltaEC} with respect to $\lambda$ for different choices of $\nu$ and of the energy constraint $N_s$.


By using the results of Section~\ref{multiplerail_section}, in the forthcoming Theorem~\ref{theorem_new_lower_multirail} we show an additional lower bound on the two-way quantum capacity of the thermal attenuator $\mathcal{E}_{\lambda,\nu}$.
\begin{thm}[Multi-rail lower bound]\label{theorem_new_lower_multirail}
 For all $\lambda\in[0,1]$ and $\nu\ge0$ the two-way capacities of the thermal attenuator $\mathcal{E}_{\lambda,\nu}$ satisfy
\bb\label{multiplerail_low_bound}
    K(\mathcal{E}_{\lambda,\nu}) &\ge  Q_2(\mathcal{E}_{\lambda,\nu})\ge \sup_{\substack{N,K\in\N^+\\ K\ge2 }} R\left(1+(1-\lambda)\nu,\frac{\lambda}{1+(1-\lambda)\nu},N,K\right)\,,
\ee
where the quantity $R$ is defined in~\eqref{rate_expr_multirail}. 
\end{thm}
\begin{proof}
    Theorem~\ref{theorem_new_lower_multirail} is a direct consequence of Theorem~\ref{thm_multirail} and Lemma~\ref{lemma_comp_bos}.
\end{proof}
Theorem~\ref{theorem_new_lower_multirail} shows an additional lower bound on $Q_2(\mathcal{E}_{\lambda,\nu})$, that we dub `multi-rail lower bound'. This bound is the ebit rate of the entanglement distribution protocol presented in Section~\ref{multiplerail_section}, which combines the multi-rail protocol introduced in~\cite{Winnel} and the qudit P1-or-P2 protocol introduced in~\cite{p1orp2}.
In Fig.~\ref{multiplerail_figures} we plot both the multi-rail lower bound (reported in~\eqref{multiplerail_low_bound}) and our previously discussed lower bound (reported in~\eqref{lowQ2_delta}) as a function of $\lambda$ for $\nu=0.1$, $\nu=0.5$, $\nu=1$, and $\nu=10$. Our numerical investigation shows that for $\nu\lesssim 1$, the multi-rail lower bound is tighter than the previously discussed lower bound, as confirmed by Fig.~\ref{multiplerail_figures}.

\subsection{Results on the two-way capacities of the thermal amplifier}
Let us consider the thermal amplifier $\Phi_{g,\nu}$ of gain $g\ge1$ and thermal noise $\nu\ge0$. Since the PLOB bound in~\eqref{PLOB_amp} vanishes for $g\ge 1+\frac{1}{\nu}$, it is already known that the two-way capacities of $\Phi_{g,\nu}$ vanish for $g\ge 1+\frac{1}{\nu}$. The following theorem establishes that also the vice-versa is true.
\begin{thm}\label{th1_amp}
Let $g\ge1$, $\nu\ge0$, and $N_s>0$. The energy-constrained two-way capacities of the thermal amplifier $Q_2(\Phi_{g,\nu},N_s)$ and $K(\Phi_{g,\nu},N_s)$ vanish if and only if $g\ge 1+\frac{1}{\nu}$, i.e.~if and only if $\Phi_{g,\nu}$ is entanglement breaking. In particular, the same holds for the unconstrained two-way capacities.
\end{thm}
\begin{proof}
    Theorem~\ref{th1_amp} is a direct consequence of Lemma~\ref{lemma_comp_bos} and Theorem~\ref{th1}.
\end{proof}
Except for the special case $\nu=0$, it is an open question whether the coherent information lower bound in Eq.~\ref{lowQ2_amp} equals the true two-way quantum capacity of the thermal amplifier $Q_2(\Phi_{g,\nu})$: Theorem~\ref{th1_amp} provides a negative answer to this question. Indeed, although $Q_2(\Phi_{g,\nu})=0$ if and only if $g > 1+\frac{1}{\nu}$ (thanks to Theorem~\ref{th1_amp}), the coherent information lower bound vanishes for all $g \ge \frac{1}{1-2^{-h(\nu)}}$. Hence, since $1+\frac{1}{\nu}>\frac{1}{1-2^{-h(\nu)}}$ for all $\nu>0$, the coherent information lower bound is not equal to $Q_2(\Phi_{g,\nu})$ at least in the region $\nu>0$ and $g\in [\frac{1}{1-2^{-h(\nu)}},1+\frac{1}{\nu})$. In the following theorem we obtain an improved lower bound on the two-way capacities of the thermal amplifier.
\begin{thm}\label{th_delta_amp}
Let $g\ge 1$, $\nu\ge0$, and $N_s\ge0$. The EC two-way capacities $Q_2(\Phi_{g,\nu},N_s)$ and $K(\Phi_{g,\nu},N_s)$ of the thermal amplifier $\Phi_{g,\nu}$ satisfy the following lower bound
\bb\label{lowQ2_deltaEC_amp}
	K(\Phi_{g,\nu},N_s)&\ge Q_2(\Phi_{g,\nu},N_s) \ge \sup_{\substack{c\in(0,1),\, M\in\N^+,\, k\in\N\\ (1-c^2)M\le N_s}} \mathcal{R}\left(g+(g-1)\nu,\frac{g}{g+(g-1)\nu},M,c,k\right)\,,
\ee
and, in particular, the unconstrained two-way capacities satisfy
\bb\label{lowQ2_delta_amp}
	K(\Phi_{g,\nu})&\ge Q_2(\Phi_{g,\nu}) \ge \sup_{c\in(0,1),\, M\in\N^+,\, k\in\N} \mathcal{R}\left(g+(g-1)\nu,\frac{g}{g+(g-1)\nu},M,c,k\right)\,,
\ee
where the quantity $\mathcal{R}$ is defined in~\eqref{def_mathcal_R}. 
\end{thm}
\begin{proof} 
     Theorem~\ref{th_delta_amp} is a direct consequence of Lemma~\ref{lemma_comp_bos} and Theorem~\ref{th_comp_Q2}.
\end{proof}
Theorem~\ref{th_delta_amp} shows a new lower bound, reported in~\eqref{lowQ2_delta_amp}, on the two-way capacities of the thermal amplifier $\Phi_{g,\nu}$. Our new lower bound outperforms all the previous known lower bounds in a large region of the parameters $g$ and $\nu$.
In Fig.~\ref{capvsnu_amp}a and in Fig.~\ref{capvsnu_amp}b we plot our new bound and its ratio with the PLOB bound, respectively, with respect to $\nu$ where the transmissivity is chosen to be equal to $g(\nu)\coloneqq \frac{1}{1-2^{-h(\nu)}}$, which is the lower endpoint for the $g$-range for which the best known lower bound on $Q_2(\Phi_{g,\nu})$ (i.e.~the coherent information lower bound reported in~\eqref{lowQ2_amp}) vanishes. From Fig.~\ref{capvsnu_amp}a and Fig.~\ref{capvsnu_amp}b we see that for these choices of $\nu$ and $g(\nu)$, our new lower bound is now the best lower bound on $Q_2(\Phi_{g,\nu})$ and it achieves the $\simeq 14\%$ of the PLOB bound for $\nu\gg 1$. In Fig.~\ref{bound_vs_nu1_amp} we plot our new bound with respect to $g$ for $\nu=1$ and $\nu=10$.

Our new bound can outperform also the WOGP-bound~\cite{Wang_Q2_amplifier}, which is the best known lower bound (before our work) on the secret-key capacity $K(\Phi_{g,\nu})$. To demonstrate that our new bound can be strictly tighter than the WOGP lower bound, in Fig.~\ref{secret_g_nu10} we plot the latter bound and our new bound with respect to $g$ for $\nu=1$ and $\nu=10$. From Fig.~\ref{secret_g_nu10} we note that the WOGP lower bound vanishes for smaller values of $g$ than our bound. In particular, fixed $\nu>0$, we numerically observe that our new bound is strictly positive for all $g<1+\frac{1}{\nu}$, which is the region where the two-way capacities of $\Phi_{g,\nu}$ are strictly positive, as established by Theorem~\ref{th1_amp}. As an example, for $\nu=1$, in Fig.~\ref{log_ratio_vs_lam_amp} we plot the ratio between our bound and the PLOB bound in logarithmic scale and we see that our bound is strictly positive for $g\lesssim 1+\frac{1}{\nu}=2$.

\subsection{Results on the two-way capacities of the additive Gaussian noise}
Let us consider the additive Gaussian noise $\Lambda_\xi$ of parameter $\xi\ge0$. Since the PLOB bound in~\eqref{PLOB_add} vanishes for $\xi\ge1$, it is already known that the two-way capacities of $\Lambda_\xi$ vanish for $\xi\ge1$. The following theorem establishes that also the vice-versa is true.
\begin{thm}\label{th1_adgn}
    Let $\xi\ge 0$, and $N_s>0$. The energy-constrained two-way capacities of the additive Gaussian noise $Q_2({\Lambda}_{\xi},N_s)$ and $K({\Lambda}_{\xi},N_s)$ vanish if and only if $\xi\ge1 $. In particular, the two-way capacities $Q_2({\Lambda}_{\xi})$ and $K({\Lambda}_{\xi})$ vanish if and only if $\xi\ge 1$.
\end{thm}
\begin{proof}
    Theorem~\ref{th1_adgn} is a direct consequence of Lemma~\ref{lemma_comp_bos} and Theorem~\ref{th1}.
\end{proof}
 It is an open question whether the coherent information lower bound in Eq.~\ref{lowQ2_add} equals the true two-way quantum capacity of the additive Gaussian noise $Q_2(\Lambda_{\xi})$: Theorem~\ref{th1_adgn} provides a negative answer to this question. Indeed, although $Q_2(\Lambda_\xi)=0$ if and only if $\xi\ge1$ (thanks to Theorem~\ref{th1_adgn}), the coherent information lower bound vanishes for all $\xi\ge\frac{1}{e}$. Hence, the coherent information lower bound is not equal to $Q_2(\Lambda_\xi)$ at least in the region $\xi\in [\frac{1}{e},1)$. In the following theorem we obtain an improved lower bound on the two-way capacities of the additive Gaussian noise.
\begin{thm}\label{th_delta_add}
Let $\xi\in [0,1)$ and $N_s\ge0$. The EC two-way capacities $Q_2(\Lambda_\xi,N_s)$ and $K(\Lambda_\xi,N_s)$ of the additive Gaussian noise $\Lambda_\xi$ satisfy the following lower bound
\bb\label{lowQ2_deltaEC_add}
	K(\Lambda_\xi,N_s)&\ge Q_2(\Lambda_\xi,N_s) \ge \sup_{\substack{c\in(0,1),\, M\in\N^+,\, k\in\N\\ (1-c^2)M\le N_s}} \mathcal{R}\left(1+\xi,\frac{1}{1+\xi},M,c,k\right)\,,
\ee
and, in particular, the unconstrained two-way capacities satisfy
\bb\label{lowQ2_delta_add}
	K(\Lambda_\xi)&\ge Q_2(\Lambda_\xi) \ge \sup_{c\in(0,1),\, M\in\N^+,\, k\in\N} \mathcal{R}\left(1+\xi,\frac{1}{1+\xi},M,c,k\right)\,,
\ee 
where the quantity $\mathcal{R}$ is defined in~\eqref{def_mathcal_R}. 
\end{thm}
\begin{proof} 
     Theorem~\ref{th_delta_add} is a direct consequence of Lemma~\ref{lemma_comp_bos} and Theorem~\ref{th_comp_Q2}.
\end{proof}
Theorem~\ref{th_delta_add} shows a new lower bound, reported in~\eqref{lowQ2_delta_add}, on the two-way capacities of the additive Gaussian noise $\Lambda_\xi$. Our new lower bound outperforms all the previous known lower bounds in a large region of the parameter $\xi$, as it can been seen from Fig.~\ref{add_vs_xi}.

\begin{figure}
\begin{tabular}{c}
  \includegraphics[width=0.97\linewidth]{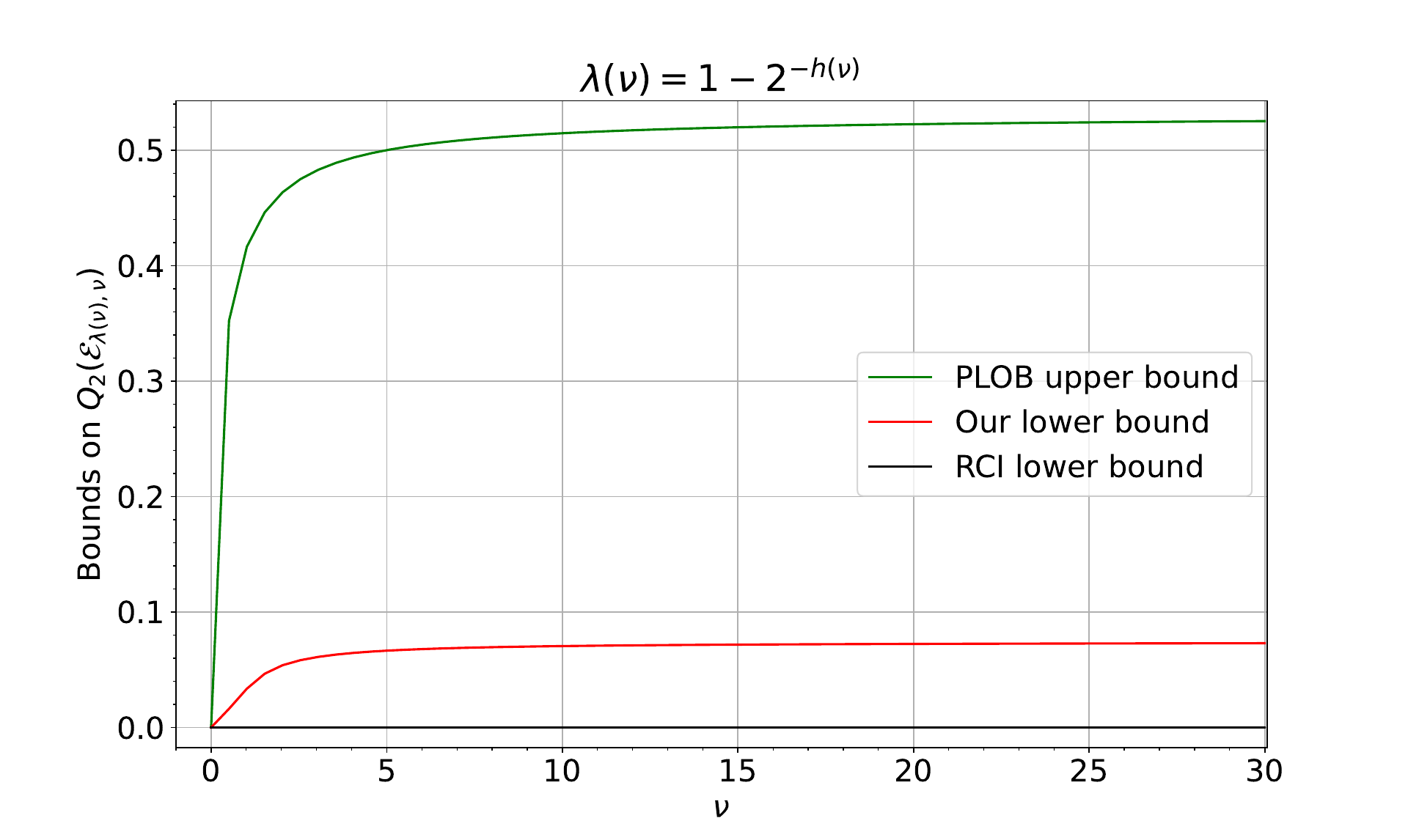} \\  
(a)\\
 \includegraphics[width=0.97\linewidth]{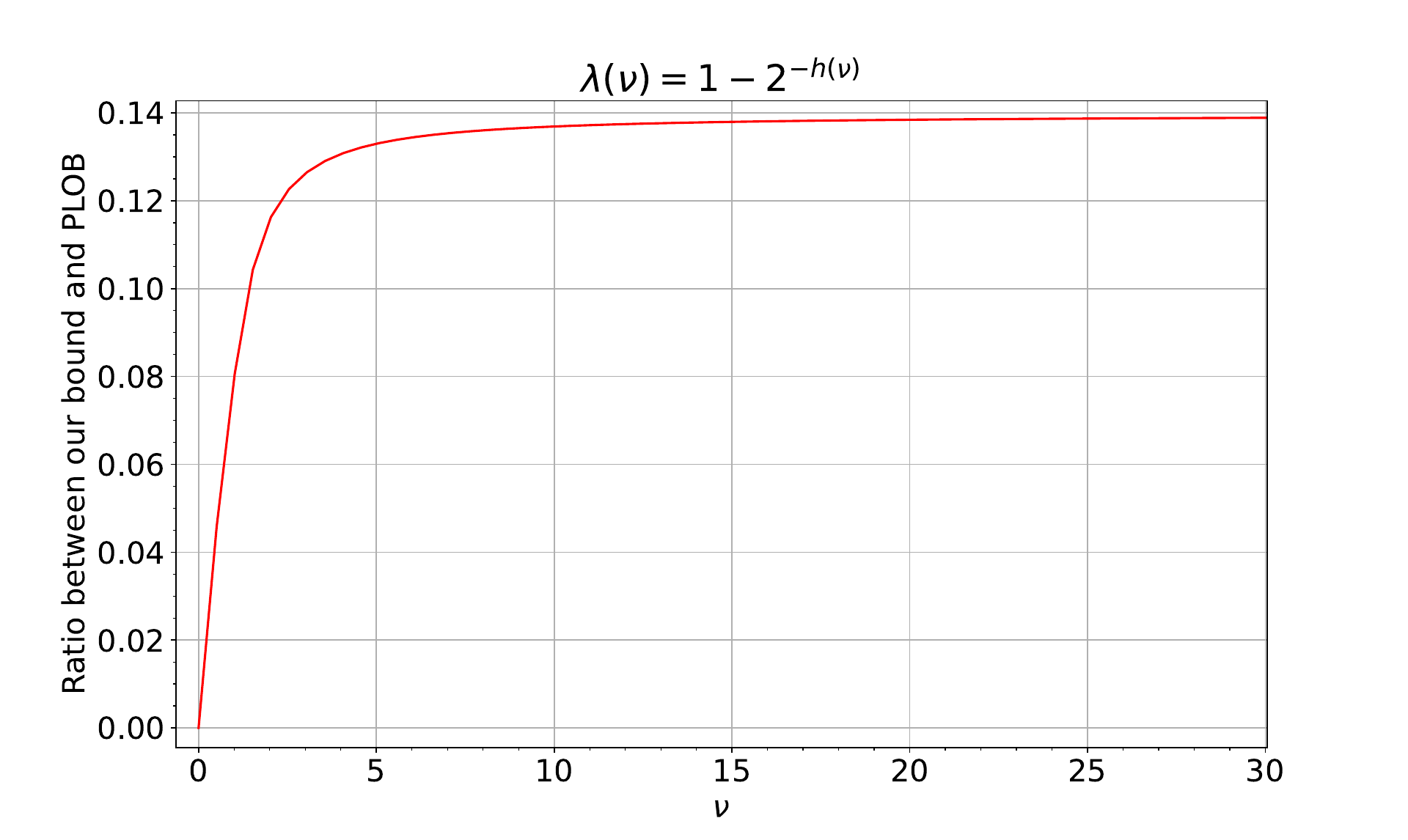} \\ 
(b) 
\end{tabular}
\caption{\textbf{(a).}~Bounds on the two-way quantum capacity of the thermal attenuator $Q_2(\mathcal{E}_{\lambda(\nu),\nu})$ plotted with respect to $\nu$, where the transmissivity is equal to the critical value $\lambda(\nu)\coloneqq 1-2^{-h(\nu)}$. The red curve is our lower bound calculated by exploiting~\eqref{lowQ2_delta}. The black curve is the best known lower bound on $Q_2(\mathcal{E}_{\lambda(\nu),\nu})$, which is the reverse coherent information lower bound reported in~\eqref{lowQ2} (which is zero since $\lambda(\nu)= 1-2^{-h(\nu)}$). The green curve is the PLOB upper bound reported in~\eqref{PLOB_Q2}. These bounds are also bounds on the secret-key capacity $K(\mathcal{E}_{\lambda(\nu),\nu})$. \textbf{(b).}~Ratio between our new lower bound on the two-way quantum and secret-key capacities in~\eqref{lowQ2_delta} and the PLOB bound in~\eqref{PLOB_Q2} as a function of $\nu$ where the transmissivity is $\lambda(\nu)\coloneqq 1-2^{-h(\nu)}$.}
\label{figure_nu2_delta}

\end{figure}

\begin{figure}[t]
	\centering
	\includegraphics[width=1\linewidth]{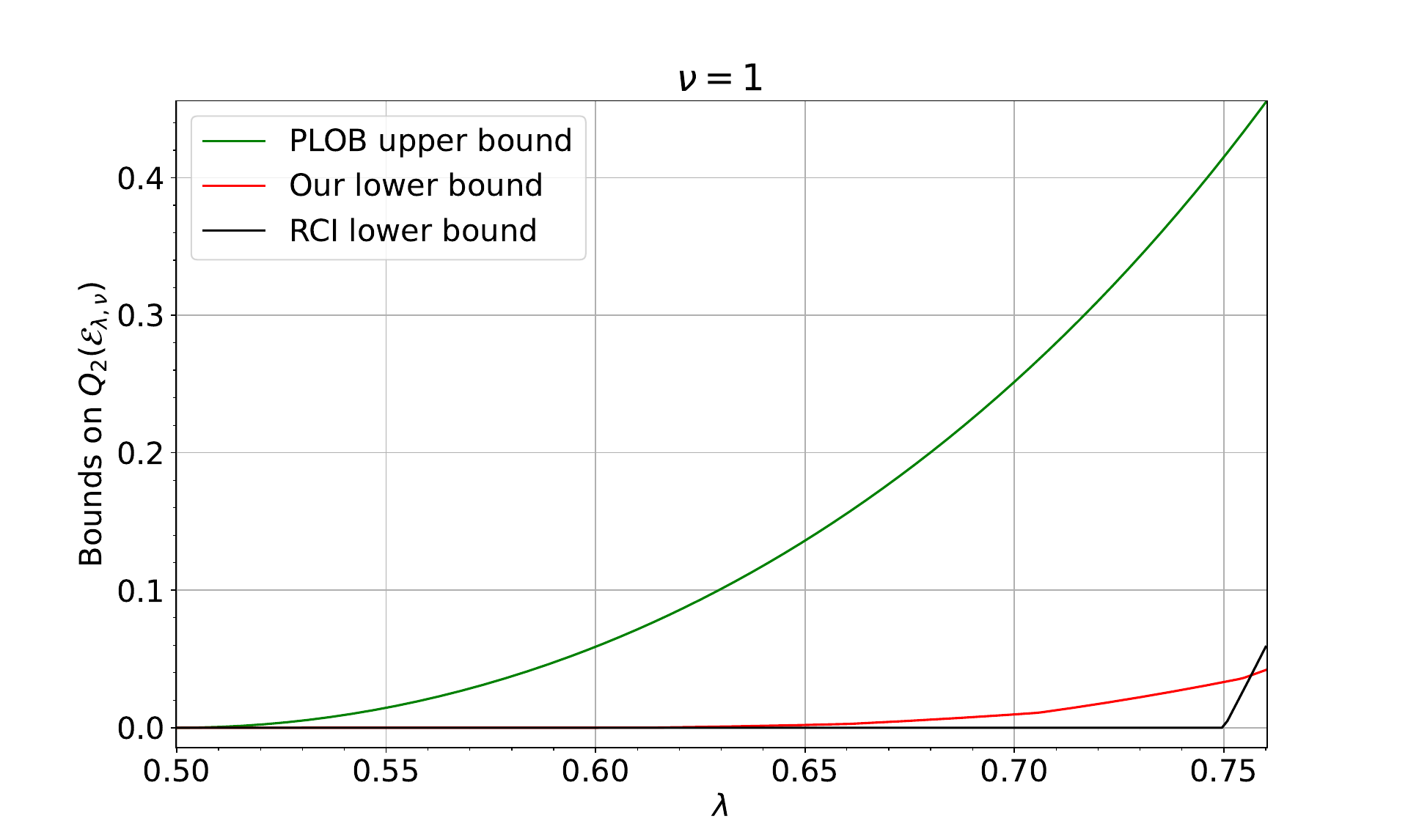}
	\includegraphics[width=1\linewidth]{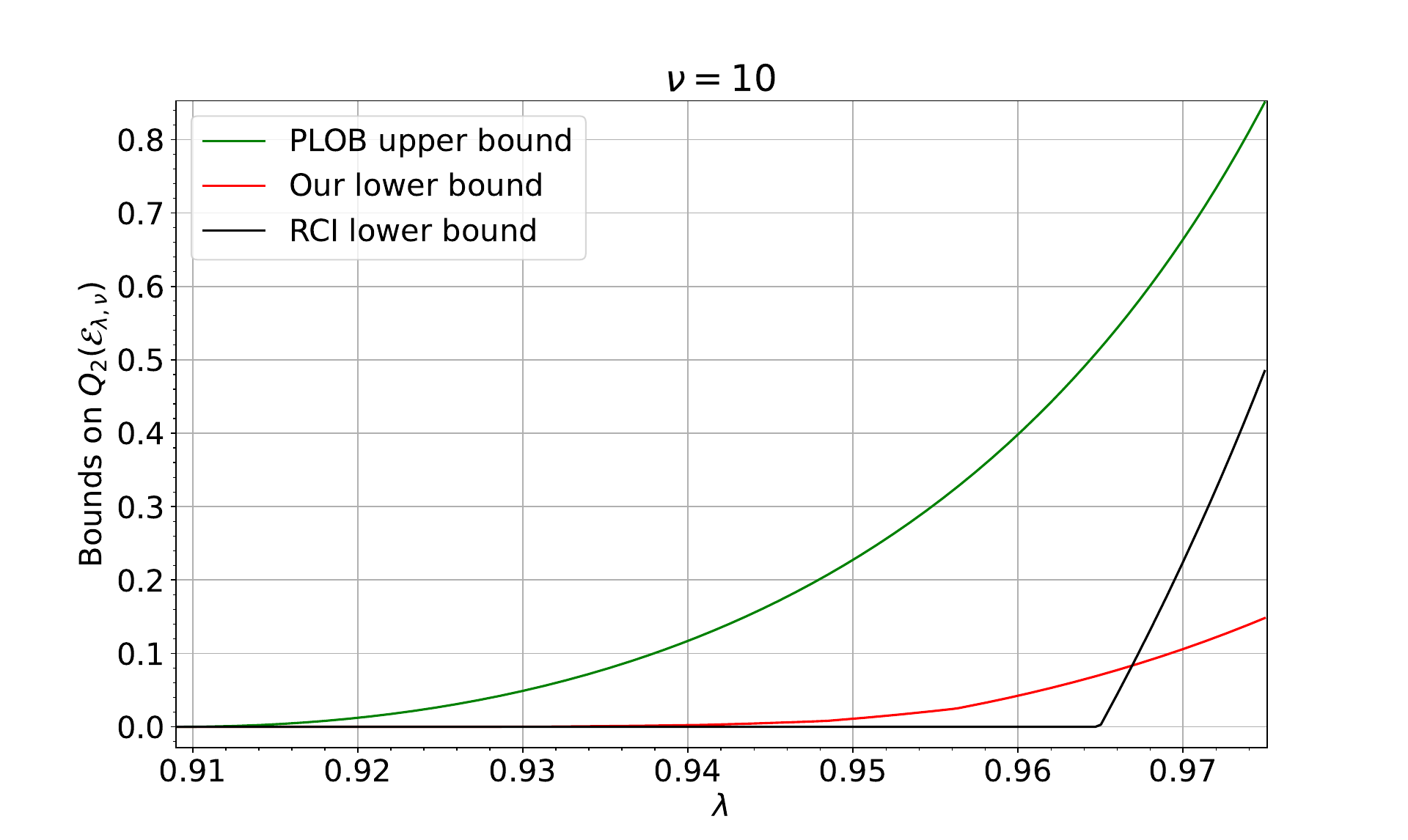} 
	\caption{Bounds on the two-way quantum capacity of the thermal attenuator $Q_2(\mathcal{E}_{\lambda,\nu})$ plotted with respect to $\lambda$. The red line is our new lower bound obtained by exploiting~\eqref{lowQ2_delta}. The black line is the best known lower bound on $Q_2(\mathcal{E}_{\lambda,\nu})$, which is the reverse coherent information lower bound reported in~\eqref{lowQ2}. The green line is the PLOB upper bound reported in~\eqref{PLOB_Q2}. These bounds are also bounds on the secret-key capacity $K(\mathcal{E}_{\lambda,\nu})$.}
	\label{rate_vs_lambda1}
\end{figure}

\begin{figure}[t]
	\centering
	\includegraphics[width=1\linewidth]{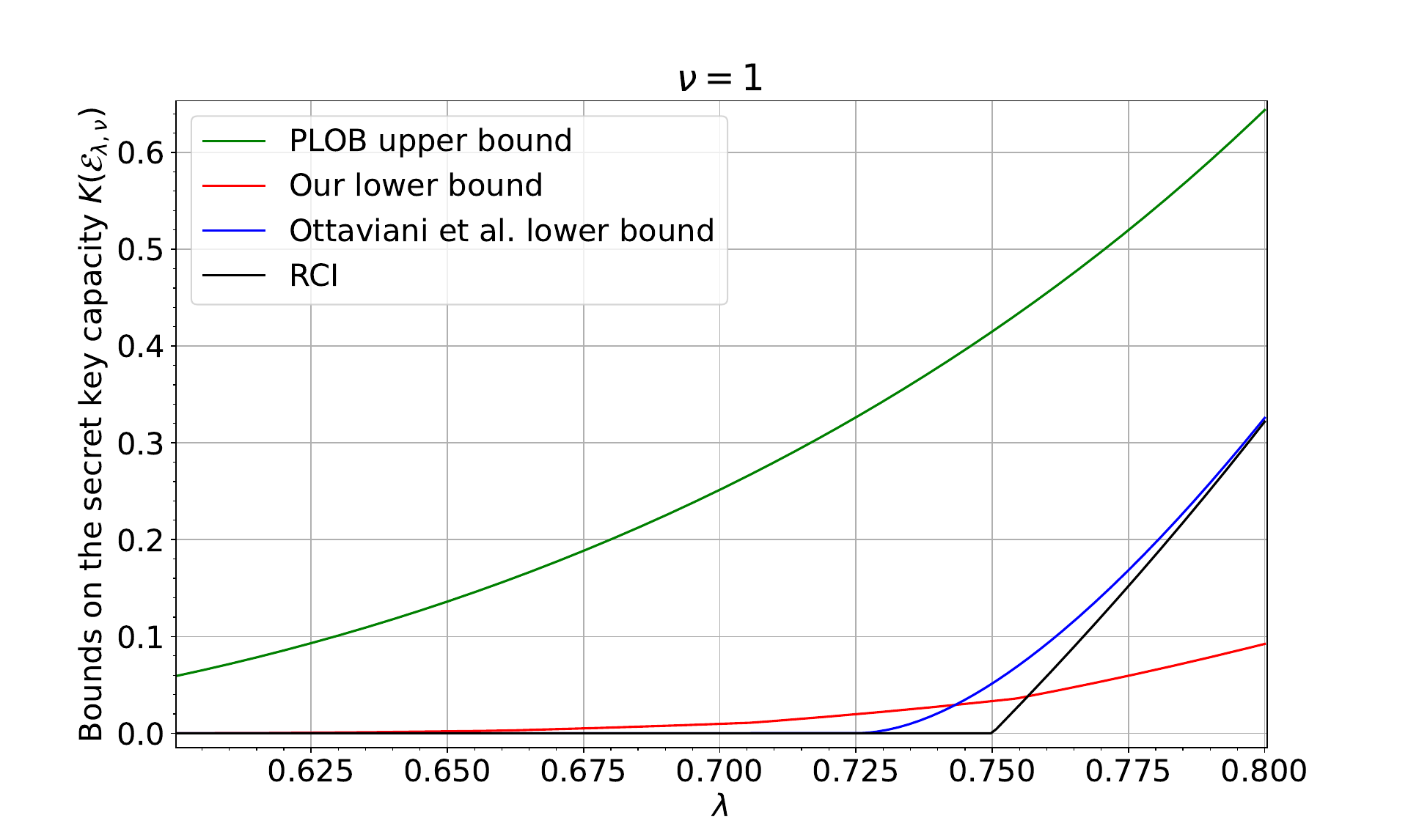}
	\includegraphics[width=1\linewidth]{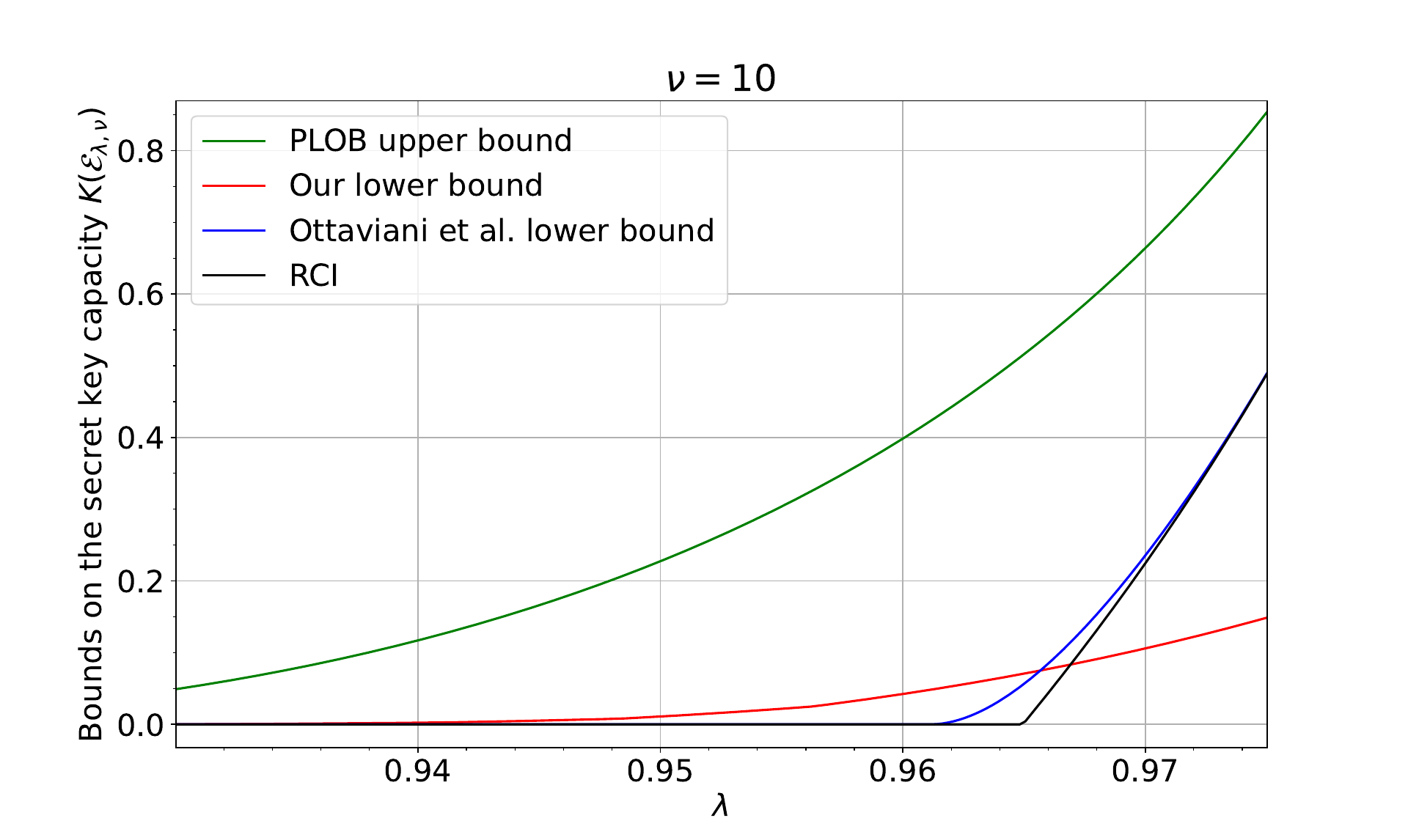} 
	\caption{Bounds on the secret-key capacity of the thermal attenuator $K(\mathcal{E}_{\lambda,\nu})$ plotted with respect to $\lambda$. The red line is our new lower bound obtained by exploiting~\eqref{lowQ2_delta}, the black line is the bound in~\eqref{lowQ2} calculated by evaluating the reverse coherent information in~\eqref{proof_lower}, the blue line is the best known lower bound discovered by~\cite{Ottaviani_new_lower}, and the green line is the PLOB upper bound reported in~\eqref{PLOB_Q2}.}
	\label{secret_key_nu}
\end{figure}

\begin{figure}[t]
	\centering
	\includegraphics[width=1\linewidth]{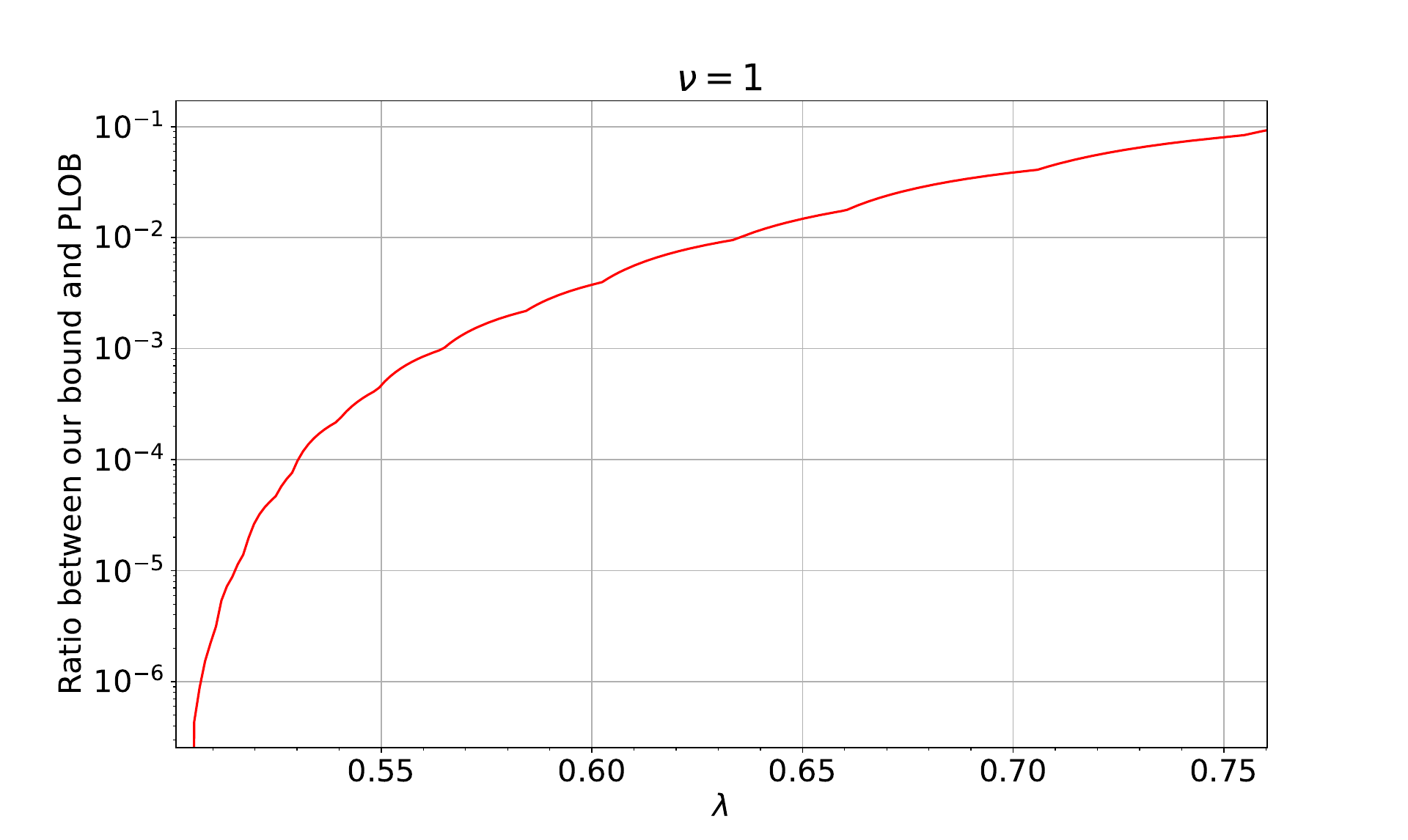}
	\caption{Ratio between our lower bound on the two-way quantum and secret-key capacities of the thermal attenuator in~\eqref{lowQ2_delta} and the PLOB bound in~\eqref{PLOB_Q2} as a function of $\lambda$ for $\nu=1$.}
	\label{log1}
\end{figure}

\begin{figure}

\begin{tabular}{cc}
  \includegraphics[width=0.5\linewidth]{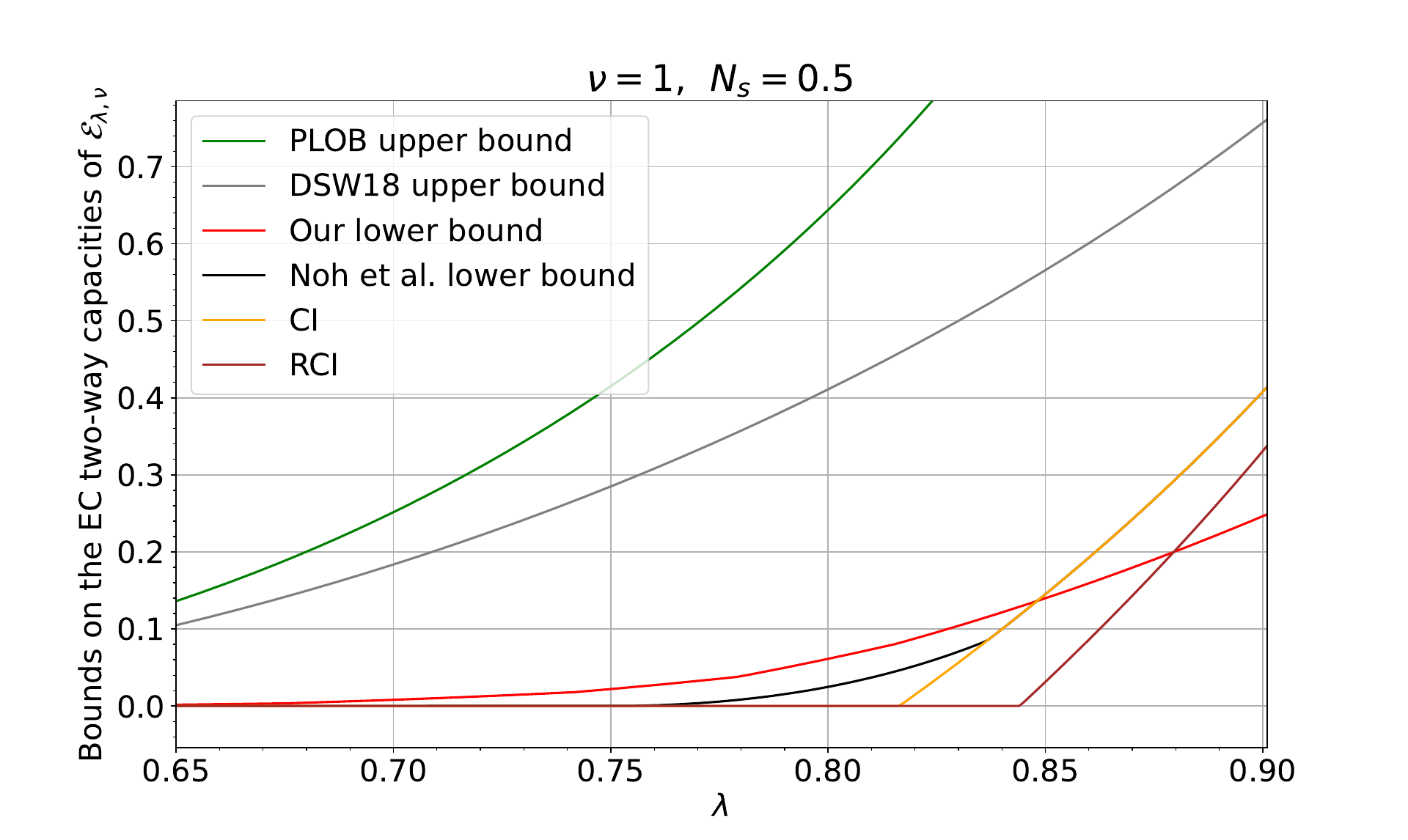} &   \includegraphics[width=0.5\linewidth]{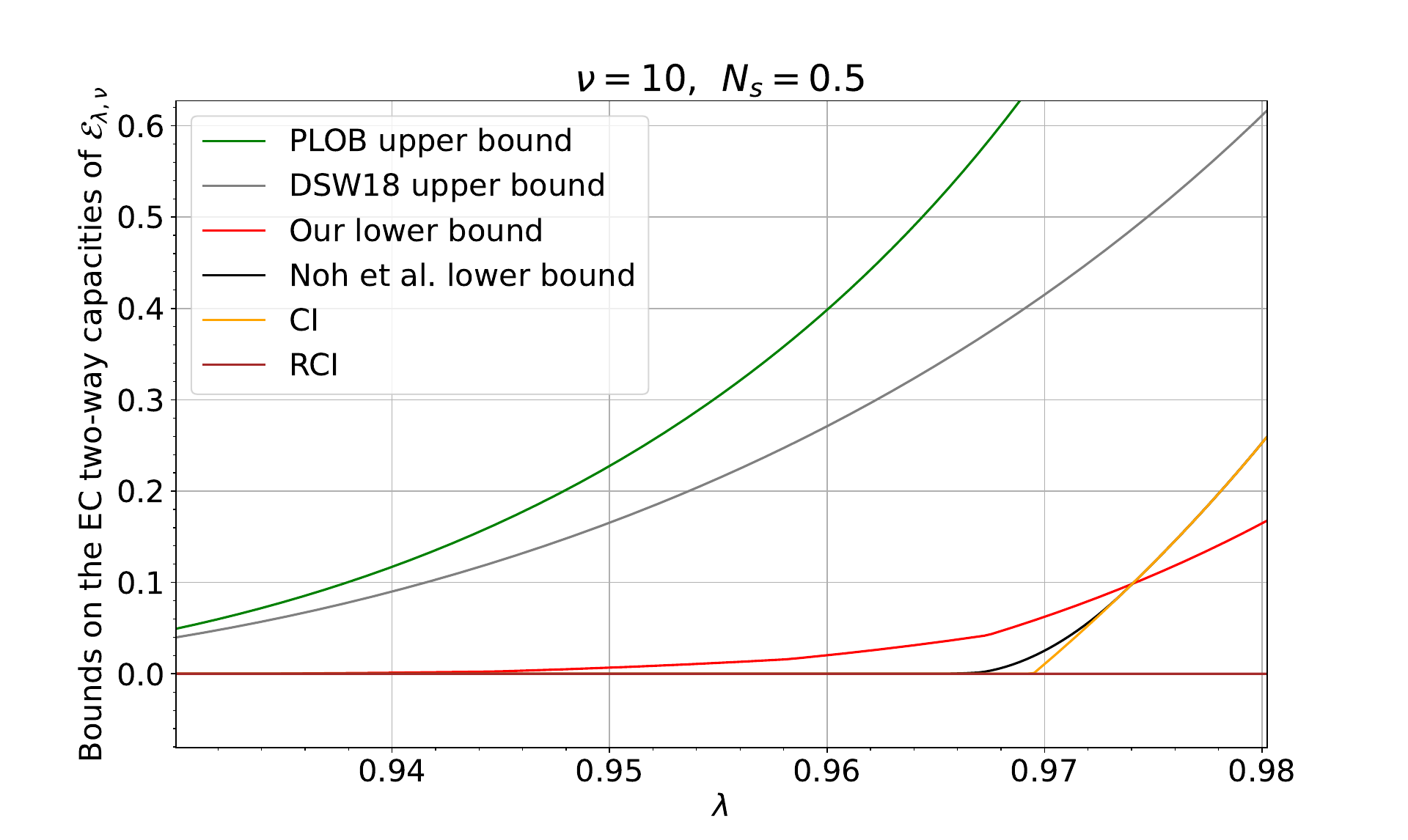} \\
 \includegraphics[width=0.5\linewidth]{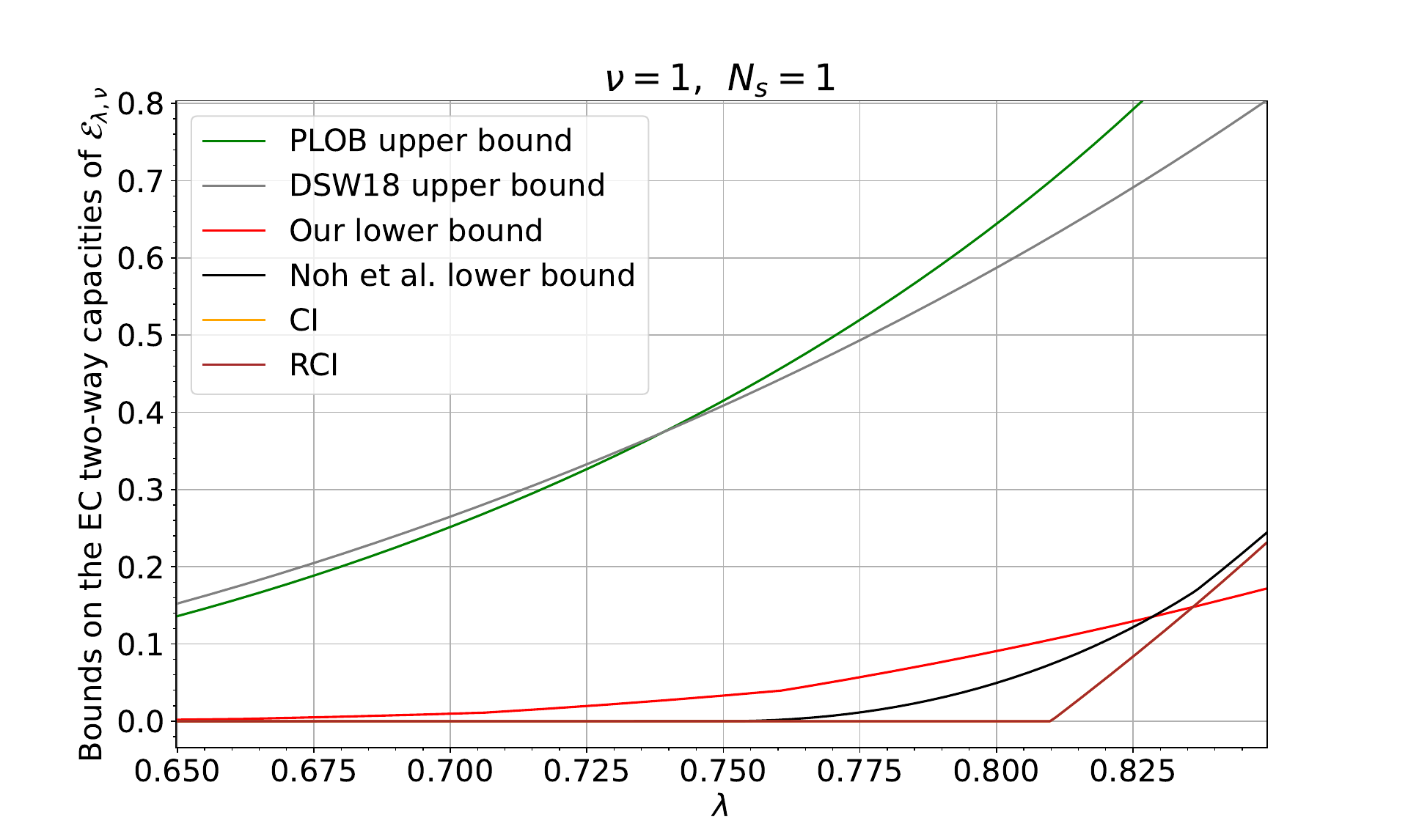} &   \includegraphics[width=0.5\linewidth]{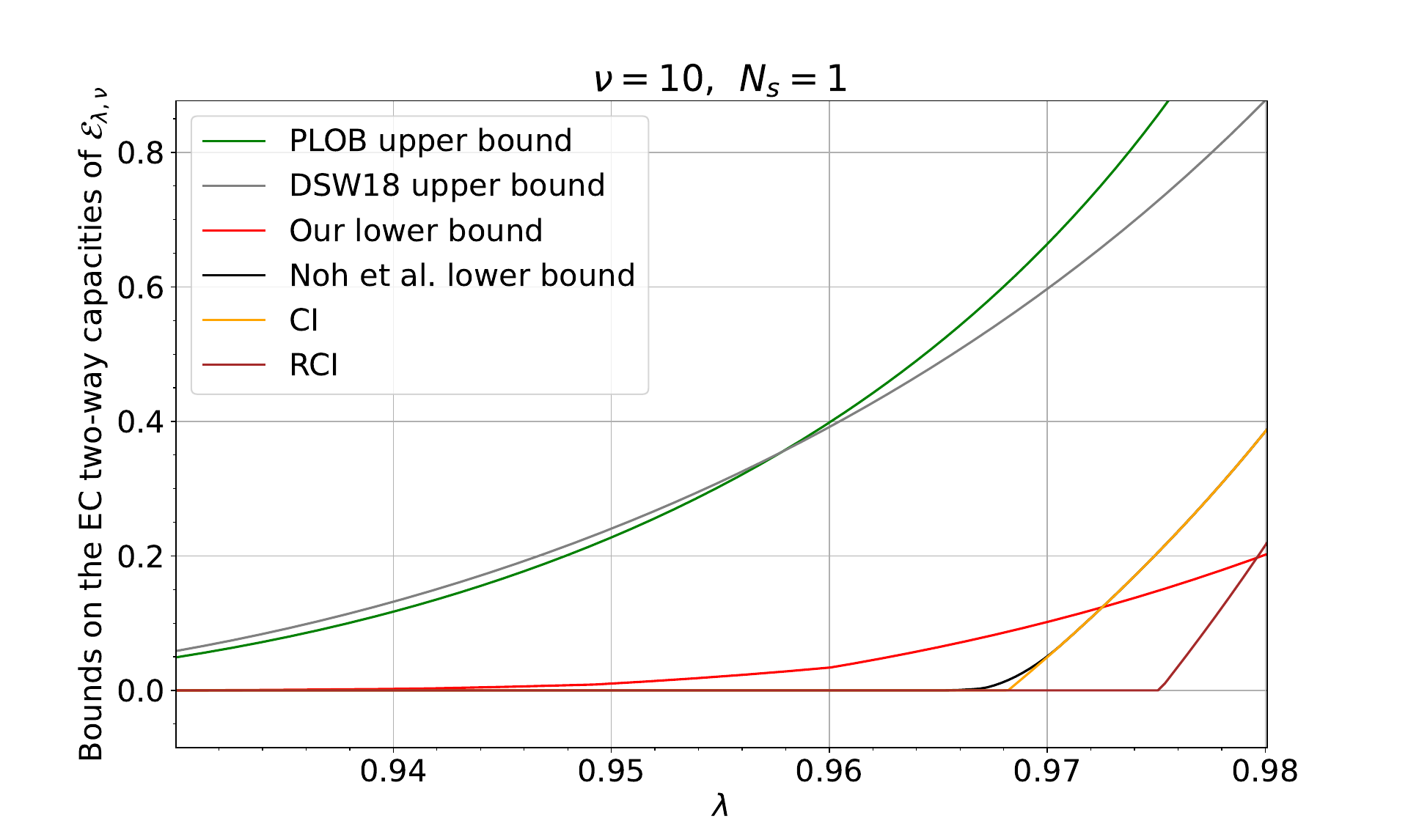} \\
 \includegraphics[width=0.5\linewidth]{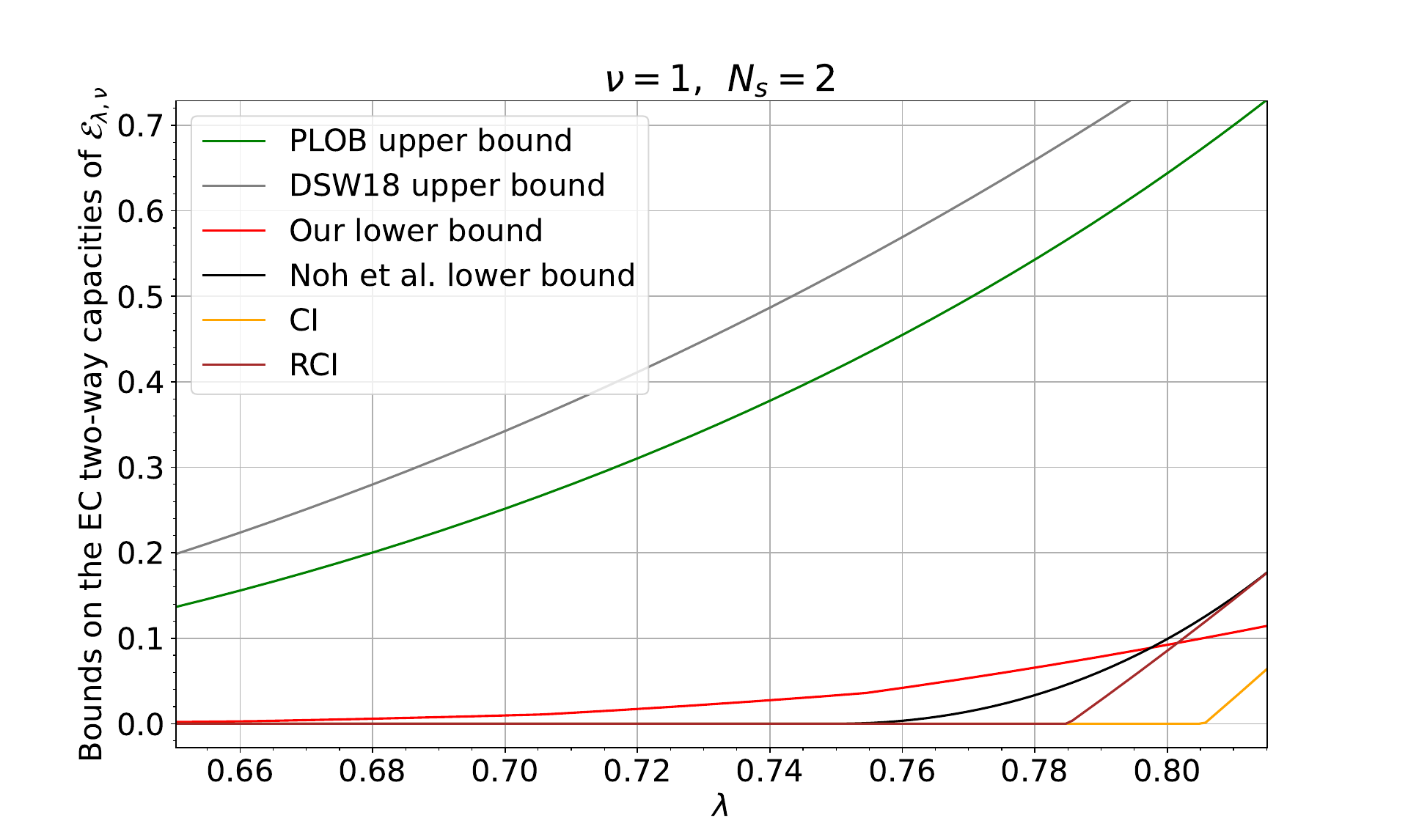} &  \includegraphics[width=0.5\linewidth]{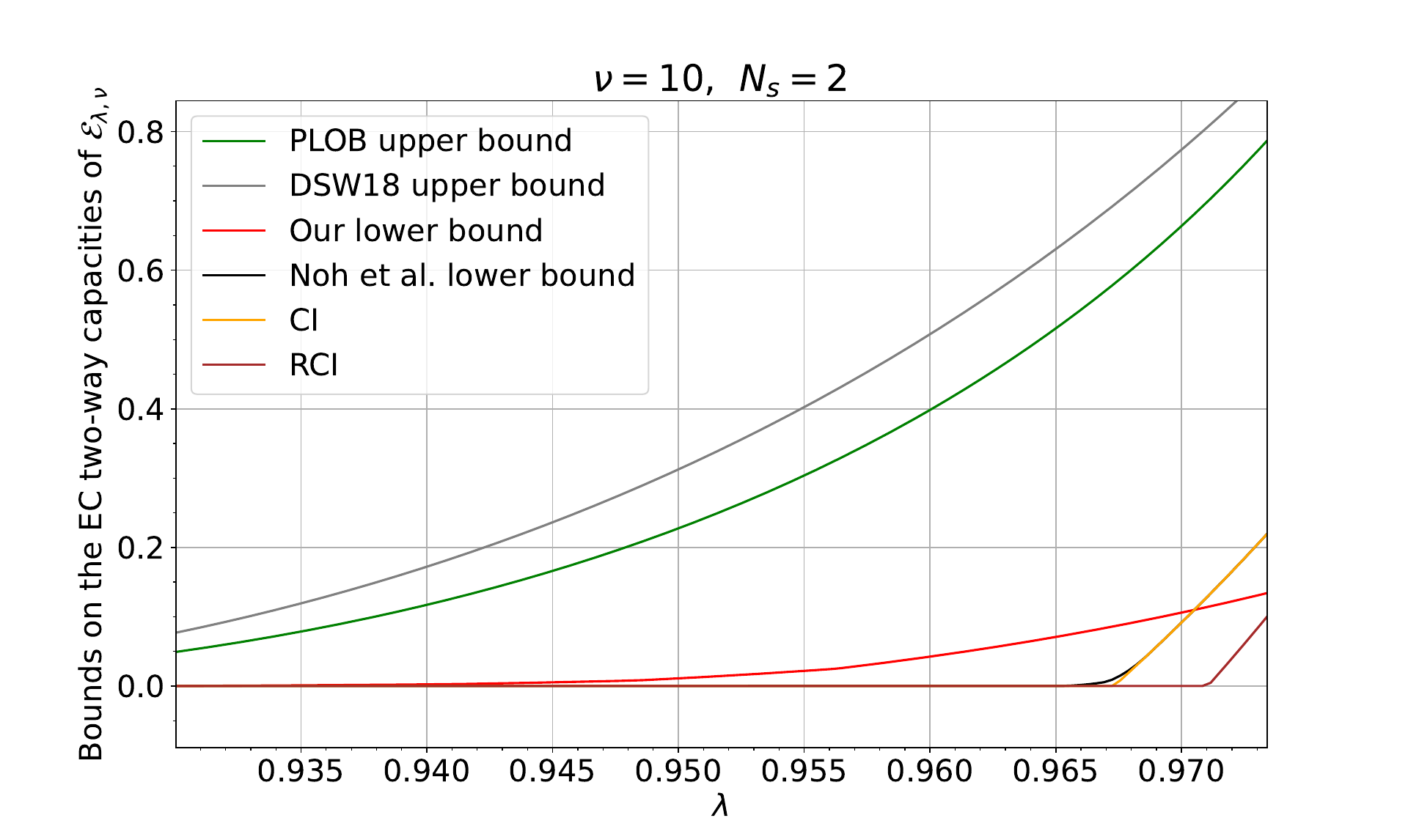} \\
\end{tabular}
\caption{Bounds on the energy-constrained two-way quantum capacity  $Q_2(\mathcal{E}_{\lambda,\nu},N_s)$ and secret-key capacity $K(\mathcal{E}_{\lambda,\nu},N_s)$ of the thermal attenuator plotted with respect to $\lambda$ for different choices of $\nu$ and of the energy constraint $N_s$. The red line is our new lower bound obtained by exploiting~\eqref{lowQ2_deltaEC}, the black line is the NPJ lower bound~\cite{Noh2020} reported in~\eqref{npj_bound_therm}, the yellow line is the coherent information lower bound reported in~\eqref{EC_coh_therm_att}, the brown line is the reverse coherent information lower bound reported in~\eqref{EC_coh_therm_att}, the grey line is the DSW18 upper bound~\cite{Davis2018}, and the green line is the PLOB upper bound reported in~\eqref{PLOB_Q2}. }
\label{ECfigures}

\end{figure}

\begin{figure}

\begin{tabular}{cc}
  \includegraphics[width=0.5\linewidth]{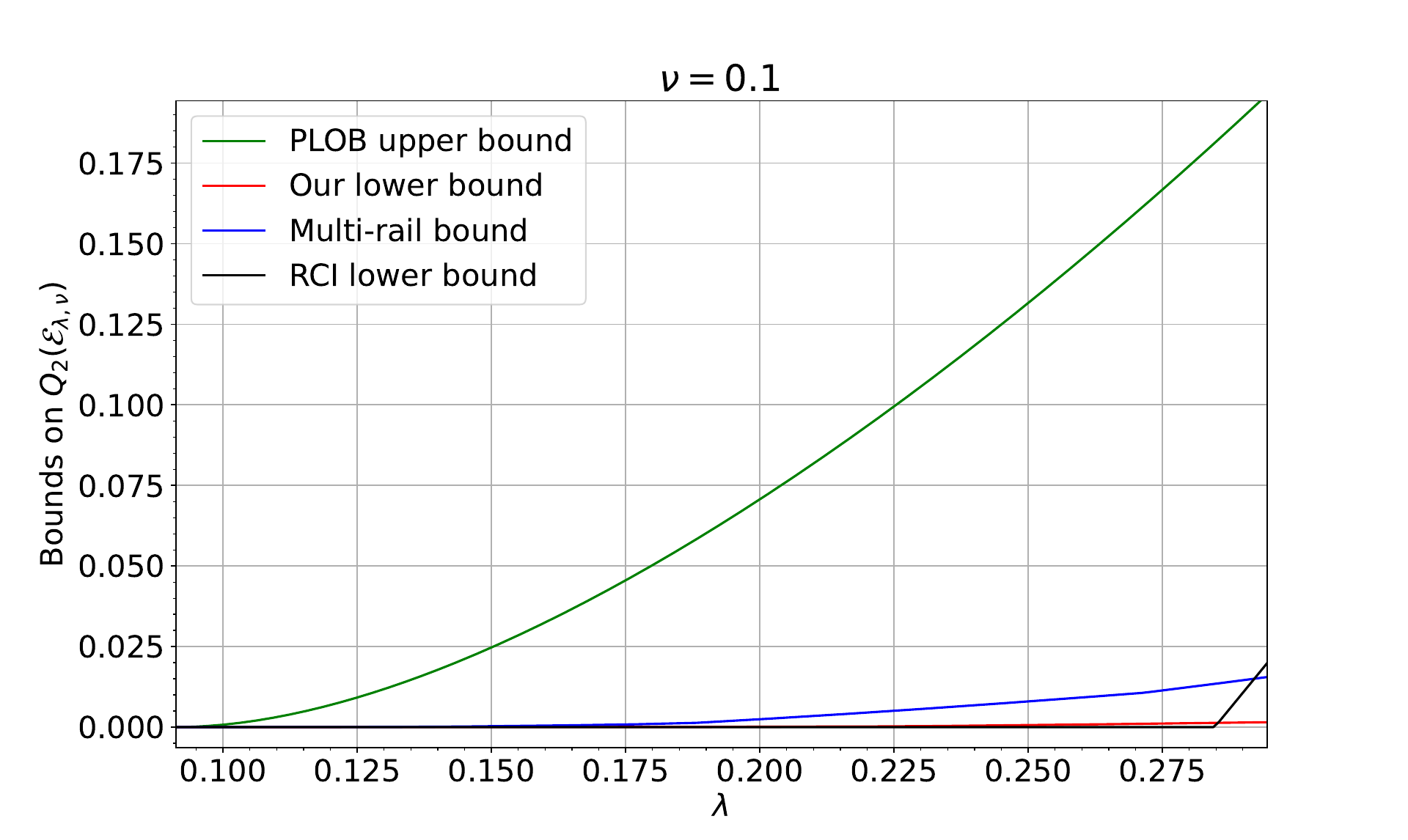} &   \includegraphics[width=0.5\linewidth]{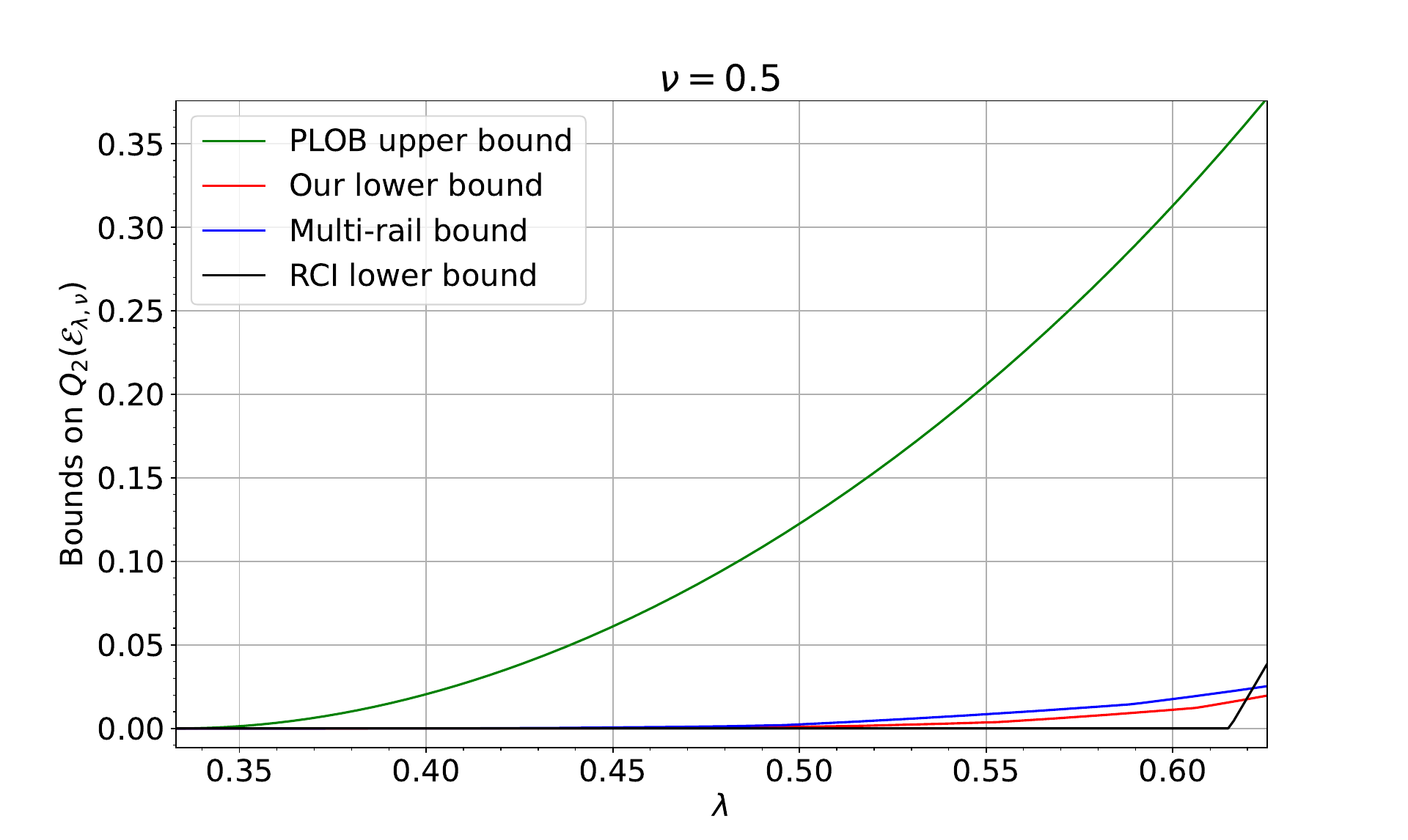} \\
 \includegraphics[width=0.5\linewidth]{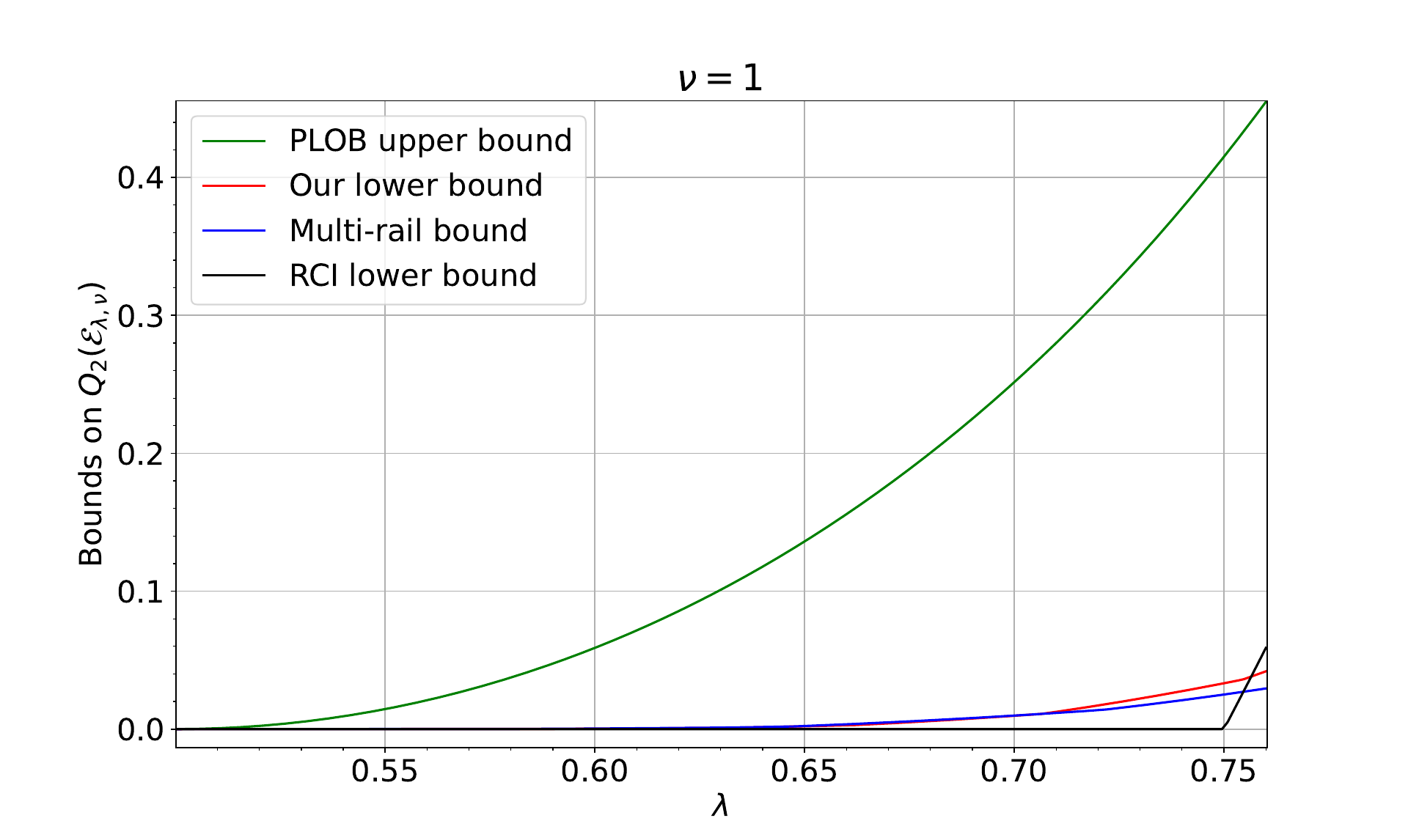} &   \includegraphics[width=0.5\linewidth]{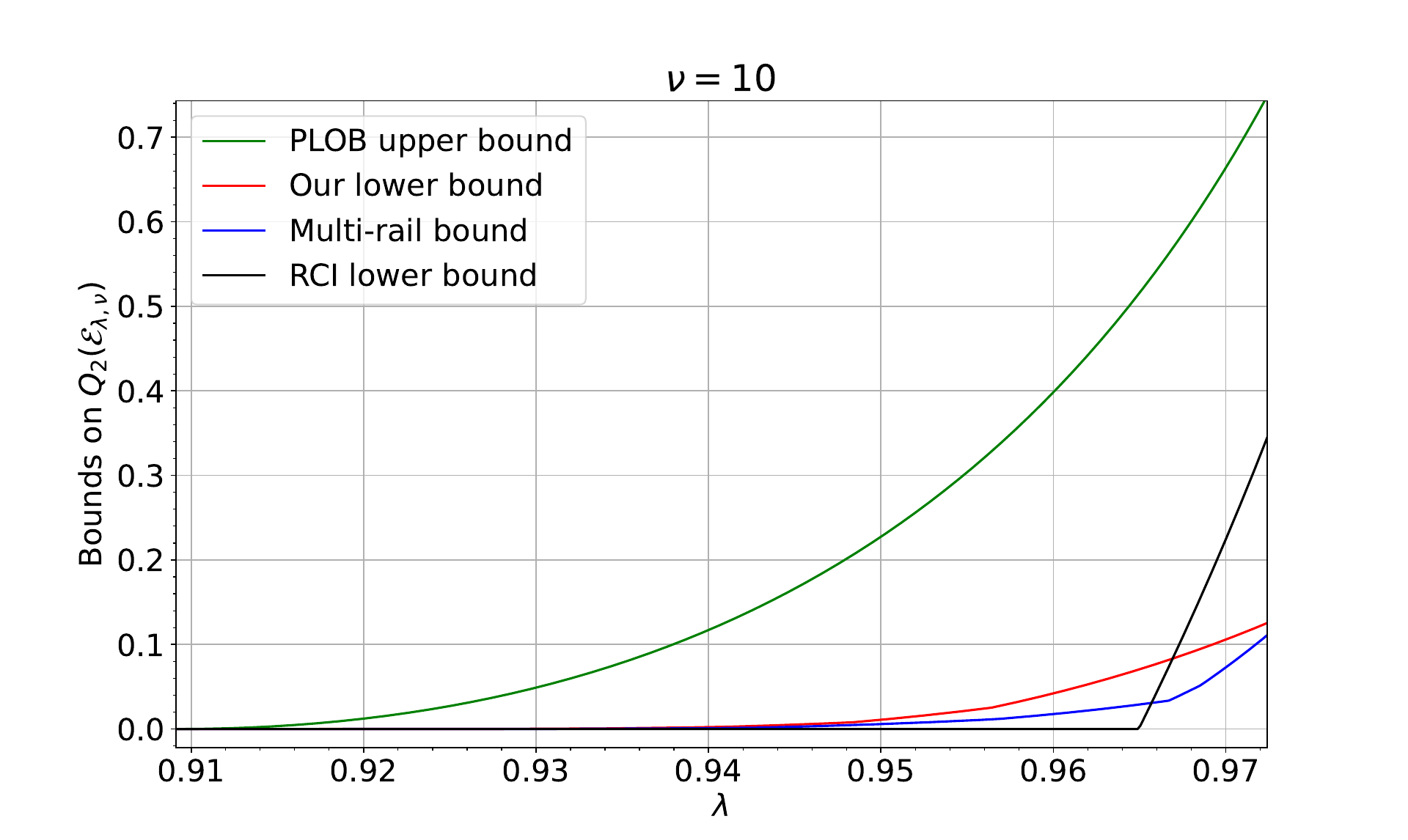} \\
\end{tabular}
\caption{Bounds on the two-way quantum capacity of the thermal attenuator $Q_2(\mathcal{E}_{\lambda,\nu})$ plotted with respect to $\lambda$. The blue line is our multi-rail lower bound obtained by exploiting~\eqref{multiplerail_low_bound}. The red line is our lower bound reported in~\eqref{lowQ2_delta}. The black line is the best known lower bound on $Q_2(\mathcal{E}_{\lambda,\nu})$, which is the reverse coherent information lower bound reported in~\eqref{lowQ2}. The green line is the PLOB upper bound reported in~\eqref{PLOB_Q2}. These bounds are also bounds on the secret-key capacity $K(\mathcal{E}_{\lambda,\nu})$.}
\label{multiplerail_figures}

\end{figure}

 \begin{figure}
\begin{tabular}{c}
  \includegraphics[width=1.0\linewidth]{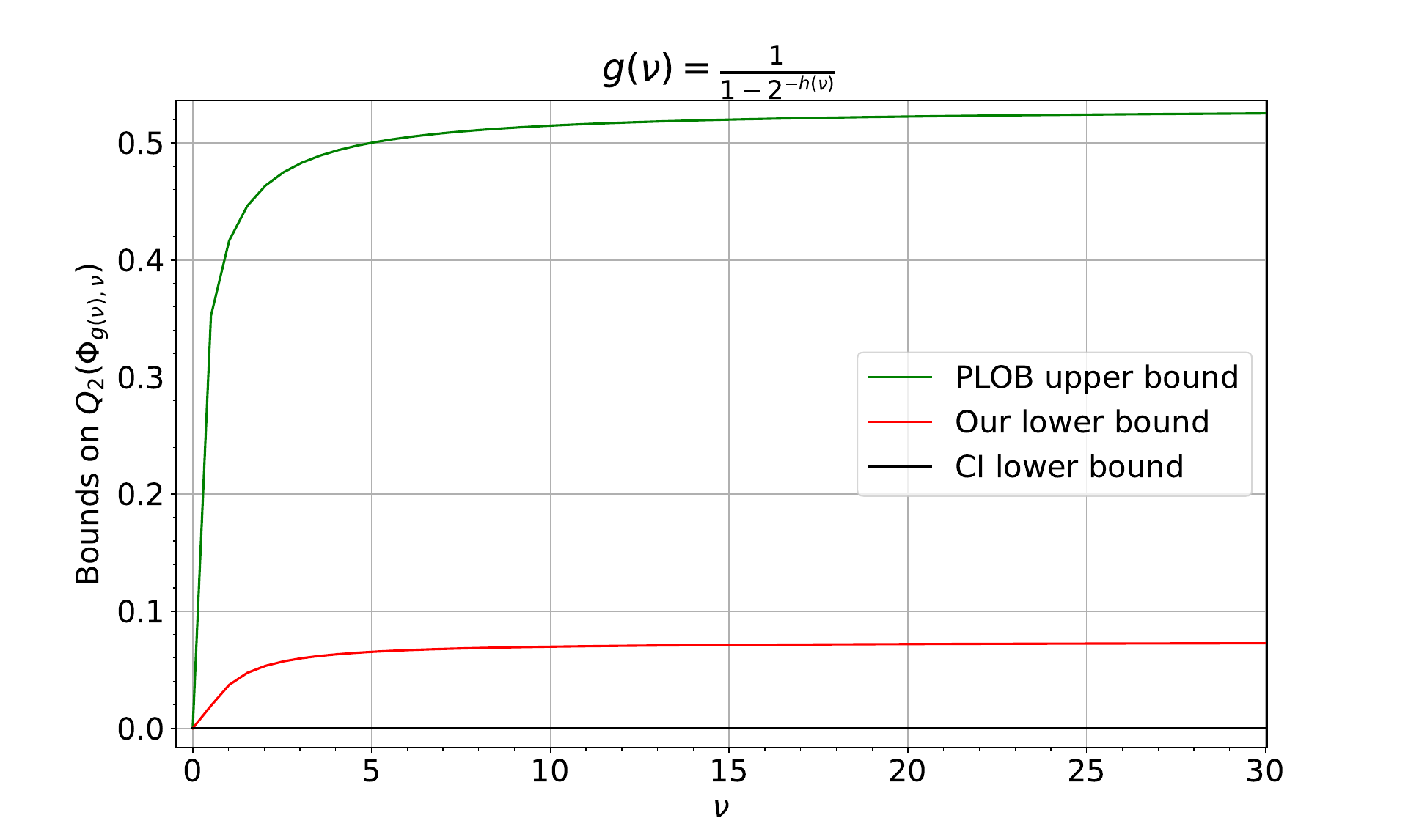} \\  
(a)\\
 \includegraphics[width=1.0\linewidth]{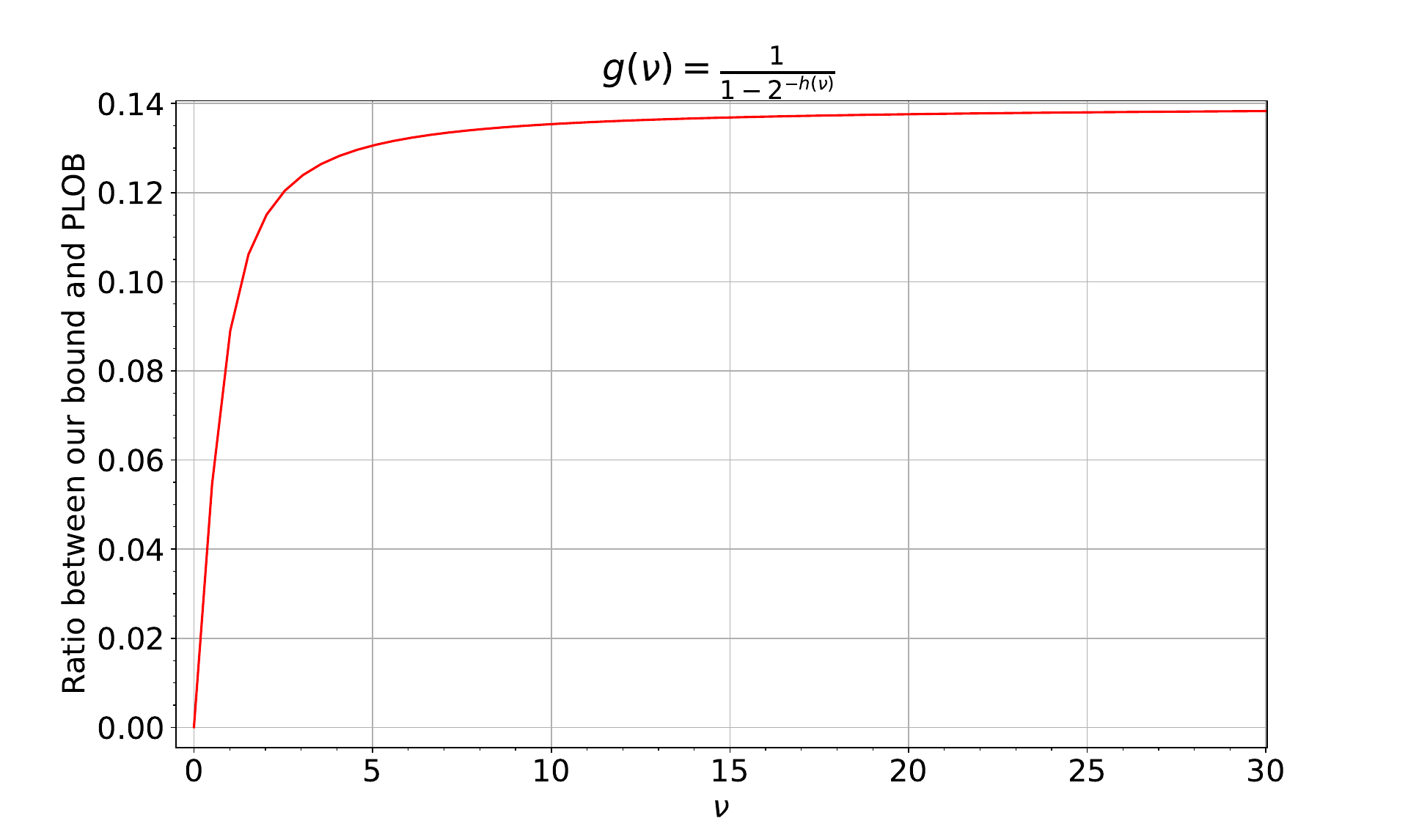} \\ 
(b) 
\end{tabular}
\caption{\textbf{(a).}~Bounds on the two-way quantum capacity of the thermal amplifier $Q_2(\Phi_{g(\nu),\nu})$ plotted with respect to $\nu$, where the gain is equal to the critical value $g(\nu)= \frac{1}{1-2^{-h(\nu)}}$. The red curve is our new lower bound calculated by exploiting~\eqref{lowQ2_delta_amp}. The black curve is the best known lower bound, i.e.~the coherent information lower bound reported in~\eqref{lowQ2_amp} (which is zero since $g(\nu)= \frac{1}{1-2^{-h(\nu)}}$). The green curve is the PLOB upper bound reported in~\eqref{PLOB_amp}. These bounds are also bounds on the secret-key capacity $K(\Phi_{g(\nu),\nu})$.  \textbf{(b).}~Ratio between our new lower bound in~\eqref{lowQ2_delta_amp} and the PLOB bound in~\eqref{PLOB_amp} as a function of $\nu$ where the gain is $g(\nu)\coloneqq \frac{1}{1-2^{-h(\nu)}}$.}
\label{capvsnu_amp}
\end{figure}

\begin{figure}[t]
	\centering
	\includegraphics[width=1\linewidth]{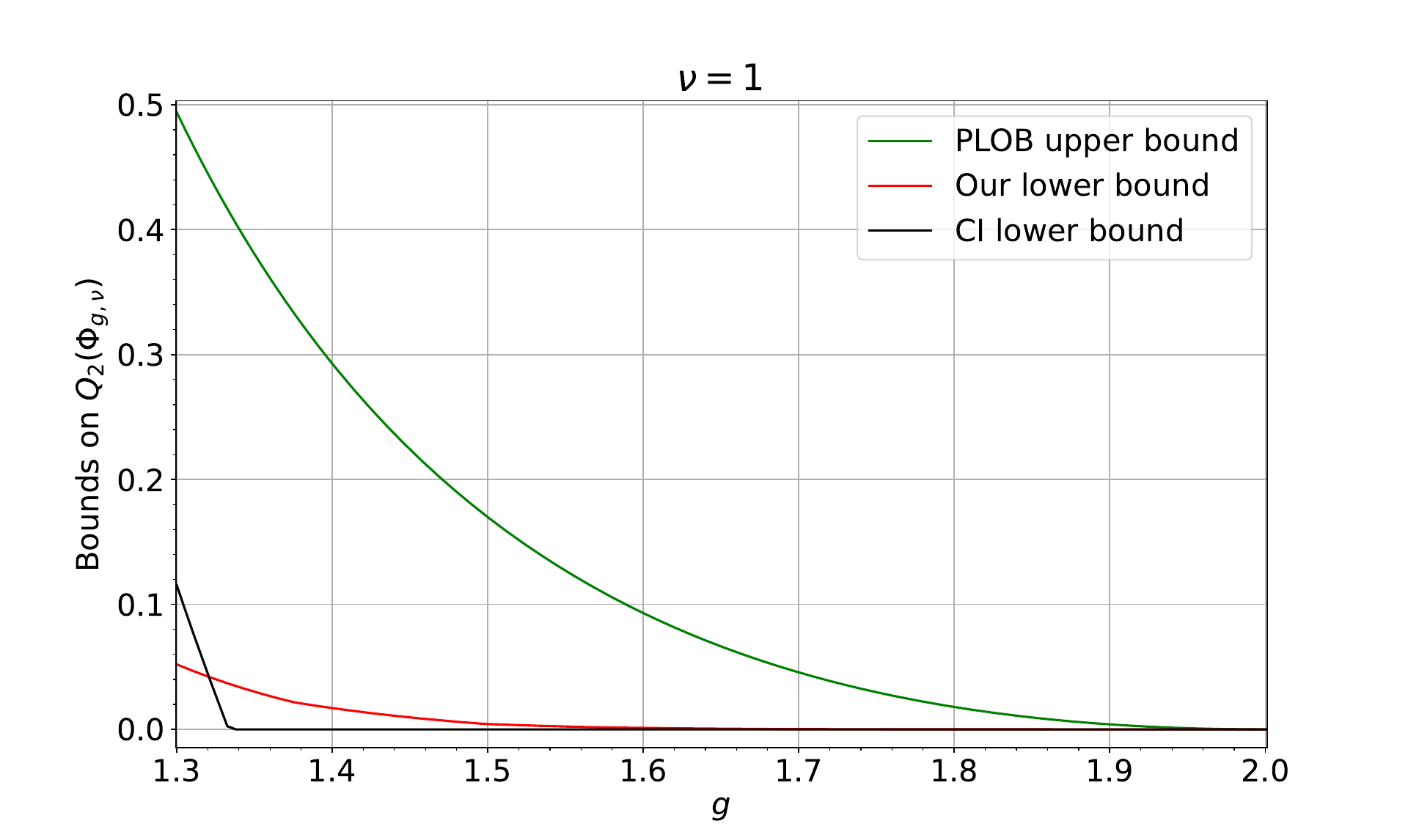}
	\includegraphics[width=1\linewidth]{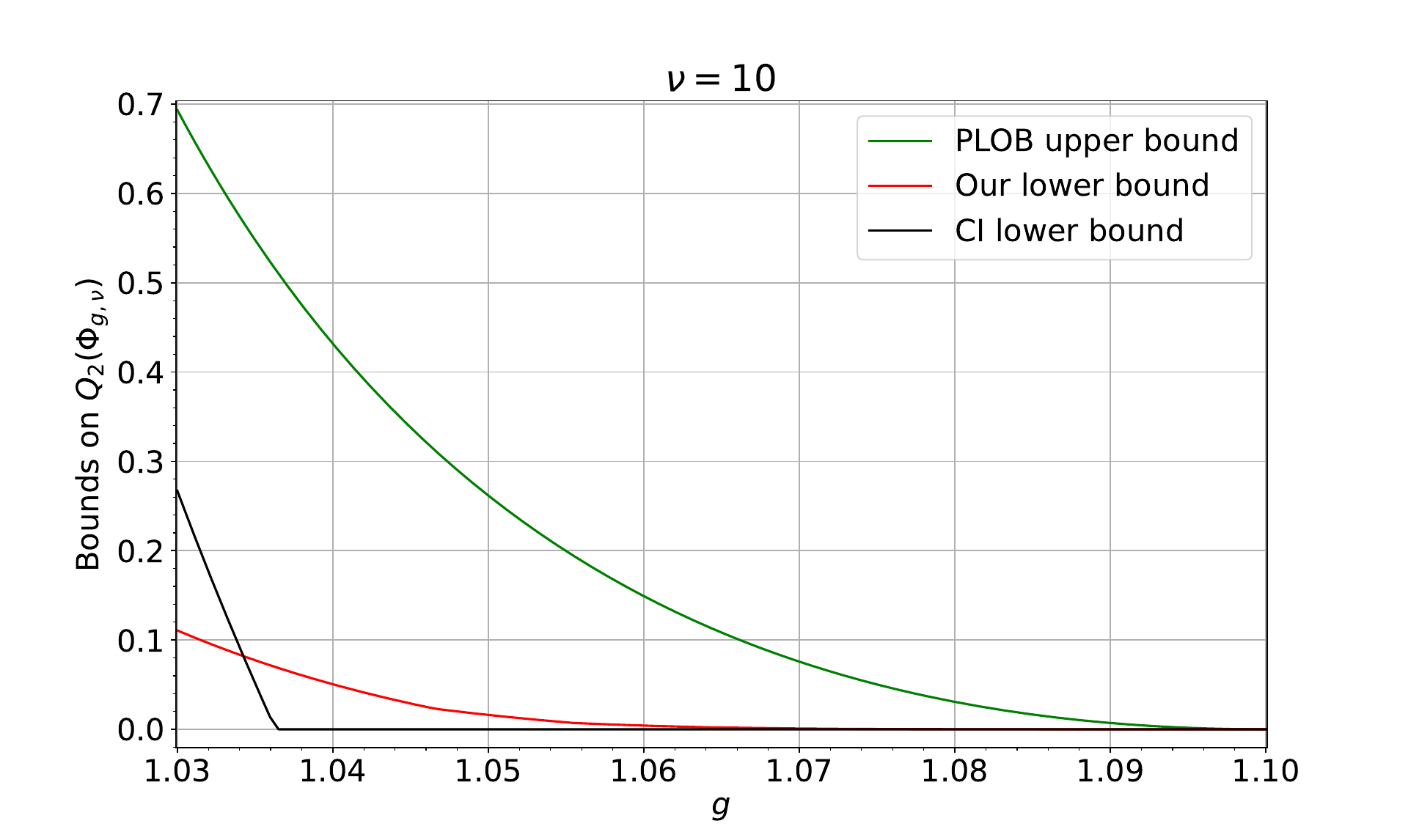} 
	\caption{Bounds on the two-way quantum capacity of thermal amplifier $Q_2(\Phi_{g,\nu})$ plotted with respect to $g$. The red line is our new lower bound obtained by exploiting~\eqref{lowQ2_delta_amp}. The black line is the best known lower bound on $Q_2(\Phi_{g,\nu})$, which is the coherent information lower bound reported in~\eqref{lowQ2_amp}. The green line is the PLOB upper bound reported in~\eqref{PLOB_amp}. These bounds are also bounds on the secret-key capacity $K(\Phi_{g,\nu})$.}
	\label{bound_vs_nu1_amp}
\end{figure}

\begin{figure}[t]
	\centering
	\includegraphics[width=1\linewidth]{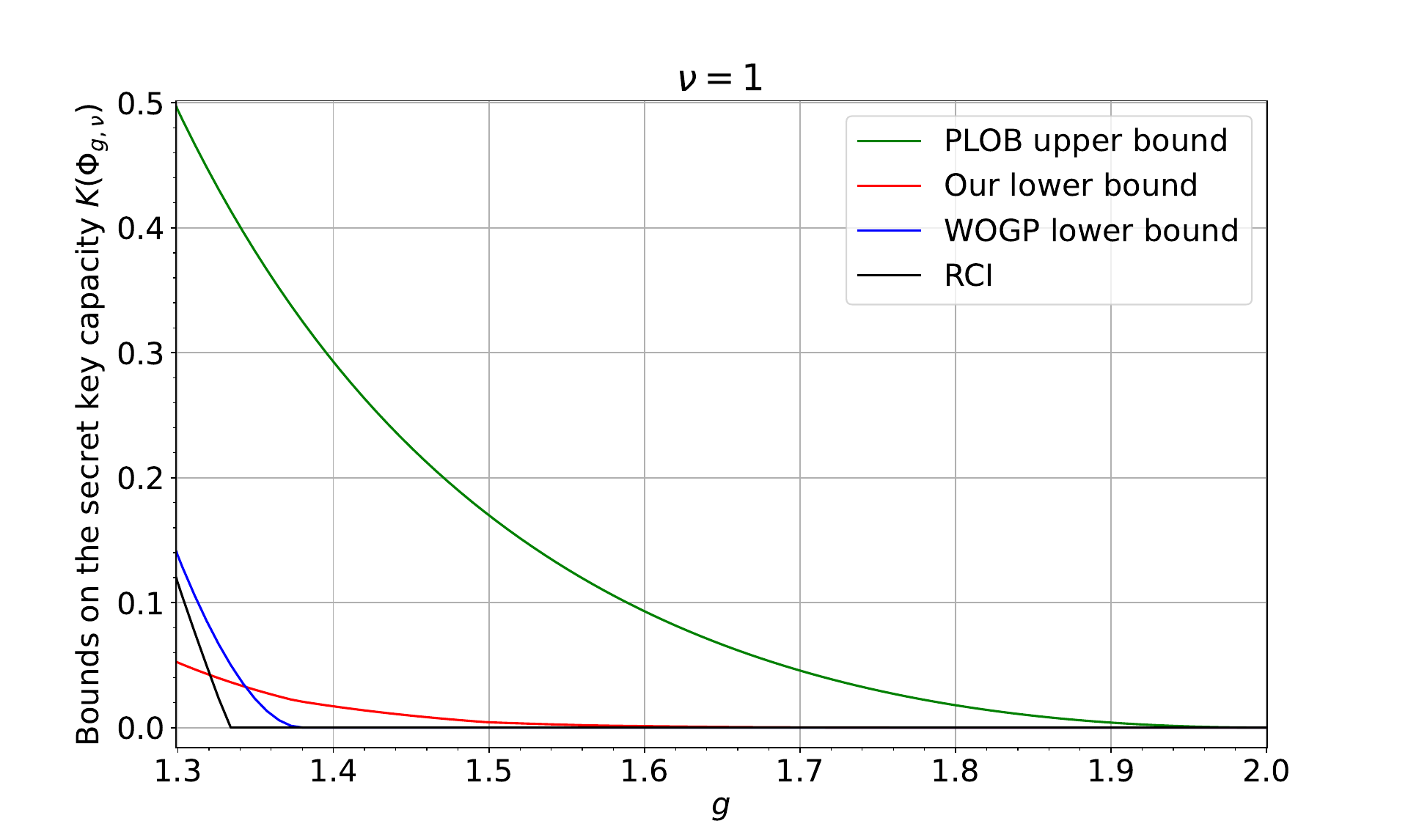}
	\includegraphics[width=1\linewidth]{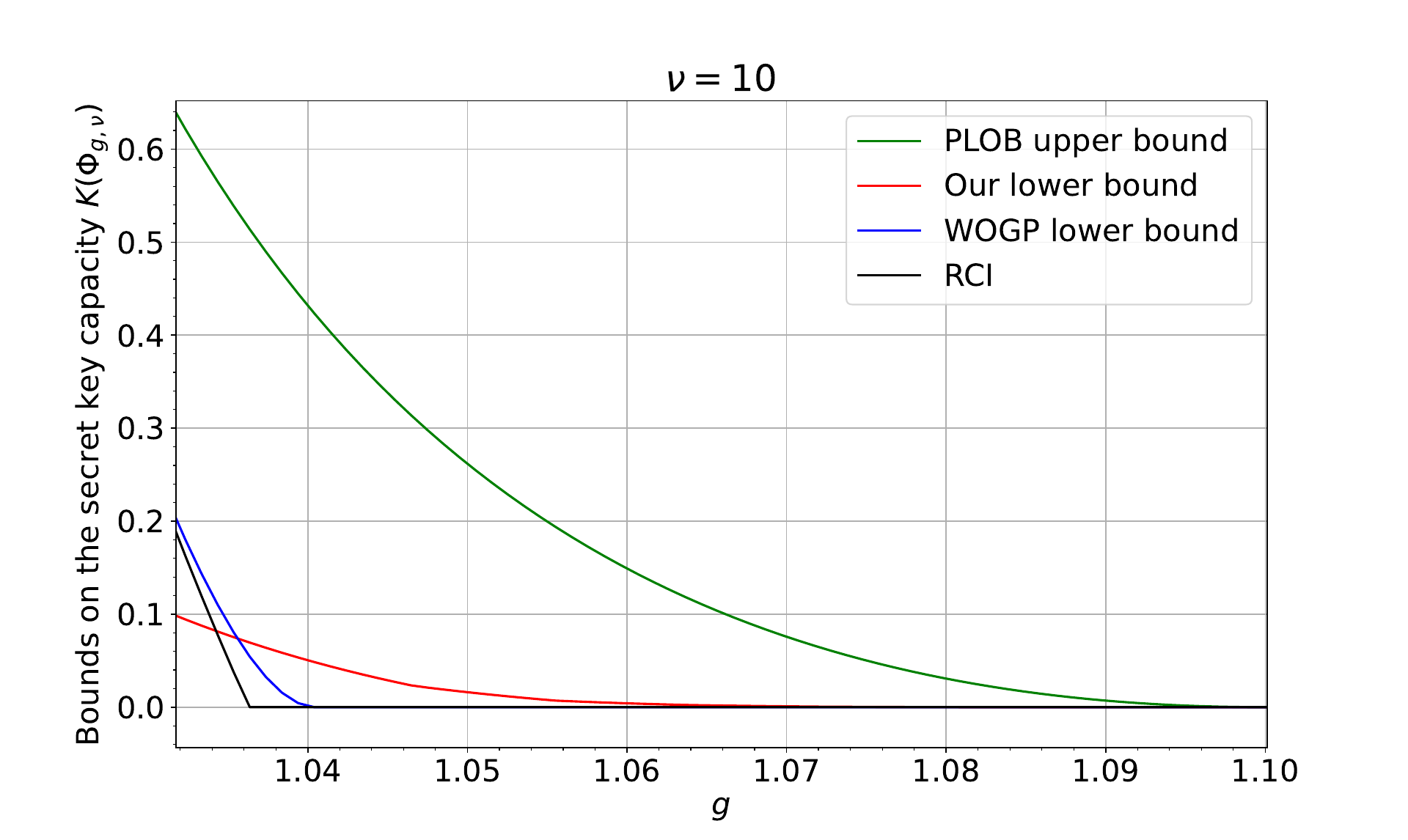} 
	\caption{Bounds on the secret-key capacity of the thermal amplifier $K(\Phi_{g,\nu})$ plotted with respect to $g$. The red line is our new lower bound obtained by exploiting~\eqref{lowQ2_delta_amp}, the black line is the bound in~\eqref{lowQ2_amp} calculated by evaluating the coherent information in~\eqref{proof_lower_ampl}, the blue line is the WOGP lower bound~\cite{Ottaviani_new_lower}, and the green line is the PLOB upper bound reported in~\eqref{PLOB_amp}.}
	\label{secret_g_nu10}
\end{figure}

\begin{figure}[t]
	\centering
	\includegraphics[width=1\linewidth]{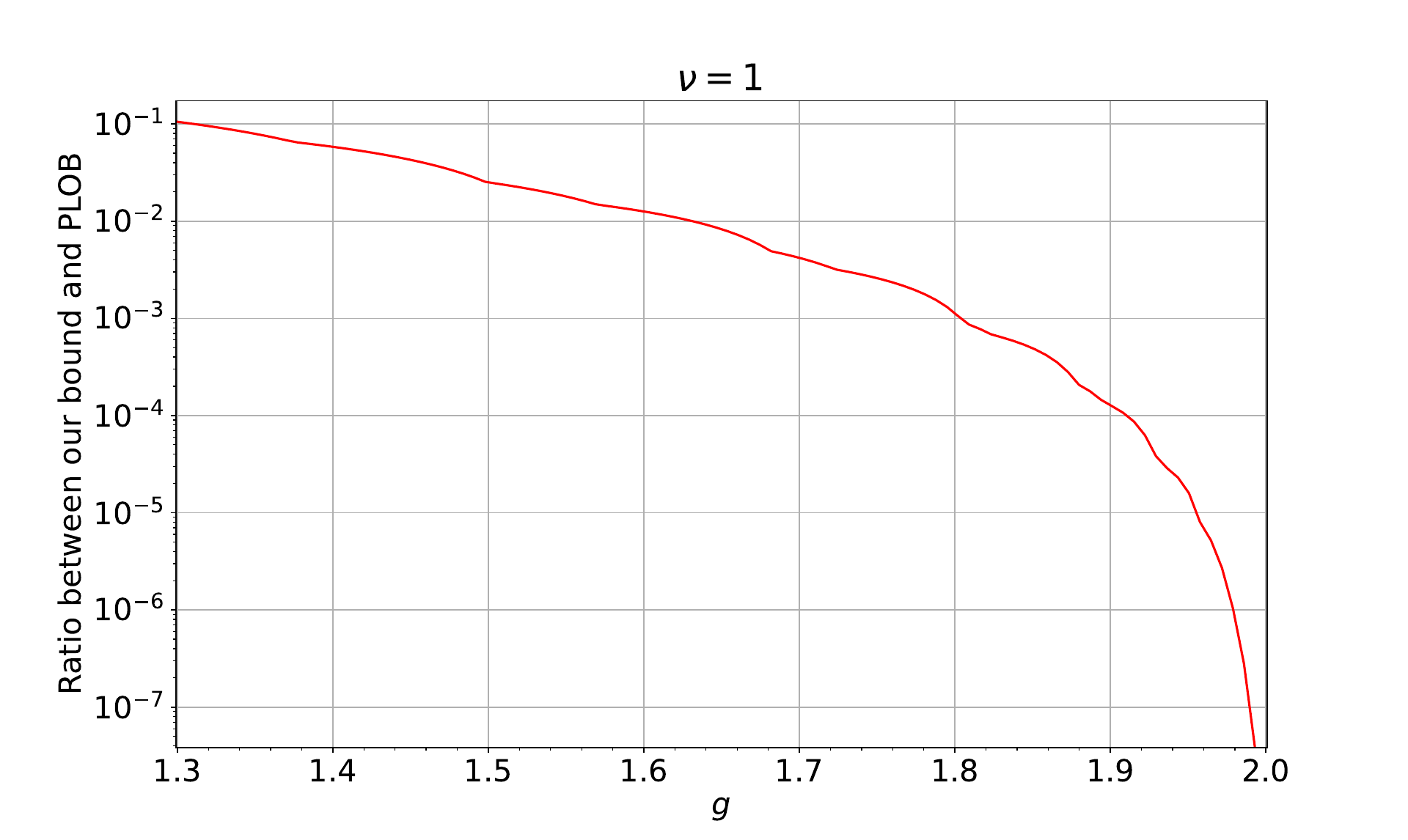}
	\caption{Ratio between our new lower bound on the two-way capacities of the thermal amplifier $\Phi_{g,\nu} $ in~\eqref{lowQ2_delta_amp} and the PLOB bound in~\eqref{PLOB_amp} as a function of $g$ for $\nu=1$.}
	\label{log_ratio_vs_lam_amp}
\end{figure}

\begin{figure}[t]
	\centering
	\includegraphics[width=1\linewidth]{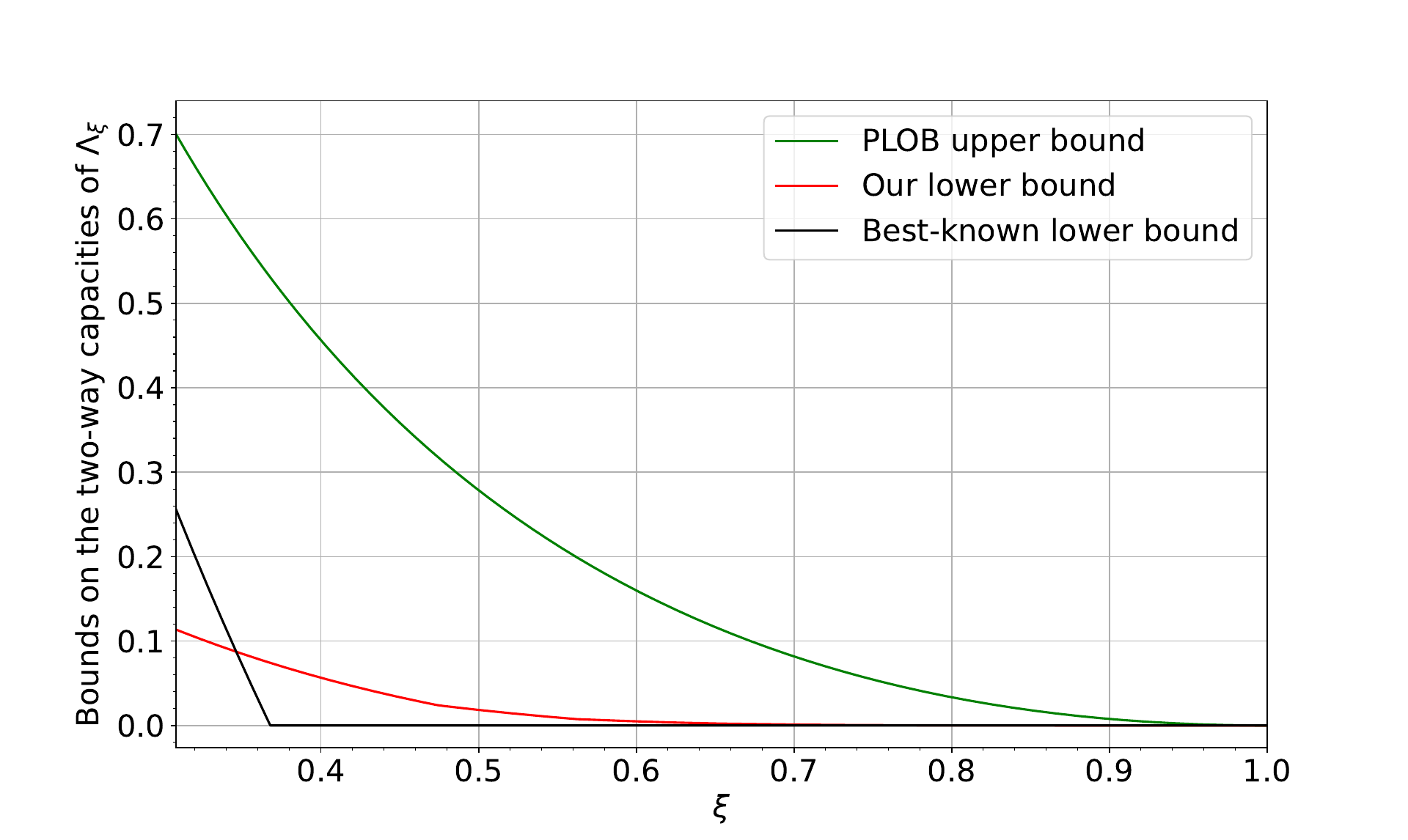}
	\caption{Bounds on the two-way quantum capacity $Q_2(\Lambda_\xi)$ and secret-key capacity $K(\Lambda_\xi)$ of the additive Gaussian noise plotted with respect to $\xi$. The red line is our new lower bound obtained by exploiting~\eqref{lowQ2_delta_add}. The black line is the best known lower bound on $Q_2(\Lambda_\xi)$, which is the coherent information lower bound reported in~\eqref{lowQ2_add}. The green line is the PLOB bound reported in~\eqref{PLOB_add}.}
	\label{add_vs_xi}
\end{figure}

\end{document}